\documentclass[11pt,a4paper]{article}

\usepackage[T1]{fontenc}  % norske tegn
\usepackage[latin1]{inputenc} % norske tegn
\usepackage{a4wide}
\usepackage{array}
\usepackage{amsmath}
\usepackage{amssymb}
\usepackage{amsthm}
\usepackage{mathtools}
\usepackage{paralist}
\usepackage{mdwlist}
\usepackage{algorithmic}
\usepackage{algorithm}
\usepackage[dvips]{graphicx}
\usepackage{url}
\usepackage{color}
\usepackage{cite}

\newcolumntype{C}[1]{>{\centering\arraybackslash$}m{#1}<{$}}
\newlength{\mycolwd}

\renewcommand{\int}{\mathop{\mathrm{int}}}
\theoremstyle{plain}
\newtheorem{theorem}{Theorem}[section]

\theoremstyle{definition}
\newtheorem{definition}{Definition}[section]

\begin{document}
\title{Extremal entanglement witnesses}

\renewcommand*{\thefootnote}{\fnsymbol{footnote}}
\setcounter{footnote}{1}

\author{
Leif Ove Hansen$^a$
\and Andreas Hauge$^a$\footnote{Deceased 31 May 2015}
\and Jan Myrheim$^a$
\and Per {\O}yvind Sollid$^b$\\
\\
%\affiliation{
$(a)$ Department of Physics,
Norwegian University of Science and Technology,\\
N--7491 Trondheim, Norway\\
$(b)$ Department of Physics, University of Oslo,\\
N--0316 Oslo, Norway
%}
}
\date{\today}

\maketitle

\setcounter{footnote}{0}

%%%%%%%%%%%%%%%%%%%%%%%%

%%%%%%%%%%%%%%%%%%%%%%%%%%%%%%%%%%%%%%%%%%%%%%%%%%%%%%%%%%%%%%%%%%%

%%%%%%%%%%%%%%%%%%%%%%%%%%%%%%%%%%%%%%%%%%%%%%%%%%%%%%%%%%%%%%%%%%%%%%%%%%%%%%%%%%%%%%%%%%%%%%%%%%%
%%%%%%%%%%%%%%%%%%%%%%%%%%%%%%%%%%%%%%%%%%%%%%%%%%%%%%%%%%%%%%%%%%%%%%%%%%%%%%%%%%%%%%%%%%%%%%%%%%%

\begin{abstract}
  We present a study of extremal entanglement witnesses on a bipartite
  composite quantum system.  We define the cone of witnesses as the
  dual of the set of separable density matrices, thus
  $\textrm{Tr}\,\Omega\rho\geq 0$ when $\Omega$ is a witness and
  $\rho$ is a pure product state, $\rho=\psi\psi^{\dagger}$ with
  $\psi=\phi\otimes\chi$.  The set of witnesses of unit trace is a
  compact convex set, uniquely defined by its extremal points.  The
  expectation value $f(\phi,\chi)=\mathrm{Tr}\,\Omega\rho$ as a
  function of vectors $\phi$ and $\chi$ is a positive semidefinite
  biquadratic form.  Every zero of $f(\phi,\chi)$ imposes strong
  real-linear constraints on $f$ and $\Omega$.  The real and symmetric
  Hessian matrix at the zero must be positive semidefinite.  Its
  eigenvectors with zero eigenvalue, if such exist, we call Hessian
  zeros.  A zero of $f(\phi,\chi)$ is quadratic if it has no Hessian
  zeros, otherwise it is quartic.  We call a witness quadratic if it
  has only quadratic zeros, and quartic if it has at least one quartic
  zero.  A main result we prove is that a witness is extremal if and
  only if no other witness has the same, or a larger, set of zeros and
  Hessian zeros.  A quadratic extremal witness has a minimum number of
  isolated zeros depending on dimensions.  If a witness is not
  extremal, then the constraints defined by its zeros and Hessian
  zeros determine all directions in which we may search for witnesses
  having more zeros or Hessian zeros.  A finite number of iterated
  searches in random directions, by numerical methods, lead to an
  extremal witness which is nearly always quadratic and has the
  minimum number of zeros.  We discuss briefly some topics related to
  extremal witnesses, in particular the relation between the facial
  structures of the dual sets of witnesses and separable states.  We
  discuss the relation between extremality and optimality of
  witnesses, and a conjecture of separability of the so called
  structural physical approximation (SPA) of an optimal witness.
  Finally, we discuss how to treat the entanglement witnesses on a
  complex Hilbert space as a subset of the witnesses on a real Hilbert
  space.
\end{abstract}

%\keywords{Entanglement witnesses; positive maps; convex sets.}

Keywords: Entanglement witnesses; positive maps; convex sets.

\vspace*{1\baselineskip}

\pagebreak

\tableofcontents

\section{Introduction}

Entanglement is the quintessence of nonclassicality in composite
quantum systems.  As a physical resource it finds many applications
especially in quantum information theory, see~\cite{hhhh,NC} and
references therein.  Entangled states are exactly those states that
can not be modelled within the classical paradigm of locality and
realism~\cite{epr1935,CHSH1969,Wer1989,Mas2008}.  Contrary to states
of classical composite systems, they allow better knowledge of the
system as a whole than of each component of the
system~\cite{schro1935,HorPR1994}.  The notion of entanglement is
therefore of interest within areas of both application and
interpretation of quantum mechanics.

The problem of distinguishing entangled states from nonentangled
(separable) states is a fundamental issue.  Pure separable states of a
bipartite system are exactly the pure product states, they are easily
recognized by singular value decomposition, also known as Schmidt
decomposition~\cite{NC}.  The situation is much more complicated for
mixed states.  A separable mixed state is a state mixed from pure
product states~\cite{Wer1989}.

The separability problem of distinguishing mixed separable states from
mixed entangled states has received attention along two lines.  On the
one hand, operational or computational tests for separability are
constructed~\cite{Peres1996}.  On the other hand, the geometry of the
problem is studied, with the aim of clarifying the structures defining
separability
~\cite{BZ,HorMPR96,Zyc,Gur2002,Sto1963,poslinmaps2013,Alfsen2010,pptI,ppt55}.
Two results, one from each approach, deserve special attention.  An
operationally simple necessary condition for separability was given by
Peres~\cite{Peres1996}.  He pointed out that the partial transpose of
a separable state is again a separable state, therefore we know that a
state is entangled if its partial transpose is not positive
semidefinite.  States with positive partial transpose, so called PPT
states, are interesting in their own right~\cite{HorMPR98}.  A
significant geometrical result is that of {\.Z}yczkowski~\textit{et al.} that
the volume of separable states is nonzero~\cite{Zyc}.  They prove the
existence of a ball of separable states surrounding the maximally
mixed state.  The maximal radius of this ball was provided by Gurvits
and Barnum~\cite{Gur2002}.

The well known necessary and sufficient separability condition in
terms of positive linear maps or entanglement witnesses connects these
two lines of thought~\cite{HorMPR96,Terhal2000}.  There exists a set
of linear maps on states, or equivalently a set of observables called
entanglement witnesses, that may detect, or witness, entanglement of a
state.  This result is the point of departure for almost all research
into the problem, be it computational, operational or geometrical.  The
major obstacle in the application of positive linear maps and
entanglement witnesses is that they are difficult to identify, in
principle as difficult as are the entangled states themselves.

Our contribution falls in the line of geometrical studies of
entanglement.  The ultimate goal is to classify the extremal points of
the convex set of entanglement witnesses, the building blocks of the
set, in a manner useful for purposes in quantum mechanics.  Positive
linear maps, and the extremal positive maps, were studied already in
1963 by St{\o}rmer in the context of partially ordered vector spaces
and $C^*$~algebras~\cite{Sto1963,poslinmaps2013}.  Since then the
extremal positive maps and extremal witnesses have only received
scattered
attention~\cite{Choi1975,ChoiLam1977,Robertson1985,Grab,Mar2008,Lew,Ha2012a,Chru3}.

Here we combine and apply basic notions from convex geometry and
optimization theory in order to study the extremal entanglement
witnesses.  We obtain computationally useful necessary and sufficient
conditions for extremality.  We use these ideas to analyse previously
known examples and construct new numerical examples of extremal
witnesses.  The extremal witnesses found in random searches are
generic, by definition.  It turns out that the previously known
examples of extremal witnesses are very far from being generic, and
the generic extremal witness is of a type not yet published in the
literature.

We discuss what we learn about the geometry of the set of witnesses,
and also the set of separable states, by studying extremal witnesses.
Other topics we comment on are the relation between extremal and
optimal witnesses~\cite{Lew,Lew2001} and the so called structural
physical approximation (SPA) of witnesses relating them to physical
maps of states~\cite{PHor2003,HorEk,Kor,Ha2012}.

Numerical methods play a central role in our work.  An important
reason is that the entanglement witnesses we want to study are so
complicated that there is little hope of treating them analytically.
Numerical work is useful for illustrating the theory, examining
questions of interest and guiding our thoughts in questions we cannot
answer rigorously.

The methods we use for studying extremal witnesses are essentially the
same as we have used previously for studying extremal states with
positive partial transpose, \textit{i.e.} PPT
states~\cite{ppt55,pptII,pptIII,pptIV,pptV}.  The work presented here
is also based in part on the master's thesis of one of us~\cite{Hauge2011}.

\subsection*{Outline of the article}

This article consists of three main parts.  Section~\ref{sec:intro} is
the first part, where we review material necessary for the
appreciation of later sections.  The basic concepts of convex geometry
are indispensable.  We review the geometry of entanglement witnesses,
and formulate the study of extremal witnesses as an optimization
problem where we represent the witness as a positive semidefinite
biquadratic form and study its zeros.

The second part consists of Section~\ref{sec:constraints} and
Section~\ref{sec:extremal}, where we develop a necessary and
sufficient condition for the extremality of a witness.  The basic idea
is that every zero of a positive biquadratic form representing a
witness imposes strong linear constraints on the form.  The
extremality condition in terms of the zeros of a witness and the
associated constraints is our most important result, not only because
it gives theoretical insight, but because it is directly useful for
numerical computations.

In the third, last, and most voluminous part, the
Sections~\ref{sec:decompwitnesses} to~\ref{sec:real}, we apply the
extremality condition from Section~\ref{sec:extremal}.  As a first
application we study decomposable witnesses in
Section~\ref{sec:decompwitnesses}.  These are well understood and not
very useful as witnesses, since they can only detect the entanglement
of a state in the trivial case when its partial transpose is not
positive.  But they serve to illustrate the concepts, and we find them
useful as stepping stones towards nondecomposable witnesses.

In Section~\ref{sec:knownexamples} we study examples of extremal
nondecomposable witnesses.  First we study two examples known from the
literature, and confirm their extremality by our numerical method.
Then we construct numerically random examples of extremal witnesses in
order to learn about their properties.  We find that the generic
extremal witnesses constructed numerically belong almost without
exception to a completely new class, not previously noticed in the
literature.  They have a fixed number of isolated quadratic zeros,
whereas the previously published extremal witnesses have continuous
sets of zeros.  The number of zeros depends on the dimensions of the
Hilbert spaces of the two subsystems.

In Section~\ref{sec:dshapedI} we show examples of a special D-shaped
type of faces of the set of witnesses in the lowest nontrivial
dimensions, $2\times 4$ and $3\times 3$.  These faces are ``next to
extremal'', bordered by extremal witnesses plus a straight section of
decomposable witnesses.

In Section~\ref{sec:separable} we study faces of the set of separable
states, using the duality between separable states and witnesses.  We
classify two different families of faces, generalizing results of
Alfsen and Shultz~\cite{Alfsen2010}, and present some statistics on
randomly generated maximal simplex faces.  We point out that the
facial structure is relevant for the question of how many pure product
states are needed in the convex decomposition of an arbitrary
separable state, and suggest that it may be possible to improve
substantially the trivial bound given by the dimension of the
set~\cite{HorP1997,Car}.

In Section~\ref{sec:optimal} we compare the notions of optimal and
extremal witnesses~\cite{Lew,Lew2001}.  In low dimensions we find that
almost every nondecomposable optimal witness is extremal, whereas an
abundance of nondecomposable and nonextremal optimal witnesses exist
in higher dimensions.

In Section~\ref{sec:spa} we comment on a separability conjecture
regarding structural physical approximations (SPAs) of optimal
witnesses~\cite{HorEk,Kor}.  This conjecture has since been
refuted~\cite{Ha2012,Stormer2013}.  We have tried to test numerically a
modified separability conjecture.  It holds within the numerical
precision of our separability test, but this test with the available
precision is clearly not a definitive proof.

In Section~\ref{sec:real} we point out that it is possible, and maybe
even natural from a certain point of view, to treat the entanglement
witnesses on a complex tensor product space as a subset of the
witnesses on a real tensor product space.

In Section~\ref{sec:conclusion} we summarize our work and suggest some
possible directions for future efforts.  Some details regarding
numerical methods are found in the Appendix.

\section{Background material}
\label{sec:intro}

In this section we review material necessary as background for the
following sections.  Convex geometry is the mathematical basis of the
theory of mixed quantum states, and is equally basic in the
entanglement theory for mixed states.  We introduce the concepts of
dual convex cones, entanglement witnesses, positive maps, and
biquadratic forms.  See \textit{e.g.} introductory sections of~\cite{BZ,BV}
for further details on convexity.  Concepts from
optimization theory are useful in the numerical treatment of a special
minimization problem, see~\cite{BV,NW}.

\subsection{Convexity}
\label{sec:convdual}

The basic concepts of convex geometry are useful or even essential for
describing mixed quantum states as probabilistic mixtures of pure
quantum states.  A convex subset of a real affine space is defined by
the property that any convex combination
\begin{equation}
\label{eq:convexcomb}
x=(1-p)x_1+px_2
\qquad
\mathrm{with}
\qquad
0\leq p\leq 1
\end{equation}
of members $x_1,x_2$ is a member.  If $x_1\neq x_2$ and $0<p<1$ we say
that $x$ is a proper convex combination of $x_1,x_2$, and $x$ is an
interior point of the line segment between $x_1$ and $x_2$.

The dimension of a convex set is the dimension of the smallest affine
space containing it.  We will be dealing here only with finite
dimensional sets.  Closed and bounded subsets of finite dimensional
Euclidean spaces are compact, according to the Heine--Borel theorem.
A compact convex set has extremal points that are not convex
combinations of other points.  It is completely described by its
extremal points, since any point in the set can be decomposed as a
convex combination involving no more than $n+1$ extremal points, where
$n$ is the dimension of the set~\cite{Car}.  It is called a polytope
if it has a finite number of extremal points, and a simplex if the
number of extremal points is one more than the dimension.

\begin{definition}
\label{def:facedef}
A bidirection at $x$ in the convex set $\mathcal{K}$ is a direction
vector $v\neq 0$ such that $x+tv\in\mathcal{K}$ for $t$ in some
interval ${[t_1,t_2]}$ with $t_1<0<t_2$.

A face $\mathcal{F}$ of $\mathcal{K}$ is a convex subset of
$\mathcal{K}$ with the property that if $x\in\mathcal{F}$ then every
bidirection in $\mathcal{K}$ at $x$ is a bidirection in $\mathcal{F}$
at $x$.

An equivalent condition defining $\mathcal{F}$ as a face is that if
$x\in\mathcal{F}$ is a proper convex combination of
$x_1,x_2\in\mathcal{K}$, then $x_1,x_2\in\mathcal{F}$.
\end{definition}

The empty set and $\mathcal{K}$ itself are faces of $\mathcal{K}$, by definition.  All
other faces are called proper faces.  The extremal points of $\mathcal{K}$ are
the zero dimensional faces of $\mathcal{K}$.  A face of dimension $n-1$ where
$n$ is the dimension of $\mathcal{K}$ is called a facet.

A point $x$ in a convex set $\mathcal{K}$ is an interior point of $\mathcal{K}$ if every
direction at $x$ (in the minimal affine space containing $\mathcal{K}$) is a
bidirection, otherwise $x$ is a boundary point of $\mathcal{K}$.  Every point
$x\in\mathcal{K}$ is either an extremal point of $\mathcal{K}$ or an interior point of a
unique face $\mathcal{F}_x$ of dimension one or higher.  This face is the
intersection of $\mathcal{K}$ with the affine space
\begin{equation}
{\mathcal A}_x=\{\,x+tv\mid t\in\mathbb{R}\,,\;v\in{\mathcal{B}}_x\,\}
\end{equation}
where ${\mathcal{B}}_x$ is the set of all bidirections in $\mathcal{K}$ at $x$.

Thus, every face $\mathcal{F}$ of $\mathcal{K}$ is an intersection
$\mathcal{F}=\mathcal{A}\cap\mathcal{K}$ of $\mathcal{K}$ with some affine space $\mathcal{A}$.
We take $\mathcal{A}$ to be a subspace of the minimal affine space
containing $\mathcal{K}$.  The minimum dimension of $\mathcal{A}$ is the
dimension of $\mathcal{F}$, but it may also be possible to choose $\mathcal{A}$
as an affine space of higher dimension than $\mathcal{F}$.  A proper face
$\mathcal{F}=\mathcal{A}\cap\mathcal{K}$ is said to be exposed if $\mathcal{A}$ has
dimension $n-1$ where $n$ is the dimension of $\mathcal{K}$.  Note that an
exposed face may have dimension less than $n-1$, thus it is not
necessarily a facet, and it may be just a point.  Every exposed point
is extremal, since it is a zero dimensional face, but a convex set may
have extremal points that are not exposed.

A proper face of $\mathcal{K}$ is part of the boundary of $\mathcal{K}$.  The faces of a
face $\mathcal{F}$ of $\mathcal{K}$ are the faces of $\mathcal{K}$ contained in $\mathcal{F}$.  In
particular, the following result is useful for understanding the face
structure of a convex set when its extremal points are known.
\begin{theorem}
\label{thm:extrpface}
\textit{Let $\mathcal{F}$ be a face of the convex set $\mathcal{K}$.  Then a point in $\mathcal{F}$ is an
extremal point of $\mathcal{F}$ if and only if it is an extremal point of $\mathcal{K}$.}
\end{theorem}

We state here another useful result which follows directly from the
definitions.
\begin{theorem}
\label{thm:convxints}
\textit{An intersection $\mathcal{K}=\mathcal{K}_1\cap\mathcal{K}_2$ of two convex sets $\mathcal{K}_1$ and $\mathcal{K}_2$
is again a convex set, and every face $\mathcal{F}$ of $\mathcal{K}$ is an intersection
$\mathcal{F}=\mathcal{F}_1\cap\mathcal{F}_2$ of faces $\mathcal{F}_1$ of $\mathcal{K}_1$ and $\mathcal{F}_2$ of $\mathcal{K}_2$.}
\end{theorem}

The facial structure of $\mathcal{K}=\mathcal{K}_1\cap\mathcal{K}_2$
follows from Definition~\ref{def:facedef} because the bidirections at
any point $x\in\mathcal{K}$ are the common bidirections at $x$ in
$\mathcal{K}_1$ and $\mathcal{K}_2$.

A cone $C$ in a real vector space is a set such that if $x\in C$,
$x\neq 0$, then $tx\in C$ but $-tx\not\in C$ for every $t>0$.
The concept of dual convex cones, to be defined below, has found a
central place in the theory of quantum entanglement.

\subsection{Quantum states}

We are concerned here with a bipartite quantum system with Hilbert
space $\mathcal{H}=\mathcal{H}_a\otimes\mathcal{H}_b$ of finite
dimension $N=N_aN_b$.  The real vector space $H$ of observables on
$\mathcal{H}$ has dimension $N^2\!$, a natural Euclidean inner product
$\langle A,B\rangle=\textrm{Tr}\, AB$ and the corresponding
Hilbert--Schmidt norm $\|A\|=\sqrt{\textrm{Tr}\,A^2}$.  We take
$\mathcal{H}_a=\mathbb{C}^{N_a}\!$, $\mathcal{H}_b=\mathbb{C}^{N_b}\!$, so
$H$ is the set of Hermitian matrices in the matrix algebra
$\mathrm{M}_{N}(\mathbb{C})$.  We write the components of a vector
$\psi\in\mathbb{C}^{N}$ as $\psi_I=\psi_{ij}$ with $I=1,2,\ldots,N$ or
$ij=11,12,\ldots,N_aN_b$.  For any $A\in H$ we define $A^P$ as the
partial transpose of $A$ with respect to the second subsystem, that
is,
\begin{equation}
(A^P)_{ij;kl}=A_{il;kj}\;.
\end{equation}

The set of positive semidefinite matrices in $H$ is a closed convex
cone which we denote by $\mathcal{D}$.  The state space of the quantum
system is the intersection $\mathcal{D}_1$ between $\mathcal{D}$ and
the hyperplane defined by $\textrm{Tr} A=1$.  The matrices in
$\mathcal{D}_1$ are called density matrices, or mixed quantum states.
A vector $\psi\in\mathcal{H}$ with $\psi^{\dagger}\psi=1$ defines a
pure state $\psi\psi^{\dagger}\in\mathcal{D}_1$.  $\mathcal{D}_1$ is a
compact convex set of dimension $N^2-1$.  The pure states are the
extremal points of $\mathcal{D}_1$, in fact this is an alternative
definition of $\mathcal{D}_1$.

A pure state $\psi\psi^{\dagger}$ is separable if $\psi$ is a product
vector, $\psi=\phi\otimes\chi$, thus a separable pure state is a tensor
product of pure states,
\begin{equation}
\psi\psi^{\dagger}
=(\phi\otimes\chi)(\phi\otimes\chi)^{\dagger}
=(\phi\phi^{\dagger})\otimes(\chi\chi^{\dagger})\;.
\end{equation}
The set of separable states $\mathcal{S}_1$ is the smallest convex
subset of $\mathcal{D}_1$ containing all the separable pure states.
Since the separable pure states are extremal points of $\mathcal{D}_1$
containing $\mathcal{S}_1$, they are also extremal points of
$\mathcal{S}_1$, and they are all the extremal points of
$\mathcal{S}_1$.  $\mathcal{S}_1$ is compact and defines a convex cone
$\mathcal{S}\subset\mathcal{D}$.  The dimension of $\mathcal{S}_1$ is
$N^2-1$, the same as the dimension of $\mathcal{D}_1$, hence every
separable state may be written as a convex combination of $N^2$ or
fewer pure product states~\cite{HorP1997,Car}.

The fact that $\mathcal{S}$ and $\mathcal{D}$ have the same dimension
$N^2$ is not quite as trivial as one might be tempted to think.  It is
a consequence of the fundamental fact that we use complex Hilbert
spaces in quantum mechanics, as the following argument shows.  In a
quantum mechanics based on real Hilbert spaces every separable state
would have to be symmetric under partial transposition, and the
dimension of $\mathcal{S}$ would be much smaller than the dimension of
$\mathcal{D}$.  With complex Hilbert spaces, a generic real symmetric
matrix which is a separable state is not symmetric under partial
transposition, and the remarkable conclusion is that its
representation as a convex combination of pure product states must
necessarily involve complex product vectors.  The basic reason that
$\mathcal{S}$ and $\mathcal{D}$ have the same dimension is the
relation $H=H_a\otimes H_b$ between the real vector spaces $H$, $H_a$
and $H_b$ of observables on $\mathcal{H}$, $\mathcal{H}_a$ and
$\mathcal{H}_b$.  We construct $\mathcal{S}$ explicitly as a subset of
$H_a\otimes H_b$, in such a way that the dimension of $\mathcal{S}$ is
the full dimension of $H_a\otimes H_b$.

We define the set of positive partial transpose states (PPT states) as
$\mathcal{P}_1=\mathcal{D}_1\cap\mathcal{D}_1^P\!$.  \,Partial transposition is an invertible linear
operation preserving the convex structure of $\mathcal{D}_1$, hence $\mathcal{P}_1$ is
also a compact convex set defining a convex cone $\mathcal{P}$.  We have that
$\mathcal{S}_1\subset\mathcal{P}_1\subset\mathcal{D}_1$.  The fact that every separable state has a
positive partial transpose is obvious, and provides a simple and
powerful test for separability, known as the Peres criterion
~\cite{Peres1996}.  In dimensions $2\times 2$ and $2\times 3$ the
converse statement is also true, that every PPT state is
separable~\cite{HorMPR96}.  The entanglement of PPT states is not
distillable into entanglement of pure states, and it is believed that
the entangled PPT states alone possess this special ``bound'' type of
entanglement~\cite{HorMPR98}.

\subsection{Dual cones, entanglement witnesses}
\label{sec:entwits}

The existence of entangled PPT states motivates the introduction of
entanglement witnesses that may reveal their entanglement.  We define
the dual cone of $\mathcal{S}$ as
\begin{equation}
\label{eq:witnessdef}
\mathcal{S}^{\circ}=\{\Omega\in H\mid\textrm{Tr}\,\Omega\rho\geq 0\,\,\,\forall\,\,\rho\in\mathcal{S}\}\,.
\end{equation}
The members of $\mathcal{S}^{\circ}$ will be called here entanglement witnesses.  The
usual convention is that a witness is required to have at least one
negative eigenvalue, but since our focus is on the geometry of $\mathcal{S}^{\circ}$
this restriction is not important to us here.

The definition of $\mathcal{S}^{\circ}$ as a dual cone implies that it is closed and
convex.  Since $\mathcal{S}$ is a closed convex cone the dual of $\mathcal{S}^{\circ}$ is $\mathcal{S}$
itself~\cite{HorMPR96,BV} (the fact that $S^{\circ\circ}=\mathcal{S}$ is given
as an exercise in~\cite{BV}).  Thus, a witness having a negative
expectation value in a state proves the state to be entangled, and
given any entangled state there exists a witness having a negative
expectation value in that state.  This two-way implication makes
entanglement witnesses powerful tools for detecting entanglement, both
theoretically and experimentally~\cite{Terhal2000,Lew,Brandao2005,Bre2006,Doh2004,Ioa2006,San2001,Toth2005}.

From the experimental point of view the testimony of an entanglement
witness is of a statistical nature.  A positive or negative result of
one single measurement of a witness gives little information about
whether the state in question is separable or entangled.  A positive
or negative average over many measurements will give a more reliable
answer, but some statistical uncertainty must always remain.

Since the extremal points of $\mathcal{S}_1$ are the pure product states, a
matrix $\Omega\in H$ is a witness if and only if its expectation value
is nonnegative in every pure product state.  Furthermore, since the
dimension of $\mathcal{S}$ is the full dimension of $H$, every witness
$\Omega\neq0$ has strictly positive expectation values in some pure
product states.  Since every product vector is a member of some basis
of orthogonal product vectors, we conclude that every witness
$\Omega\neq0$ has $\textrm{Tr}\,\Omega>0$~\cite{Zyc}.  Accordingly the set
$\mathcal{S}_1^{\circ}$ of normalized (unit trace) witnesses completely describes all
of $\mathcal{S}^{\circ}$.
\begin{theorem}
\label{thm:sudbounded}
\textit{$\mathcal{S}_1^{\circ}$ is a bounded set.}
\end{theorem}
\begin{proof}
The maximally mixed state $\rho_0=I/N$ is in $\mathcal{S}_1$.  Define for
$\theta>0$ and $\Gamma\in H$ with $\textrm{Tr}\,\Gamma=0$, $\textrm{Tr}\,\Gamma^2=1$,
\begin{equation}
f(\theta,\Gamma)=\min_{\rho\in\mathcal{S}_1}\textrm{Tr}\,\rho\left(\rho_0+\theta\Gamma\right)\\
=\frac{1}{N}+\theta\min_{\rho\in\mathcal{S}_1}\textrm{Tr}\,\rho\Gamma\,.
\end{equation}
Since $\mathcal{S}_1$ is compact $f$ is well defined.  Since $\textrm{Tr}\,\Gamma=0$, we
have $\Gamma\notin\mathcal{S}^{\circ}$ and the minimum of $\textrm{Tr}\,\rho\Gamma$ is strictly
negative.  Therefore, given any $\Gamma$ we can always find the
$\theta(\Gamma)$ which makes $f(\theta,\Gamma)=0$.  Since $\Gamma$
lies in a compact set there exists a $\Gamma^*$ with
$\theta^*=\theta(\Gamma^*)=\max_\Gamma\theta(\Gamma)$.  This
$\theta^*$ defines an $N^2-1$ dimensional Euclidean ball
$\mathcal{B}_1$ centered at $\rho_0$ and containing $\mathcal{S}_1^{\circ}$.
\end{proof} 
Note that the dual set $\mathcal{B}_1^\circ$ is an $N^2-1$
dimensional ball contained in $\mathcal{S}_1$, also centered at $\rho_0$~\cite{Zyc}.

From now on, when we talk about entanglement witnesses we will usually
assume that they are normalized and lie in $\mathcal{S}_1^{\circ}$.  Since $\mathcal{S}_1^{\circ}$ is a
compact convex set it is completely described by its extremal points.
The ultimate objective of the present work is to characterize these
extremal witnesses.

The concept of dual cones applies of course also to the cone $\mathcal{D}$ of
positive semidefinite matrices and the cone $\mathcal{P}$ of PPT matrices.  The
cone $\mathcal{D}$ is self-dual, $\mathcal{D}^{\circ}=\mathcal{D}$.  Since $\mathcal{P}=\mathcal{D}\cap\mathcal{D}^P\!$, the dual
cone $\mathcal{P}^{\circ}$ is the convex hull of $\mathcal{D}\cup\mathcal{D}^P$.  Hence, the extremal
points of $\mathcal{P}_1d$ are the pure states $\psi\psi^{\dagger}$ and the
partially transposed pure states $(\psi\psi^{\dagger})^P\!$.  These are
also extremal in $\mathcal{S}_1^{\circ}$~\cite{Sto1963,poslinmaps2013,Grab,Mar2008}.  A witness
$\Omega\in\mathcal{P}^{\circ}$ is called decomposable, because it has the form
\begin{equation}
\label{eq:decomp}
\Omega=\rho+\sigma^P\quad \textrm{with}\quad\rho,\sigma\in\mathcal{D}\,.
\end{equation}
This terminology comes from the mathematical theory of positive
maps~\cite{Sto1963}.  Decomposable witnesses are in a sense trivial, and
not very useful as witnesses, since they can not be used for detecting
entangled PPT states.

Altogether, we have the following sequence of compact convex sets,
\begin{equation}
\mathcal{S}_1\subset\mathcal{P}_1\subset\mathcal{D}_1=\mathcal{D}_1^{\circ}\subset\mathcal{P}_1^{\circ}\subset\mathcal{S}_1^{\circ}\,.
\end{equation}
All these sets, with the exception of $\mathcal{D}_1$, are invariant under
partial transposition.  The sets $\mathcal{S}_1$, $\mathcal{D}_1$, and $\mathcal{P}_1^{\circ}$ are very
simply described in terms of their extremal points.  The extremal
points of $\mathcal{P}_1$ are not fully understood, though some progress has
been made in this direction~\cite{ppt55,pptII,pptIII,pptIV,pptV}.  The
extremal points of $\mathcal{S}_1^{\circ}$ are what we investigate here, they include
the extremal points of $\mathcal{S}_1$, $\mathcal{D}_1$, and $\mathcal{P}_1^{\circ}$.

A face of $\mathcal{D}_1$ is a complete set of density matrices on some subspace
of $\mathcal{H}$~\cite{BZ}.  Faces of $\mathcal{S}_1$ have recently been studied by
Alfsen and Shultz~\cite{Alfsen2010,Alfsen2012}.  We also comment on
faces of $\mathcal{S}_1$ in Section~\ref{sec:separable}.  As we develop our
results regarding extremal witnesses we will simultaneously obtain a
classification of faces of $\mathcal{S}_1^{\circ}$.

\subsection{Positive maps}
\label{sec:posmaps}

The study of the set of separable states through the dual set of
entanglement witnesses was started by Micha{\l}, Pawe{\l}, and Ryszard
Horodecki~\cite{HorMPR96}.  They pointed out the relation between
witnesses and positive maps, and used known results from the
mathematical theory of positive maps to throw light on the
separability problem~\cite{Sto1963}.  In particular, the fundamental
result that there exist entangled PPT states is equivalent to the
existence of nondecomposable positive maps.

We describe here briefly how entanglement witnesses are related to
positive maps, but do not aim at developing the perspectives of
positive maps in great detail.  In Section~\ref{sec:spa} we return
briefly to the use of positive maps for detecting entanglement.

We use a matrix $A\in H$ to define a real linear map
$\textbf{L}_A: H_a\rightarrow H_b$ such that $Y=\mathbf{L}_A X$ when
\begin{equation}
\label{eq:jamiso01}
Y_{jl}=\sum_{i,k}A_{ij;kl}X_{ki}\,.
\end{equation}
The correspondence $A\leftrightarrow\mathbf{L}_A$ is a vector space
isomorphism between $H$ and the space of real linear maps $H_a\rightarrow
H_b$.  A slightly different version of this isomorphism is more
common in the literature, this is the Jamio{\l}kowski isomorphism
$A\leftrightarrow\mathbf{J}_A$ by which $\mathbf{J}_AX=\mathbf{L}_A(X^T)$~\cite{Jam1972}.

The transposed real linear map
$\mathbf{L}_A^T: H_b\rightarrow H_a$ is defined such that $X=\mathbf{L}_A^T Y$ when
\begin{equation}
\label{eq:jamiso02}
X_{ik}=\sum_{j,l}A_{ij;kl}Y_{lj}\,.
\end{equation}
It is the transpose of $\mathbf{L}_A$ with respect to the natural scalar
products $\langle U,V\rangle=\textrm{Tr}\,UV$ in $H_a$ and $H_b$.  In fact,
for any $X\in H_a$, $Y\in H_b$ we have
\begin{equation}
\langle\mathbf{L}_AX,Y\rangle=
\textrm{Tr}\,((\mathbf{L}_AX)Y)=
\sum_{i,j,k,l}A_{ij;kl}X_{ki}Y_{lj}=
\textrm{Tr}\,(X(\mathbf{L}_A^T Y)),
\end{equation}
and therefore
\begin{equation}
\langle\mathbf{L}_AX,Y\rangle=\langle X,\mathbf{L}_A^T Y\rangle.
\end{equation}

The maps $\mathbf{L}_A$ and $\mathbf{L}_A^T$ act on one dimensional projection
operators $\phi\phi^{\dagger}\in H_a$ and $\chi\chi^{\dagger}\in H_b$ according to
\begin{eqnarray}
\label{eq:jamiso}
\nonumber
\mathbf{L}_A(\phi\phi^{\dagger})&=(\phi\otimes I_b)^\dagger A (\phi\otimes I_b)\,,
\\
\mathbf{L}_A^T(\chi\chi^{\dagger})&=(I_a\otimes\chi)^\dagger A (I_a\otimes\chi)\,.
\end{eqnarray}
Note that $\phi\otimes I_b$ is an $N\times N_b$ matrix such that
$(\phi\otimes I_b)\chi=\phi\otimes\chi$, whereas $I_a\otimes\chi$ is an $N\times
N_a$ matrix such that $(I_a\otimes\chi)\phi=\phi\otimes\chi$.  It follows that
\begin{equation}
\label{eq:jamisoI}
\chi^{\dagger}\,(\mathbf{L}_A(\phi\phi^{\dagger}))\,\chi=
\phi^{\dagger}\,(\mathbf{L}_A^T(\chi\chi^{\dagger}))\,\phi=
(\phi\otimes\chi)^\dagger A (\phi\otimes\chi)\,.
\end{equation}
If $\Omega$ is an entanglement witness then
equation~\eqref{eq:jamisoI} implies that $\mathbf{L}_{\Omega}$ is a positive
linear map $H_a\rightarrow H_b$, mapping positive semidefinite matrices in
$H_a$ to positive semidefinite matrices in $H_b$.  Similarly,
$\mathbf{L}_{\Omega}^T$ is a positive linear map $H_b\rightarrow H_a$.  The
correspondence $\Omega\leftrightarrow\mathbf{L}_{\Omega}$ is a vector space
isomorphism between the set of entanglement witnesses and the set of
positive maps.

A positive map $\mathbf{M}$ is said to be completely positive if every map
$\mathbf{I}\otimes\mathbf{M}$, where $\mathbf{I}$ is the identity map in an arbitrary
dimension, is positive.  It is easily shown that $\mathbf{J}_A$ is completely
positive if and only if $A$ is a positive matrix.  Thus $\mathbf{L}_A$ is
completely positive if and only if $A^P$ is a positive matrix.

\subsection{Biquadratic forms and optimization}
\label{sec:biquformopt}

The expectation value of an observable $A\in H$ in a pure product
state $\psi\psi^{\dagger}$ with $\psi=\phi\otimes\chi$ is a biquadratic form
\begin{equation}
f_A(\phi,\chi)
=(\phi\otimes\chi)^\dagger A(\phi\otimes\chi)\,.
\end{equation}
The condition for $\Omega\in H$ to be an entanglement witness is that
the biquadratic form $f_{\Omega}$ is positive semidefinite.  Our
approach here is to study witnesses through the associated biquadratic
forms.

In the present work we study especially boundary witnesses,
corresponding to biquadratic forms that are marginally positive.  This
leads naturally to an alternative definition of entanglement witnesses
in terms of an optimization problem.
\begin{definition}
\label{def:witndef}
A matrix $A\in H$ is an entanglement witness if and only if the
minimum value $p^*$ of the problem
\begin{align}
\label{eq:kernelproblem}
\nonumber
\textrm{minimize} & \qquad f_A(\phi,\chi)=(\phi\otimes\chi)^\dagger A(\phi\otimes\chi)\,,
\quad\phi\in\mathcal{H}_a\,,\,\,\,\,\chi\in\mathcal{H}_b\,,\\
\textrm{subject to} & \qquad \|\phi\|=\|\chi\|=1\,,
\end{align}
is nonnegative. 
\end{definition}
The normalization of $\phi$ and $\chi$ ensures that $p^*>-\infty$
when $A\notin\mathcal{S}^{\circ}$, and that $p^*>0$ for every $A$ in the interior of
$\mathcal{S}^{\circ}\!$.  It is natural to use here the Hilbert--Schmidt norm introduced
above, but in principle any norm would serve the same purpose.

The boundary $\partial\mathcal{S}^{\circ}$ of $\mathcal{S}^{\circ}$ consists of those $\Omega\in\mathcal{S}^{\circ}$
for which there exists a separable state $\rho$ orthogonal to
$\Omega$, \textit{i.e.} with $\textrm{Tr}\,\Omega\rho=0$.  In terms of problem
\eqref{eq:kernelproblem}, $A\in\partial\mathcal{S}^{\circ}$ if and only if $p^*=0$.
Since $\mathcal{S}^{\circ\circ}=\mathcal{S}$ we can similarly understand the boundary of
$\mathcal{S}$ as the set of separable states orthogonal to some witness.  We
will return to this duality between the boundaries of $\mathcal{S}$ and $\mathcal{S}^{\circ}$
in Section~\ref{sec:separable}.

\subsection{Equivalence under $\textrm{SL}\otimes\textrm{SL}$ transformations}

It is very useful to observe that all the main concepts discussed in
the present article are invariant under what we call
$\textrm{SL}\otimes\textrm{SL}$ transformations, in which a matrix $A$ is
transformed into $VAV^{\dagger}$ with an invertible product matrix
$V=V_a\otimes V_b$.

For example, since such a transformation is linear in $A$, it
preserves convex combinations, extremal points, and in general all the
convexity properties of different sets.  It preserves the positivity
of matrices, the tensor product structure of vectors and matrices, the
number of zeros of witnesses, and in general all properties related to
entanglement, except that it may increase or decrease entanglement as
measured quantitatively if either $V_a$ or $V_b$ is not unitary.

Thus, for our purposes it is useful to consider two density matrices
or two entanglement witnesses to be equivalent if they are related by
an $\textrm{SL}\otimes\textrm{SL}$ transformation.  This sorting into
equivalence classes helps to reduce the problem of understanding and
classifying entangled states and entanglement witnesses.

%%%%%%%%%%%%%%%%%%%%%%%%%%%%%%%%%%%%%%%%%%%%%%%%%%%%%%%%%%%%%%%%%%%%%%%%%%%%%%%%%%%%%%%%%%%%%%%%%%%%%%%
%%%%%%%%%%%%%%%%%%%%%%%%%%%%%%%%%%%%%%%%%%%%%%%%%%%%%%%%%%%%%%%%%%%%%%%%%%%%%%%%%%%%%%%%%%%%%%%%%%%%%%%

\section{Secondary constraints at zeros of witnesses}
\label{sec:constraints}

By definition, a witness $\Omega$ satisfies the following infinite set
of inequalities, all linear in $\Omega$,
\begin{equation}
\label{eq:priminv}
f_{\Omega}(\phi,\chi)\geq 0
\qquad\mathrm{with}\qquad\phi\in\mathcal{H}_a\,,\quad\chi\in\mathcal{H}_b\,,
\quad\|\phi\|=\|\chi\|=1\,.
\end{equation}
These constraints on $\Omega$, represented here as a biquadratic form
$f_{\Omega}$, are the primary constraints defining the set $\mathcal{S}_1^{\circ}$,
apart from the trivial linear constraint $\textrm{Tr}\,\Omega=1$.

If $\Omega$ is situated on the boundary of $\mathcal{S}_1^{\circ}$ it means that at
least one of these primary inequalities is an equality.  We will call
the pair $(\phi_0,\chi_0)$ a zero of $\Omega$ if
$f_{\Omega}(\phi_0,\chi_0)=0$.  We count $(a\phi_0,b\chi_0)$ with
$a,b\in\mathbb{C}$ as the same zero.  The primary
constraints~(\ref{eq:priminv}) with $(\phi,\chi)$ close to the zero
$(\phi_0,\chi_0)$ lead to rather stringent constraints on $\Omega$,
which we introduce as secondary constraints to be imposed at the zero.
These secondary constraints are both equalities and inequalities, and
they are linear in $\Omega$, like the primary constraints from which
they are derived.  They are summarized in explicit form in
Appendix~\ref{sec:explicit}.  In the next section we apply all the
secondary constraints to the problem of constructing and classifying
extremal witnesses.

The analysis of secondary constraints presented in this section is
essentially the same as carried out by Lewenstein \textit{et al.}
in their study of optimal witnesses~\cite{Lew}.

\subsection{Positivity constraints on polynomials}
\label{sec:polynomeq}

A model example may illustrate how we treat constraints.  Let $f(t)$
be a real polynomial in one real variable $t$, of degree four and with
a strictly positive quartic term, satisfying the primary constraints
$f(t)\geq 0$ for all $t$.  These are constraints on the coefficients
of the polynomial.  The equation $f(t)=0$ can have zero, one or two
real roots for $t$.  Assume that $f(t_0)=0$.  Because this is a
minimum, we must have $f'(t_0)=0$ and $f''(t_0)\geq 0$.  In the
limiting case $f''(t_0)=0$ we must have also $f^{(3)}(t_0)=0$.

Thus, if $f(t_0)=0$ we get secondary constraints $f'(t_0)=0$, and
either $f''(t_0)>0$ or $f''(t_0)=0$, $f^{(3)}(t_0)=0$.  If there is a
second zero $t_1$, similar secondary constraints must hold there.  It
should be clear that we may replace the infinite set of primary
constraints $f(t)\geq 0$ for every $t$ by the finite set of secondary
constraints at the zeros $t_0$ and $t_1$.

Because the zeros of a witness $\Omega$ are roots of a polynomial
equation in several variables, they have the following property.
\begin{theorem}
\label{thm:natureofzeros}
\textit{The set of zeros of a witness consists of at most a finite number of
components, where each component is either an isolated point or a
continuous connected surface.}
\end{theorem}
\begin{proof}
Choose $(\phi_2,\chi_2)$ such that $f_{\Omega}(\phi_2,\chi_2)>0$,
and define $f(t)=f_{\Omega}(\phi_1+t\phi_2,\chi_1+t\chi_2)$.  For
given $(\phi_1,\chi_1)$ and $(\phi_2,\chi_2)$ this is a nonnegative
polynomial in the real variable $t$ of degree four, hence it has
zero, one or two real roots for $t$.  By varying $(\phi_1,\chi_1)$
we reach all the zeros of $\Omega$.  The zeros of $f(t)$ will move
continuously when we vary $(\phi_1,\chi_1)$, except that they may
appear or disappear.  This construction should result in a set of
zeros as described in the theorem.
\end{proof}
Now let $(\phi_0,\chi_0)$ be a zero of $\Omega$, with
$\|\phi_0\|=\|\chi_0\|=1$.  Since the constraints~(\ref{eq:priminv})
on the polynomial $f_{\Omega}$ at
$(\phi,\chi)=(\phi_0+\xi,\chi_0+\zeta)$ are actually independent of
the normalization conditions $\|\phi\|=\|\chi\|=1$, we choose to
abandon these nonlinear normalization conditions (nonlinear in the
Hilbert--Schmidt norm) and replace them by the linear constraints
$\phi_0^{\dagger}\xi=0$, $\chi_0^{\dagger}\zeta=0$.  Strictly
speaking, even these orthogonality conditions are not essential, the
important point is that we vary $\phi$ and $\chi$ in directions away
from $\phi_0$ and $\chi_0$.

We find that some of the secondary constraints on $f_{\Omega}$ are
intrinsically real equations rather than complex equations, therefore
we introduce real variables $x\in\mathbb{R}^{2N_a-2}\!$, $y\in\mathbb{R}^{2N_b-2}$ and
write $\xi=J_0x$, $\zeta=K_0y$ with $\phi_0^{\dagger}J_0=0$,
$\chi_0^{\dagger}K_0=0$.  Our biquadratic form is then a real
inhomogeneous polynomial quadratic in $x$ and quadratic in $y$,
\begin{equation}
\label{eq:realparam}
f(x,y)=f_{\Omega}(\phi,\chi)
=((\phi_0+\xi)\otimes(\chi_0+\zeta))^{\dagger}
\,\Omega\,
((\phi_0+\xi)\otimes(\chi_0+\zeta))\,.
\end{equation}
The linear term of the polynomial is
\begin{align}
\nonumber
f_1(x,y) & = 2\,\mathrm{Re}\,\,(
(\xi\otimes\chi_0)^\dagger\Omega(\phi_0\otimes\chi_0)
+(\phi_0\otimes\zeta)^\dagger\Omega(\phi_0\otimes\chi_0))
\\
&=x^T\!\mathrm{D}_xf(0,0)+y^T\!\mathrm{D}_yf(0,0)\,,
\end{align}
in terms of the gradient
\begin{align}
\label{eq:grad}
\nonumber
\mathrm{D}_xf(0,0)
&=2\,\mathrm{Re}\,\,(J_0\otimes\chi_0)^\dagger\Omega(\phi_0\otimes\chi_0)\,,\\
\mathrm{D}_yf(0,0) 
&=2\,\mathrm{Re}\,\,(\phi_0\otimes K_0)^\dagger\Omega(\phi_0\otimes\chi_0)\,.
\end{align}
The quadratic term is
\begin{align}
\label{eq:hess0}
\nonumber
f_2(x,y)
& =
(\xi\otimes\chi_0)^{\dagger}\,\Omega\,(\xi\otimes\chi_0)
+(\phi_0\otimes\zeta)^{\dagger}\,\Omega\,(\phi_0\otimes\zeta)
\\
\nonumber
&\quad+
2\,\mathrm{Re}\,((\phi_0\otimes\zeta)^{\dagger}\,\Omega\,(\xi\otimes\chi_0)
+(\phi_0\otimes\chi_0)^{\dagger}\,\Omega\,(\xi\otimes\zeta))
\\
& = z^T G_{\Omega}\,z
\,,
\end{align}
where $z^T=(x^T,y^T)$, and $2G_{\Omega}=\mathrm{D}^2f(0,0)$ is the second
derivative, or Hessian, matrix, which is real and symmetric,
\begin{align}
\label{eq:hess}
\nonumber
G_{\Omega}&=\mathrm{Re}\,\,
\begin{bmatrix}
g_{xx} & g_{yx}^T \\
g_{yx} & g_{yy}
\end{bmatrix},\\
\nonumber
g_{xx} & =(J_0\otimes\chi_0)^\dagger\Omega(J_0\otimes\chi_0)\,,\\
\nonumber
g_{yy} & =(\phi_0\otimes K_0)^\dagger\Omega(\phi_0\otimes K_0)\,,\\
g_{yx} & =(\phi_0\otimes K_0)^\dagger\Omega(J_0\otimes\chi_0)
+(\phi_0\otimes K_0^*)^\dagger\Omega^P(J_0\otimes\chi_0^*)\,.
\end{align}
The cubic term is like the linear term but with
$\phi_0\leftrightarrow\xi$ and $\chi_0\leftrightarrow\zeta$,
\begin{equation}
f_3(x,y)=2\,\mathrm{Re}\,\,((\phi_0\otimes\zeta)^\dagger\Omega(\xi\otimes\zeta)
+(\xi\otimes\chi_0)^\dagger\Omega(\xi\otimes\zeta))\,.
\end{equation}
The quartic term is simply
\begin{equation}
f_4(x,y)=f_\Omega(\xi,\zeta)=(\xi\otimes\zeta)^\dagger\Omega(\xi\otimes\zeta)\,.
\end{equation}

The fact that the constant term of the polynomial vanishes,
$f(0,0)=0$, is one real linear constraint on $\Omega$,
\begin{equation}
\mathbf{T}_0:H\rightarrow\mathbb{R}\,,\quad \mathbf{T}_0\Omega=
(\phi_0\otimes\chi_0)^\dagger\Omega(\phi_0\otimes\chi_0)=0\,.
\end{equation}
Because $(x,y)=(0,0)$ is a minimum of the polynomial
the linear term  must also vanish identically,
and this results in another linear system of constraints,
\begin{equation}
\mathbf{T}_1:H\rightarrow\mathbb{R}^{2(N_a+N_b-2)},\quad
\mathbf{T}_1\Omega=
\begin{bmatrix}
\mathrm{D}_xf(0,0)\\
\mathrm{D}_yf(0,0)
\end{bmatrix}=0\,.
\end{equation}
Note that the equality constraints $\mathbf{T}_0$ and $\mathbf{T}_1$ are the same
for every witness with the zero $(\phi_0,\chi_0)$, they are uniquely
defined by the zero alone.  The total number of constraints in $\mathbf{T}_0$
and $\mathbf{T}_1$ is
\begin{equation}
M_2=2(N_a+N_b)-3\,.
\end{equation}
All of these $M_2$ constraints are linearly independent.  They are
given an explicitly real representation in
Appendix~\ref{sec:explicit}.

It is worth noting that the vanishing of the constant and linear terms
of the polynomial $f(x,y)$ is equivalent to the conditions that
\begin{align}
\nonumber
(\phi\otimes\chi_0)^\dagger\Omega(\phi_0\otimes\chi_0)
&=0\quad\forall\;\phi\in\mathcal{H}_a\,,\\
(\phi_0\otimes\chi)^\dagger\Omega(\phi_0\otimes\chi_0)
&=0\quad\forall\;\chi\in\mathcal{H}_b\,.
\end{align}
Hence the $2(N_a+N_a)-3$ real constraints $\mathbf{T}_0\Omega=0$ and
$\mathbf{T}_1\Omega=0$ may be expressed more simply as the following set of
$N_a+N_b$ complex constraints, equivalent to $2(N_a+N_b)$ real
constraints that are then not completely independent,
\begin{align}
\label{eq:gradCLL}
\nonumber
\mathbf{L}_\Omega^T(\chi_0\chi_0^\dagger)\,\phi_0
&=(I_a\otimes\chi_0)^\dagger\Omega(\phi_0\otimes\chi_0)
=0\,,\\
\mathbf{L}_\Omega(\phi_0\phi_0^\dagger)\,\chi_0
&=(\phi_0\otimes I_b)^\dagger\Omega(\phi_0\otimes\chi_0)
=0\,.
\end{align}

The quadratic term of the polynomial has to be nonnegative, again
because $(x,y)=(0,0)$ is a minimum.  The inequalities
\begin{equation}
\label{eq:f2constr}
z^T G_{\Omega}\,z\geq 0\qquad\forall\;z\in\mathbb{R}^{2(N_a+N_b-2)}
\end{equation}
are secondary inequality constraints, linear in $\Omega$, equivalent
to the nonlinear constraints that all the eigenvalues of the Hessian
matrix $G_{\Omega}$ must be nonnegative.

There are now two alternatives.  If the
inequalities~(\ref{eq:f2constr}) hold with strict inequality for all
$z\neq 0$, in other words, if all the eigenvalues of the Hessian are
strictly positive, then we call $(\phi_0,\chi_0)$ a quadratic
zero.  In this case $\mathbf{T}_0$ and $\mathbf{T}_1$ are the only equality
constraints placed on $\Omega$ by the existence of the zero
$(\phi_0,\chi_0)$.  We can immediately state the following important
result.
\begin{theorem}
\label{thm:quadzeros}
\textit{A quadratic zero is isolated: there is a finite distance to the next
zero.  Hence, a witness can have at most a finite number of
quadratic zeros.}
\end{theorem}

\subsection{Hessian zeros}
\label{sec:Hessianzeros}

The second alternative is that the Hessian has $K$ zero eigenvalues
with $K\geq 1$.  We will call $z\neq 0$ a Hessian zero at the zero
$(\phi_0,\chi_0)$ if $G_{\Omega}z=0$.  It is then a real linear
combination
\begin{equation}
z=\sum_{i=1}^Ka_iz_i
\end{equation}
of basis vectors $z_i\in\textrm{Ker}\,G_{\Omega}$.  The $K$ linearly
independent eigenvectors $z_i$ define the following system of linear
constraints on $\Omega$,
\begin{equation}
\label{eq:T2icnstr}
\mathbf{T}_2:H\rightarrow\mathbb{R}^{2K(N_a+N_b-2)}\,,\quad
\left(\mathbf{T}_2\Omega\right)_i=
G_{\Omega}z_i=0\,,\quad i=1,\ldots,K\,.
\end{equation}
These constraints ensure that $G_{\Omega}z=0$ and hence
$f_2(x,y)=z^TG_{\Omega}z=0$.  The vanishing of the quadratic term
of the polynomial implies that the cubic term must vanish as well.  We
therefore call $(\phi_0,\chi_0)$ a quartic zero, and the direction $z$
at $(\phi_0,\chi_0)$ a quartic direction.  The vanishing cubic term is
\begin{equation}
f_3(x,y)=2\sum_{l,m,n}a_la_ma_n
\mathrm{Re}\,(
(\phi_0\otimes\zeta_l)^\dagger\Omega(\xi_m\otimes\zeta_n)
+(\xi_l\otimes\chi_0)^\dagger\Omega(\xi_m\otimes\zeta_n))=0,
\end{equation}
where $\xi_i=J_0x_i$ and $\zeta_i=K_0y_i$.  Since the
product $a_la_ma_n$ is completely symmetric in the indices $l,m,n$ we
should symmetrize the expression it multiplies.  All the different
properly symmetrized expressions must vanish.  In this way we obtain a
linear system
\begin{equation}
\left(\mathbf{T}_3\Omega\right)_{lmn}=0
\qquad\mathrm{with}\qquad
1\leq l\leq m\leq n\leq K\,.
\end{equation}
The number of equations is the binomial coefficient $\tbinom{K+2}{3}$,
but it is likely that these constraints on $\Omega$ will in general
not be independent.

The total number of constraints in $\mathbf{T}_2$ and $\mathbf{T}_3$ is therefore
\begin{equation}
M_4(K)=2K(N_a+N_b-2)+\binom{K+2}{3}\,.
\end{equation}
Note that $M_4(1)=M_2=2(N_a+N_b)-3$.  Since the constraints in
$\mathbf{T}_2$ and $\mathbf{T}_3$ address other terms in the polynomial than those in
$\mathbf{T}_0$ and $\mathbf{T}_1$, the two sets of constraints should be independent
from each other.  Furthermore, in the case $K=1$, all $M_4(1)$
equations in the total system $\mathbf{T}_2,\mathbf{T}_3$ should be linearly
independent.  In the case of $K>1$ we expect overlapping constraints.

\subsection{Summary of constraints}

Let $\Omega$ be a witness, and let $Z$ be the complete set of zeros
of $\Omega$, with $Z'\subseteq Z$ as the subset of quartic zeros.
Each zero in $Z$ defines zeroth and first order equality constraints
$\mathbf{T}_0$ and $\mathbf{T}_1$.  The combination of a set of
constraints is the direct sum.  Thus we define $\mathbf{U}_0$ as the direct
sum of the constraints $\mathbf{T}_0$ over all the zeros in $Z$, and
similarly $\mathbf{U}_1$ as the direct sum of the constraints $\mathbf{T}_1$.  We
define $\mathbf{U}_{01}=\mathbf{U}_0\oplus\mathbf{U}_1$.  Each zero in $Z'$ introduces
additional second and third order constraints $\mathbf{T}_2$ and $\mathbf{T}_3$. We
denote the direct sums of these by $\mathbf{U}_2$ and $\mathbf{U}_3$, respectively,
and we define $\mathbf{U}_{23}=\mathbf{U}_2\oplus\mathbf{U}_3$.  Denote by $\mathbf{U}_\Omega$ the
full system of constraints on $\Omega$,
$\mathbf{U}_\Omega=\mathbf{U}_{01}\oplus\mathbf{U}_{23}$.

It is an important observation that all the constraints are completely
determined by the zeros and Hessian zeros of $\Omega$.  Thus they
depend only indirectly on $\Omega$.  In particular, $\mathbf{U}_{01}$ depends
only on $Z$, whereas $\mathbf{U}_{23}$ depends on $Z'$ and on the kernel of
the Hessian at each quartic zero.

A boundary witness with only quadratic zeros will be called here a
quadratic boundary witness, or simply a quadratic witness, since we
are talking most of the time about boundary witnesses.  A boundary
witness with at least one quartic zero will be called a quartic
witness.  Note that if $\Omega$ is quadratic then $Z'$ is empty so
$\mathbf{U}_\Omega=\mathbf{U}_{01}$.  This classification of boundary witnesses as
quadratic and quartic is fundamental, since the two classes have
rather different properties.  The quadratic witnesses turn up in much
greater numbers in random searches, and are also simpler to
understand theoretically.

%%%%%%%%%%%%%%%%%%%%%%%%%%%%%%%%%%%%%%%%%%%%%%%%%%%%%%%%%%%%%%%%%%%%%%%%%%%%%%%%%%%%%%%%%%%%%%%%%%%%%%%%%
%%%%%%%%%%%%%%%%%%%%%%%%%%%%%%%%%%%%%%%%%%%%%%%%%%%%%%%%%%%%%%%%%%%%%%%%%%%%%%%%%%%%%%%%%%%%%%%%%%%%%%%%%

\section{Extremal witnesses}
\label{sec:extremal}

In this section we apply the constraints developed in the previous
section to study extremal witnesses.  We approach the problem through
the basic definition that a nonextremal witness is one which can be
written as a convex combination of other witnesses.  In the language
of convex sets it is an interior point of a face of $\mathcal{S}_1^{\circ}$ of
dimension one or higher.

In Section~\ref{sec:characterizationI} we characterize the extremal
witnesses in terms of their zeros.  In
Sections~\ref{sec:characterization} and~\ref{sec:facesSud} we
characterize them in terms of the constraints related to their zeros.
and present a search algorithm for finding them numerically.  The
geometrical interpretation of this algorithm is that we start from a
given witness, reconstruct the unique face of $\mathcal{S}_1^{\circ}$ of which it is
an interior point, go to the boundary of this face, and repeat the
process.  In Section~\ref{sec:zerosgeneric} we discuss the expected
number of zeros of extremal witnesses.

The theorems~\ref{thm:extremecondI} and~\ref{thm:extremecond} give
necessary and sufficient conditions for a witness to be extremal, and
they are main results of our work.

\subsection{Zeros of witnesses, convexity and extremality}
\label{sec:characterizationI}

Let $\Omega$ be a convex combination of two different witnesses
$\Lambda$ and $\Sigma$,
\begin{equation}
\Omega=(1-p)\,\Lambda+p\,\Sigma\;,
\end{equation}
with $0<p<1$.  Then the corresponding biquadratic form is a convex
combination
\begin{equation}
\label{eq:biqfOLS}
f_{\Omega}(\phi,\chi)=(1-p)\,f_{\Lambda}(\phi,\chi)+p\,f_{\Sigma}(\phi,\chi)\,,
\end{equation}
and the Hessian matrix at any zero $(\phi_0,\chi_0)$ of $\Omega$ is
also a convex combination
\begin{equation}
G_{\Omega}=(1-p)\,G_{\Lambda}+p\,G_{\Sigma}\,.
\end{equation}
The facts that the biquadratic form of an entanglement witness is
positive semidefinite, and that the Hessian matrix at any zero of a
witness is also positive semidefinite, imply the following result.
\begin{theorem}
\label{thm:witzerosHessian}
\textit{If $\Omega$ is a convex combination of two witnesses $\Lambda$ and
$\Sigma$, as above, then $(\phi_0,\chi_0)$ is a zero of $\Omega$ if
and only if it is a zero of both $\Lambda$ and $\Sigma$.}

\textit{Similarly, $z$ is a Hessian zero of $\Omega$ at the zero
$(\phi_0,\chi_0)$ if and only if it is a Hessian zero at
$(\phi_0,\chi_0)$ of both $\Lambda$ and $\Sigma$.}

\textit{Thus, all witnesses in the interior of the line segment between
$\Lambda$ and $\Sigma$ have exactly the same zeros and Hessian
zeros.}
\end{theorem}

Assume further that $\Lambda$ and $\Sigma$ are extremal points of the
intersection of the straight line with $\mathcal{S}_1^{\circ}$, so that
$\Omega\notin\mathcal{S}_1^{\circ}$ when $p<0$ and when $p>1$.  To be specific,
consider the case $p>1$.  We want to show that $\Sigma$ has at least
one zero or Hessian zero in addition to the zeros and Hessian zeros of
the witnesses in the interior of the line segment.

By assumption, the set of negative points
\begin{equation}
X_-(p)=\{(\phi,\chi)\mid f_{\Omega}(\phi,\chi)<0\}
\end{equation}
is nonempty for $p>1$.  Clearly $X_-(p)\subset X_-(q)$ for $1<p<q$,
since $X_-(p)$ is empty for $0\leq p\leq 1$, and
$f_{\Omega}(\phi,\chi)$ is a linear function of $p$ for fixed
$(\phi,\chi)$ as given by equation \eqref{eq:biqfOLS}.  The closure
$\overline{X}_-(p)$ is a compact set when we normalize so that
$\|\phi\|=\|\chi\|=1$.  It follows that the limit of
$\overline{X}_-(p)$ as $p\to 1+$ is a nonempty set $X_0$ of zeros of
the witness $\Sigma$,
\begin{equation}
X_0=\bigcap_{p>1}\overline{X}_-(p)\,.
\end{equation}

Every zero $(\phi_0,\chi_0)\in X_0$ is of the kind we are looking for.
If it is not a zero of $\Omega$ for $p<1$, then it is a new zero of
$\Sigma$ as compared to the zeros of the witnesses in the interior of
the line segment.  If it is a zero of $\Omega$ for $p<1$, then
$\Sigma$ has at least one Hessian zero at $(\phi_0,\chi_0)$ which is
not a Hessian zero at $(\phi_0,\chi_0)$ of the witnesses in the
interior of the line segment.

To prove the last statement, assume that $(\phi_0,\chi_0)\in X_0$ is a
zero of $\Omega$ for $0<p<1$.  Then it is a zero of $\Omega$ for any
$p$, again because $f_{\Omega}(\phi,\chi)$ is a linear function of
$p$.  Similarly, if $z$ is a Hessian zero of $\Omega$ at
$(\phi_0,\chi_0)$ for $0<p<1$, it is a Hessian zero of $\Omega$ at
$(\phi_0,\chi_0)$ for any $p$.

From the assumption that $(\phi_0,\chi_0)\in X_0$ follows that
$f_{\Omega}(\phi,\chi)$ for any $p>1$ takes negative values for some
$(\phi,\chi)$ arbitrarily close to $(\phi_0,\chi_0)$.  The only way
this may happen is that for $p>1$ the second derivative
$z^TG_{\Omega}z$ is negative in some direction $z$, meaning that
$G_{\Omega}$ for $p>1$ has one or more negative eigenvalues
$\lambda_i(p)$, with eigenvectors $z_i(p)$ that are orthogonal to each
other and orthogonal to the $p$ independent part of $\textrm{Ker}\,G_{\Omega}$.

In the limit $p\rightarrow 1+$ the negative eigenvalues $\lambda_i(p)$ of
$G_{\Omega}$ go to zero, and the corresponding eigenvectors $z_i(p)$
go to eigenvectors $z_i(1)$ of $G_{\Sigma}$ with zero eigenvalues.
These eigenvectors are then Hessian zeros of $\Sigma$ at
$(\phi_0,\chi_0)$ that are not Hessian zeros of $\Omega$ for $0<p<1$.

We summarize the present discussion as follows.
\begin{theorem}
\label{thm:witzerosHessianII}
\textit{If a line segment in $\mathcal{S}_1^{\circ}$ with end points $\Lambda$ and $\Sigma$
can not be prolonged within $\mathcal{S}_1^{\circ}$ in either direction, then
$\Sigma$ has the same zeros and Hessian zeros as the interior points
of the line segment, and at least one additional zero or Hessian zero.}

\textit{The same holds for $\Lambda$, with additional zeros and Hessian
zeros that are different from those of $\Sigma$.}
\end{theorem}
These theorems lead to the following extremality criterion for
witnesses.
\begin{theorem}
\label{thm:extremecondI}
\textit{A witness $\Omega$ is extremal if and only if no witness
$\Lambda\neq\Omega$ has a set of zeros and Hessian zeros including
the zeros and Hessian zeros of $\Omega$.}

\textit{An equivalent condition is that there can exist no witness $\Lambda\neq\Omega$
satisfying the constraints $\mathbf{U}_{\Omega}\Lambda=0$.}
\end{theorem}
\begin{proof}
The ``if'' part follows directly from
Theorem~\ref{thm:witzerosHessian}.

To prove the ``only if'' part, assume that the set of zeros and
Hessian zeros of some witness $\Lambda\neq\Omega$ include the zeros
and Hessian zeros of $\Omega$.  Then by
Theorem~\ref{thm:witzerosHessian} the interior points of the line
segment with $\Lambda$ and $\Omega$ as end points have exactly the
same zeros and Hessian zeros as $\Omega$.  By
Theorem~\ref{thm:witzerosHessianII} this line segment can be
prolonged within $\mathcal{S}_1^{\circ}$ so that it gets $\Omega$ as an interior
point.  Hence $\Omega$ is not extremal.\ \
\end{proof}
\subsection{How to search for extremal witnesses}
\label{sec:characterization}

Once the zeros and Hessian zeros of a witness $\Omega$ are known, it
is a simple computational task to find a finite perturbation of
$\Omega$ within the unique face of $\mathcal{S}_1^{\circ}$ where $\Omega$ is an
interior point.  The most general direction for such a perturbation is
a traceless $\Gamma\in\textrm{Ker}\,\mathbf{U}_\Omega$.  Note that we only need to find
some $\Gamma'\in\textrm{Ker}\,\mathbf{U}_\Omega$ not proportional to $\Omega$, then
$\Gamma=\Gamma'-(\textrm{Tr}\,\Gamma')\,\Omega$ is nonzero and traceless and
lies in $\textrm{Ker}\,\mathbf{U}_\Omega$.

\begin{theorem}
\label{thm:witperturb}
\textit{Let $\Omega$ be a witness, and let $\Gamma\in H$, $\Gamma\neq 0$,
$\mathrm{Tr}\,\Gamma=0$.  Then $\Lambda=\Omega+t\Gamma$ is a witness for all
$t$ in some interval $[t_1,t_2]$, with $t_1<0<t_2$, if and only if
$\Gamma\in\mathrm{Ker}\,\mathbf{U}_\Omega$.}

\textit{The maximal value of $t_2$ is the value of $t$ where $\Lambda$
acquires a new zero or Hessian zero.  The minimal value of $t_1$ is
determined in the same way.}
\end{theorem}
\begin{proof}
To prove the ``only if'' part, assume that $\Lambda=\Omega+t\Gamma$
is a witness for $t_1\leq t\leq t_2$.  Then by
Theorem~\ref{thm:witzerosHessian}, $\Lambda$ has the same zeros and
Hessian zeros as $\Omega$ for $t_1<t<t_2$.  Since the constraints
$\mathbf{U}_{\Omega}$ depend only on the zeros and Hessian zeros of
$\Omega$, we conclude that $\mathbf{U}_{\Omega}\Lambda=0$ for $t_1<t<t_2$,
and hence $\mathbf{U}_{\Omega}\Gamma=0$.

To prove the ``if'' part, assume that $\mathbf{U}_{\Omega}\Gamma=0$.  Since
$\mathbf{U}_{\Omega}\Omega=0$, it follows that $\mathbf{U}_{\Omega}\Lambda=0$ for
any value of $t$.  Consider the set of negative points,
\begin{equation}
X_-(t)=\{(\phi,\chi)\mid f_{\Lambda}(\phi,\chi)<0\}\,.
\end{equation}

We want to argue that $X_-(t)$ must be empty for $t$ in some interval
$[t_1,t_2]$ with $t_1<0<t_2$.

In fact, since $\mathcal{S}_1^{\circ}$ is a compact set, there must exist some $t_2\geq
0$ such that $X_-(t)$ is empty for $0\leq t\leq t_2$ and nonempty for
$t>t_2$.  Then the limit of $\overline{X}_-(t)$ as $t\to t_2+$ is a
nonempty set $X_0$ of zeros of $\Lambda_2=\Omega+t_2\Gamma$.  A
similar reasoning as the one leading to
Theorem~\ref{thm:witzerosHessianII} now leads us to the conclusion
that $\Lambda_2$ must have at least one zero or Hessian zero which is
not a zero or Hessian zero of $\Omega$.  This proves that
$\Lambda_2\neq\Omega$ and $t_2>0$.

By a similar argument we deduce the existence of a lower limit
$t_1<0$.
\end{proof}
Figure~\ref{fig:boundary} shows a model for how a new isolated
quadratic zero of $\Lambda$ appears, or how an existing quadratic zero
turns into a quartic zero, as the parameter $t$ increases in the
function
\begin{equation}
f_t(u)=f_{\Omega+t\Gamma}(\phi+u\phi',\chi+u\chi')\,.
\end{equation}
\begin{figure}
\centering
\includegraphics[width=0.49\textwidth]{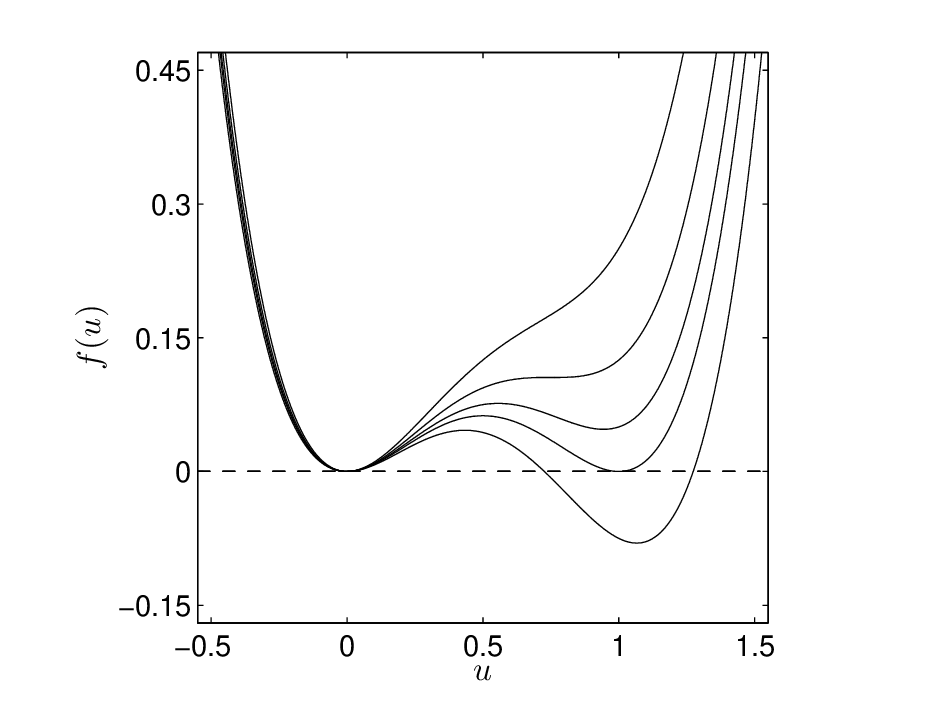}
\includegraphics[width=0.49\textwidth]{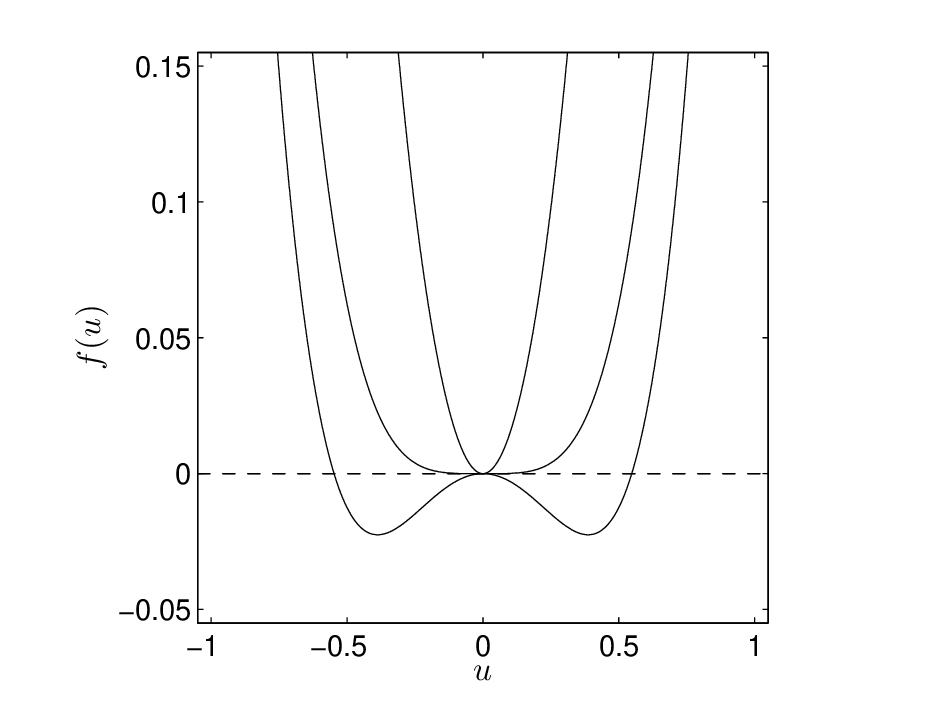}
\caption{Models for how the positivity limit is reached.  Left: a
new local minimum appears, turning into a zero and then a negative
minimum.  The model function plotted for five different values of
$t$ is $f_t(u)=u^2((u-1)^2+1-t)$, depending on a parameter $t$
with the critical value $t_c=1$.  Right: a quadratic zero turns
into a quartic zero, and then new negative minima branch off.  The
model function is $f_t(u)=u^2(u^2+1-t)$, again with the critical
parameter value $t_c=1$.}
\label{fig:boundary}
\end{figure}

The next theorem is an immediate corollary.  It is a slightly stronger
version of Theorem~\ref{thm:extremecondI}.  It is interesting for the
theoretical understanding, and together with
Theorem~\ref{thm:witperturb} it is directly useful for numerical
calculations.
\begin{theorem}
\label{thm:extremecond}
\textit{A witness $\Omega$ is extremal if and only if $\mathrm{Ker}\,\mathbf{U}_\Omega$ is
one dimensional (spanned by $\Omega$).}
\end{theorem}

Thus, once the zeros and Hessian zeros of the witness $\Omega$ are
known we test for extremality by computing the dimension of
$\textrm{Ker}\,\mathbf{U}_\Omega$, for example by a singular value decomposition of
$\mathbf{U}_\Omega$.  This allows for a simple numerical implementation of
the extremality criterion.

These theorems motivate Algorithm~1, a search algorithm for finding
extremal witnesses.  Obviously, any extremal witness might be reached
by this algorithm already in the first iteration, starting for example
from the maximally mixed state.  The search is guaranteed to converge
to an extremal witness in a finite number of iterations, since the
number of possible search directions is reduced in each iteration. We
here summarize this in Algorithm~1. Starting from an initial witness
$\Omega=\Omega_0\in\textrm{int}(\mathcal{S}^{\circ})$, we proceed down
through a hierarchy of faces on $\mathcal{S}^{\circ}$ of decreasing
dimensions. The number of possible search directions $\Gamma$ on each
face, is reduced in each iteration as
$\textrm{Ker}\,\mathbf{U}_\Omega$ is reduced.
\begin{table*}
{\begin{tabular}{ll}
{\bf{Algorithm 1:}} Finding an extremal entanglement witness \\
\hline\hline
{\bf{Precondition:}} Choose an initial witness $\Omega=\Omega_0\in\textrm{int}(\mathcal{S}^{\circ})$ and construct $\mathbf{U}_\Omega$\\
\hline
1. {\bf{while}} $\textrm{dim}\,\textrm{Ker}\,(\mathbf{U}_\Omega)>1$\\
2. \qquad choose a $\Gamma\in\textrm{Ker}\,(\mathbf{U}_\Omega)$\\
3. \qquad {\bf{if}} $\textrm{Tr}(\Gamma)\neq 0$\\
4. \qquad\qquad redefine $\Gamma\leftarrow\Gamma-\textrm{Tr}(\Gamma)\Omega$\\
5. \qquad {\bf{endif}}\\
6. \qquad find $t_c$ as the maximal $t$ such that $\Omega+t\Gamma$ is still a witness\\
7. \qquad redefine $\Omega\leftarrow \Omega+t_c\Gamma$\\
8. \qquad locate all zeros of $\Omega$ and construct an updated $\mathbf{U}_\Omega$\\
9. {\bf{endwhile}}\\
\hline
{\bf{return}} $\Omega$
\end{tabular}}
\end{table*}

\subsection{Faces of the set $\mathcal{S}_1^{\circ}$ of normalized witnesses}
\label{sec:facesSud}

Theorem~\ref{thm:witperturb} has the following geometrical meaning.
Define
\begin{equation}
\label{eq:faceOmegadef}
\mathcal{F}_{\Omega}=(\textrm{Ker}\,\mathbf{U}_{\Omega})\cap\mathcal{S}_1^{\circ}\,.
\end{equation}
Equivalently, define $\mathcal{F}_{\Omega}$ as the set of all witnesses of the
form $\Omega+t\Gamma$ with $\Gamma\in\textrm{Ker}\,\mathbf{U}_{\Omega}$ and
$\textrm{Tr}\,\Gamma=0$.  If $\Omega$ is extremal in $\mathcal{S}_1^{\circ}$ then $\mathcal{F}_{\Omega}$
consists of the single point $\Omega$.  Otherwise, $\mathcal{F}_{\Omega}$ is
the unique face of $\mathcal{S}_1^{\circ}$ having $\Omega$ as an interior point.

Thus, Algorithm~1 produces a decreasing sequence
of faces of $\mathcal{S}_1^{\circ}$, $\mathcal{F}_1\supset\mathcal{F}_2\supset\ldots\supset\mathcal{F}_n$, where
every face $\mathcal{F}_j$ is a face of every $\mathcal{F}_i$ with $i<j$, and the
extremal point found is an extremal point of all these faces.

The theorems in Section~\ref{sec:characterizationI} imply that
every face $\mathcal{F}$ of $\mathcal{S}_1^{\circ}$ is uniquely characterized by a set of zeros
and Hessian zeros that is the complete set of zeros and Hessian zeros
of every witness in the interior of $\mathcal{F}$.  Every boundary point of
$\mathcal{F}$ is a witness having the zeros and Hessian zeros characteristic of
the interior points, plus at least one more zero or Hessian zero.  In
general, the boundary of $\mathcal{S}_1^{\circ}$ is a hierarchy of faces, faces of
faces, faces of faces of faces, and so on.  The number of zeros and
Hessian zeros increases each time we step from one face onto a face of
the face, and the descent through the hierarchy from one face to the
next always ends in an extremal witness.

\subsection{Zeros of extremal witnesses}
\label{sec:zerosgeneric}

A natural question regards the number of zeros a witness must have in
order to be extremal.  We provide two lower bounds for quadratic
witnesses, one obtained by comparison with a pure state as a witness
and the other by counting constraints.  Similar bounds for quartic
witnesses are not easy to obtain.

The first of these bounds reveals a ``double spanning'' property of
zeros of quadratic extremal witnesses.  We define the partial
conjugate of $(\phi,\chi)$ as $(\phi,\chi^*)$.
\begin{theorem}
\label{thm:zerosspan2}
\textit{The zeros of a quadratic extremal witness $\Omega$ span the Hilbert
space.  The partially conjugated zeros also span the Hilbert space.}
\end{theorem}
\begin{proof}
Assume that the zeros span less than the whole Hilbert space, so
that there exists a vector $\psi$ orthogonal to all the zeros.  Then
the projection $P=\psi\psi^{\dagger}$ is a witness with a set of
zeros including all the zeros of $\Omega$.  But $P$ has a continuum
of zeros, all quartic, see Theorem~\ref{thm:r1zeros}.  Hence
$P\neq\Omega$ and $\Omega$ is not extremal, by
Theorem~\ref{thm:extremecondI}.

The partial conjugates of the zeros of $\Omega$ span the Hilbert
space because they are the zeros of the quadratic extremal witness
$\Omega^P$.
\end{proof}

This proof does not hold when $\Omega$ is a quartic witness, because
$\Omega$ may then have Hessian zeros that are not Hessian zeros of
$P$.  A pure state $\psi\psi^{\dagger}$ as a quartic extremal witness
is a counterexample, where the zeros do not span the Hilbert space,
although the partially conjugated zeros may span the Hilbert space.
The partial transpose $(\psi\psi^{\dagger})^P$ is a counterexample of
the opposite kind, where the zeros may span the Hilbert space but the
partially conjugated zeros do not.  Furthermore, even if neither the
zeros nor the partially conjugated zeros span the Hilbert space, it is
still possible for a quartic witness to be extremal, because a quartic
zero leads to more constraints than a quadratic zero.  In
Section~\ref{sec:quartic} we describe the numerical construction of
an extremal quartic witness in dimension $3\times 3$ with only eight
zeros, and with one Hessian zero at one of the zeros.

Counting constraints from quadratic zeros gives a different lower
bound for the number of zeros of a quadratic extremal witness.  There
are $M_2=2(N_a+N_b)-3$ constraints per quadratic zero.  With $n$
quadratic zeros there are a total of $nM_2$ constraints, which in the
generic case are linearly independent until $nM_2$ approaches
$N^2-1$, the dimension of $\mathcal{S}_1^{\circ}$.  A lower bound on the number of
zeros of a quadratic extremal witness is then given by
\begin{equation}
\label{eq:nzquadrextrw}
n_c=\mathrm{ceil}\,\frac{(N_aN_b)^2-1}{2(N_a+N_b)-3}\,,
\end{equation}
where $\textrm{ceil}$ rounds upwards to the nearest integer. 

In the special case $N_a=2$ this formula gives the lower bound
$n_c=N-1$, which is weaker than the lower bound $N$ given by
Theorem~\ref{thm:zerosspan2}.  When $N_a\geq 3$ and $N_b\geq 3$
the formula gives a lower bound $n_c\geq N$, consistent with
Theorem~\ref{thm:zerosspan2}.

Table~\ref{tbl:constraints} lists the numerically computed number of
linearly independent constraints $m_{\textrm{ind}}$ arising from $n_c$ randomly
chosen pairs $(\phi,\chi)$. $n_c$ is the estimated minimum number of zeros required 
for a quadratic witness to be extremal, assuming maximal
independence of constraints.  $M_2$ is the number of linearly 
independent constraints from each zero. In order to compute $m_{\textrm{ind}}$ we interpret
each $(\phi,\chi)$ as a quadratic zero, build $\mathbf{U}_{01}$ and do a
singular value decomposition.  The table shows that in many cases the
number of independent constraints from a random set of product vectors
is equal to the dimension of the real vector space $H$ of Hermitian
matrices.  In such a case there exists no biquadratic form which is
positive semidefinite and has these zeros.  We conclude that the zeros
of a positive semidefinite biquadratic form have to satisfy some
relations that reduce the number of independent constraints.

\begin{table}
  \centering
  \begin{tabular}{crrrr}
    \noalign{\smallskip}
    \hline
    \noalign{\smallskip}
    $N_b,N_b$  &  $N^2$  &  $M_2$          &  $n_c$     &
    \multicolumn{1}{c}{$m_{\mathrm{ind}}$}  \\
    \noalign{\smallskip}
    \hline
    \noalign{\smallskip}
     2,2       &  16     &      5          &   3        &     14         \\
     2,3       &  36     &      7          &   5        &     34         \\
     2,4       &  64     &      9          &   7        &     62         \\
     2,5       & 100     &     11          &   9        &     98         \\
     3,3       &  81     &      9          &   9        &     81         \\
     3,4       & 144     &     11          &  13        &    143         \\
     3,5       & 225     &     13          &  18        &    225         \\ 
     4,4       & 256     &     13          &  20        &    256         \\
     4,5       & 400     &     15          &  27        &    400         \\
     5,5       & 625     &     17          &  37        &    625         \\
    \noalign{\smallskip}
    \hline
  \end{tabular}
\caption{Numbers related to quadratic zeros of witnesses
    in dimension $N=N_aN_b$.
    $M_2$ is the number of linearly 
    independent constraints from each zero.
    $n_c$ is the estimated minimum number of zeros required 
    for a quadratic witness to be extremal, assuming maximal
    independence of constraints. 
    $m_{\mathrm{ind}}$ is the actual number of independent
    constraints from $n_c$ random product vectors.
    Note that for $N_a=2$ there is an intrinsic 
    degeneracy causing the number of constraints from $n_c$ zeros
    to be $n_cM_2-1$, so  the actual minimum number
    of zeros is $n_c+1=N$.}
  \label{tbl:constraints}
\end{table}

Note that for $N_a=2$ there is an intrinsic 
degeneracy causing the number of constraints from $n_c$ zeros
to be $n_cM_2-1$, so  the actual minimum number
of zeros is $n_c+1=N$. So the number of constraints is reduced by one.
For example, in the $2\times 4$ system the actual number of constraints
from $n_c=7$ random product vectors is $62$ instead of $63=n_cM_2$.
This degeneracy implies that one extra zero is needed, thus saving
Theorem~\ref{thm:zerosspan2}.  The proof of
Theorem~\ref{thm:zerosspan2} indicates the origin of the degeneracy.
With $N-1$ zeros there exists a vector $\psi$ orthogonal to all the
zeros and a vector $\eta$ orthogonal to all the partially conjugated
zeros, and because $\psi\psi^{\dagger}$ and $(\eta\eta^{\dagger})^P$
both lie in $\textrm{Ker}\,\mathbf{U}_{01}$ we must have $\dim\textrm{Ker}\,\mathbf{U}_{01}\geq 2$.
For other dimensions there is no similar degeneracy.

The $3\times 4$ system is special in that the number of constraints
from a generic set of $n_c=13$ product vectors add up to exactly the
number needed to define a unique $A\in H$ with $\mathbf{U}_{01} A=0$.  This
does not however imply that $A$ defined in this way from a random set
of product vectors will be an extremal witness, in fact it will
usually not be a witness, because it will have both positive and
negative expectation values in product states.

A central question in general, well worth further attention, is how to
choose a set of product vectors such that they may serve as the zeros
of an extremal witness.  According to Theorem~\ref{th:exposeddualsIV},
the zeros of a witness must define pure states lying on a face of the
set $\mathcal{S}$ of separable states.  However, this statement leaves open the
practical question of how to test numerically whether a set of pure
product states lies on a face of $\mathcal{S}$.

In $2\times N_b$ and $3\times 3$ systems, the minimum number of zeros
of a quadratic extremal witness is exactly equal to the dimension
$N=N_aN_b$ of the composite Hilbert space $\mathcal{H}$.  This implies the
following result, which we state as a theorem.
\begin{theorem}
\label{thm:2xn3x3}
\textit{In dimensions $2\times N_b$ and $3\times 3$ a generic witness with
doubly spanning zeros will be extremal.}
\end{theorem}
We do not know whether this is a general theorem in these dimensions,
valid for all witnesses and not only for the generic witnesses.  It
may be that some construction like the one described in
Section~\ref{sec:nongeneric} will lead to nongeneric counterexamples.

In higher dimensions the minimum number of zeros to define a quadratic
extremal witness is strictly larger than the dimension of $\mathcal{H}$. This
difference between $2\times N_b$ and $3\times 3$ systems on the one
hand, and higher dimensional systems on the other hand, has important
consequences discussed in Sections~\ref{sec:separable}
and~\ref{sec:optimal}.

%%%%%%%%%%%%%%%%%%%%%%%%%%%%%%%%%%%%%%%%%%%%%%%%%%%%%%%%%%%%%%%%%%%%%%%%%%%%%%%%%%%%%%%%%%%%%%%%%%%%%%%%%%%%%%%%%%%%%%%%%%%%%%%%
%%%%%%%%%%%%%%%%%%%%%%%%%%%%%%%%%%%%%%%%%%%%%%%%%%%%%%%%%%%%%%%%%%%%%%%%%%%%%%%%%%%%%%%%%%%%%%%%%%%%%%%%%%%%%%%%%%%%%%%%%%%%%%%%

\section{Decomposable witnesses}
\label{sec:decompwitnesses}

A decomposable witness is called so because it corresponds to a
decomposable positive map, and because it has the form
$\Omega=\rho+\sigma^P$ with $\rho,\sigma\in\mathcal{D}$, possibly $\rho=0$ or
$\sigma=0$.  Although the decomposable witnesses are useless for
detecting entangled PPT states, they are useful in other ways, for
example as stepping stones in some of our numerical methods for
constructing extremal nondecomposable witnesses.

In this section we will summarize some basic properties of
decomposable witnesses.  This is a natural place to start when we want
to understand entanglement witnesses in general.  In particular, we
are interested in the relation between a witness and its zeros.

Since the set $\mathcal{P}_1^{\circ}$ of decomposable witnesses is the convex hull of
$\mathcal{D}_1$ and $\mathcal{D}_1^P$, an extremal point of $\mathcal{P}_1^{\circ}$ must be an extremal
point of either $\mathcal{D}_1$ or $\mathcal{D}_1^P$.  That is, it must be either a pure
state $\psi\psi^\dagger$ or a partially conjugated pure state
$(\psi\psi^\dagger)^P\!$, or both if $\psi$ is a product vector.  In
the present section we verify the results in~\cite{Sto1963,Grab,Mar2008}, that witnesses of the forms
$\psi\psi^\dagger$ and $(\psi\psi^\dagger)^P$ are extremal in $\mathcal{S}_1^{\circ}$.
Since they are extremal in $\mathcal{S}_1^{\circ}$, they are also extremal in the
subset $\mathcal{P}_1^{\circ}\subset\mathcal{S}_1^{\circ}$.  Thus they are precisely the extremal points
of $\mathcal{P}_1^{\circ}$.

\subsection{Zeros of a decomposable witness}

The biquadratic form corresponding to the decomposable witness
$\Omega=\rho+\sigma^P$ is
\begin{equation}
f_{\Omega}(\phi,\chi)=
f_{\rho}(\phi,\chi)+
f_{\sigma}(\phi,\chi^*)\,.
\end{equation}
It is positive semidefinite because $f_{\rho}$ and $f_{\sigma}$ are
positive semidefinite.

Assume now that $(\phi_0,\chi_0)$ is a zero of $\Omega$.  The above
decomposition of $f_{\Omega}$ shows that $(\phi_0,\chi_0)$ must be a
zero of $\rho$, and the partial conjugate $(\phi_0,\chi_0^*)$ must be
a zero of $\sigma$.  But because $\rho,\sigma\in\mathcal{D}$ it follows that
\begin{equation}
\rho(\phi_0\otimes\chi_0)=0\,,\qquad
\sigma(\phi_0\otimes\chi_0^*)=0\,.
\end{equation}
This proves the following theorem.
\begin{theorem}
\label{thm:zerosnospan2}
\textit{The zeros of a decomposable witness $\Omega=\rho+\sigma^P$ span the
Hilbert space only if $\rho=0$.  The partially conjugated zeros span
the Hilbert space only if $\sigma=0$. Hence, by Theorem~\ref{thm:zerosspan2},
a quadratic extremal witness is nondecomposable.}
\end{theorem}

Note that if $(\phi_0,\chi_0)$ is a zero of $\Omega$
and  $\phi\otimes\chi$ is any product vector, then
\begin{equation}
(\phi\otimes\chi_0)\sigma^P\!(\phi_0\otimes\chi)=
(\phi\otimes\chi^*)\sigma(\phi_0\otimes\chi_0^*)=0\,,
\end{equation}
and hence
\begin{equation}
\label{eq:rhomatel}
(\phi\otimes\chi_0)\Omega(\phi_0\otimes\chi)=
(\phi\otimes\chi_0)\rho(\phi_0\otimes\chi)\,.
\end{equation}
Similarly,
\begin{equation}
(\phi\otimes\chi_0^*)\rho^P\!(\phi_0\otimes\chi^*)=
(\phi\otimes\chi)\rho(\phi_0\otimes\chi_0)=0\,,
\end{equation}
and hence
\begin{equation}
\label{eq:sigmamatel}
(\phi\otimes\chi_0^*)\Omega^P\!(\phi_0\otimes\chi^*)=
(\phi\otimes\chi_0^*)\sigma(\phi_0\otimes\chi^*)\,.
\end{equation}
These equations may be useful for computing $\rho$ and $\sigma$ if
$\Omega$ is known but its decomposition $\Omega=\rho+\sigma^P$ is
unknown.  We will return to this problem in
Section~\ref{subsec:decompdecompwitness}.

\subsection{Pure states and partially transposed pure states}

Let $P$ be a pure state, $P=\psi\psi^{\dagger}$, and let $Q$ be the
partial transpose of a pure state, $Q=(\eta\eta^{\dagger})^P\!$, with
$\psi,\eta\in\mathcal{H}$.  The corresponding positive maps $\mathbf{L}_P$ and $\mathbf{L}_Q$
are rank one preservers: they map matrices of rank one to matrices of
rank one or zero, because
\begin{equation}
\mathbf{L}_P(\phi\phi`{\dagger})
=(\phi\otimes I_b)^\dagger P (\phi\otimes I_b)
=\zeta\zeta^\dagger
\qquad\mathrm{with}\qquad
\zeta=(\phi\otimes I_b)^\dagger\psi\,,
\end{equation}
and
\begin{equation}
\mathbf{L}_Q(\phi\phi`{\dagger})
=(\phi\otimes I_b)^\dagger Q (\phi\otimes I_b)
=\zeta\zeta^\dagger
\qquad\mathrm{with}\qquad
\zeta=(\phi^*\otimes I_b)^\dagger\eta^*\,.
\end{equation}
The corresponding biquadratic forms are positive semidefinite because
they are absolute squares,
\begin{equation}
f_P(\phi,\chi)
=(\phi\otimes\chi)^\dagger P(\phi\otimes\chi)
=|\psi^{\dagger}(\phi\otimes\chi)|^2\,,
\end{equation}
and
\begin{equation}
f_Q(\phi,\chi)
=(\phi\otimes\chi)^\dagger Q(\phi\otimes\chi)
=(\phi\otimes\chi^*)^\dagger Q^P(\phi\otimes\chi^*)
=|\eta^{\dagger}(\phi\otimes\chi^*)|^2\,.
\end{equation}

The zeros of $P$ are the product vectors orthogonal to $\psi$.
A singular value decomposition (Schmidt decomposition) gives
orthonormal bases $\{u_i\}$ in $\mathcal{H}_a$ and $\{v_j\}$ in $\mathcal{H}_b$
such that
\begin{equation}
\psi=\sum_{i=1}^m c_i\,u_i\otimes v_i
\qquad\mathrm{with}\quad c_i>0\,.
\end{equation}
Here $m$ is the Schmidt number of $\psi$, $1\leq
m\leq\min(N_a,N_b)$.  The condition for a product vector
\begin{equation}
\phi\otimes\chi
=\sum_{i=1}^{N_a}\sum_{j=1}^{N_b} a_ib_j\,u_i\otimes v_j
\end{equation}
to be orthogonal to $\psi$ is that
\begin{equation}
\sum_{i=1}^m c_ia_ib_i=0\,.
\end{equation}
For any dimensions $N_a\geq 2$, $N_b\geq 2$ and any Schmidt
number $m$ the set of zeros is continuous and connected.  All the
zeros are quartic, since quadratic zeros are isolated.  The zeros do
not span the whole Hilbert space, only the subspace orthogonal to
$\psi$.  However, the partially conjugated zeros span the Hilbert
space, except when $m=1$ so that $\psi$ is a product
vector~\cite{Aug}.  This almost completes the proof of the following
theorem, what remains is to prove the extremality.

\begin{theorem}
\label{thm:r1zeros}
\textit{A pure state as a witness has only quartic zeros, all in one
continuous connected set, and it is extremal in $\mathcal{S}_1^{\circ}$.  The same is
true for the partial transpose of a pure state.}
\end{theorem}
\begin{proof}
We want to prove that $P$ is extremal, the proof for $Q$ is similar.
We show that the only witness $\Omega$ satisfying the constraints
$\mathbf{U}_P\Omega=0$ is $\Omega=P$, from this we conclude that $P$ is
extremal, using Theorem~\ref{thm:extremecondI}.

Even though the zeros of $P$ are all quartic, they are so many that
we need not use the secondary constraints coming from the quartic
nature of the zeros.
According to equation \eqref{eq:gradCLL}, the constraints
$\mathbf{U}_P\Omega=0$ include the constraints
$\mathbf{L}_\Omega(\phi\phi^\dagger)\chi=0$ for every zero $(\phi,\chi)$ of
$P$.  The zeros of $P$ are the solutions of the equation
$\psi^{\dagger}(\phi\otimes\chi)=0$.  For every $\phi\in\mathcal{H}_a$ this is one
linear equation for $\chi\in\mathcal{H}_a$, the solutions of which form a
subspace of $\mathcal{H}_a$ of dimension either $N_b-1$ or $N_b$.  But
this means that $\mathbf{L}_\Omega(\phi\phi^\dagger)$ has rank either one
or zero, so that the positive map $\mathbf{L}_{\Omega}$ is a rank one
preserver.  Hence, either $\Omega=\omega\omega^{\dagger}$ or
$\Omega=(\widetilde{\omega}\widetilde{\omega}^{\dagger})^P$ for
some $\omega,\widetilde{\omega}\in\mathcal{H}$~\cite{Grab}.

In the first case, $\Omega=\omega\omega^{\dagger}\!$, we must have
$\omega^{\dagger}(\phi\otimes\chi)=0$ for every $(\phi,\chi)$ such that
$\psi^{\dagger}(\phi\otimes\chi)=0$.  But then $\omega$ must be proportional
to $\psi$, and $\Omega=P$.

In the second case,
$\Omega=(\widetilde{\omega}\widetilde{\omega}^{\dagger})^P\!$, we
must have $\widetilde{\omega}^{\dagger}(\phi\otimes\chi^*)=0$ for every
$(\phi,\chi)$ such that $\psi^{\dagger}(\phi\otimes\chi)=0$.  This is
possible, but only if $\psi$ and $\widetilde{\omega}$ are product
vectors, with $\widetilde{\omega}$ the partial conjugate of $\psi$.
Then we again have $\Omega=P$.
\end{proof}

The fact that every zero $(\phi_0,\chi_0)$ of $P$ is quartic can also
be seen directly from equation~(\ref{eq:hess0}).  With
$\Omega=P=\psi\psi^{\dagger}$ and
$\psi^{\dagger}(\phi_0\otimes\chi_0)=0$ equation~(\ref{eq:hess0})
takes the form
\begin{equation}
\label{eq:hess01}
z^T G_{\Omega}\,z
=f_2(x,y)
=|\psi^{\dagger}(\xi\otimes\chi_0+\phi_0\otimes\zeta)|^2\,.
\end{equation}
Thus, $z^T G_{\Omega}\,z=0$ unless the vector
$\xi\otimes\chi_0+\phi_0\otimes\zeta\in\mathcal{H}$ has a component along $\psi$.  This
means that the Hessian matrix $G_{\Omega}$ at the zero
$(\phi_0,\chi_0)$ has rank at most two.

\subsection{Decomposing a decomposable witness}
\label{subsec:decompdecompwitness}

If we want to prove that a given witness $\Omega$ is decomposable, the
definitive solution is of course to decompose it as
$\Omega=\rho+\sigma^P$ with $\rho,\sigma\in\mathcal{D}$.  We want to discuss
here methods for doing this, based on a knowledge of zeros of
$\Omega$.  Unfortunately, the present discussion does not lead to a
complete solution of the problem.

Assume that we know a finite set of $k$ zeros of $\Omega$,
$Z=\{\phi_i\otimes\chi_i\}$, with partial conjugates
$\widetilde{Z}=\{\phi_i\otimes\chi_i^*\}$.  $Z$ may be the complete
set of zeros of $\Omega$, or only a subset.  The orthogonal complement
$Z^{\perp}$ is a $d_1$ dimensional subspace of $\mathcal{H}$, and
$\widetilde{Z}^{\perp}$ is a $d_2$ dimensional subspace, with $d_1\geq
N-k$ and $d_2\geq N-k$.  Let $P$ and $\widetilde{P}$ be the orthogonal
projections onto $Z^{\perp}$ and $\widetilde{Z}^{\perp}\!$,
respectively.

Define orthogonal projections $\mathbf{P}$ and
$\widetilde{\mathbf{P}}$ on the real Hilbert space $H$ such that
$\mathbf{P}X=PXP$ and
$\widetilde{\mathbf{P}}X=(\widetilde{P}X^P\!\widetilde{P})^P$ for
$X\in H$.  Define the overlap of these two projections as the
orthogonal projection $\mathbf{O}$ onto the subspace
$(\mathbf{P}H)\cap(\widetilde{\mathbf{P}}H)$.  Since $\mathbf{O}X=X$
if and only if $\mathbf{P}X=X$ and $\widetilde{\mathbf{P}}X=X$, we may
compute $\mathbf{O}$ numerically by picking out the eigenvectors of
$\mathbf{P}+\widetilde{\mathbf{P}}$ with eigenvalue two.  It follows
that
\begin{equation}
\mathbf{O}
=\mathbf{P}\mathbf{O}
=\mathbf{O}\mathbf{P}
=\widetilde{\mathbf{P}}\mathbf{O}
=\mathbf{O}\widetilde{\mathbf{P}}\,,
\end{equation}
and that $\mathbf{P}'=\mathbf{P}-\mathbf{O}$ and
$\widetilde{\mathbf{P}}'=\widetilde{\mathbf{P}}-\mathbf{O}$ are
orthogonal projections.

If $\Omega=\rho+\sigma^P$ with $\rho,\sigma\in\mathcal{D}$, then we must have
$Z\subset\textrm{Ker}\,\rho$ and $\widetilde{Z}\subset\textrm{Ker}\,\sigma$, hence
$\mathbf{P}\rho=\rho$ and
$\widetilde{\mathbf{P}}\sigma^P=(\widetilde{P}\sigma\widetilde{P})^P
=\sigma^P\!$.  Defining $\rho_1=\mathbf{P}'\rho$,
$\rho_2=\mathbf{O}\rho$,
$\sigma_1^P=\widetilde{\mathbf{P}}'\sigma^P\!$,
$\sigma_2^P=\mathbf{O}\sigma^P$ we have that
\begin{equation}
\Omega
=\rho_1+\rho_2+\sigma_2^P+\sigma_1^P
\end{equation}
with $\rho_1\in\mathbf{P}'H$,
$\sigma_1^P\in\widetilde{\mathbf{P}}'H$, and
$\rho_2+\sigma_2^P\in\mathbf{O}H$.  It follows that
\begin{equation}
\rho_1+\sigma_1^P=\Omega-\mathbf{O}\Omega\,,
\qquad
\rho_2+\sigma_2^P=\mathbf{O}\Omega\,.
\end{equation}
The decomposition of $\Omega-\mathbf{O}\Omega$ into $\rho_1$ and
$\sigma_1^P$ is unique and easily computed.

If the overlap $\mathbf{O}$ is zero, then $\rho_2=0$, $\sigma_2=0$,
and this is the end of the story, except that we should check that
both $\rho_1$ and $\sigma_1$ are positive semidefinite.  Otherwise it
remains to decompose $\mathbf{O}\Omega$ into $\rho_2$ and
$\sigma_2^P$ in such a way that $\rho=\rho_1+\rho_2$ and
$\sigma=\sigma_1+\sigma_2$ are both positive semidefinite.  This is
the difficult part of the problem, and the solution, if it exists,
need not be unique.

We have not pursued the problem further.  To do so one should use the
equations~(\ref{eq:rhomatel}) and~(\ref{eq:sigmamatel}), which put
strong and presumably useful restrictions on $\rho$ and $\sigma$.  For
any product vector $\phi\otimes\chi$ and any
$\phi_i\otimes\chi_i\in Z$ it is required that
\begin{align}
\nonumber
(\phi\otimes\chi_i)\rho(\phi_i\otimes\chi) & =
(\phi\otimes\chi_i)\Omega(\phi_i\otimes\chi)\,,\\
(\phi\otimes\chi_i^*)\sigma(\phi_i\otimes\chi^*) & =
(\phi\otimes\chi_i^*)\Omega^P(\phi_i\otimes\chi^*)\,.
\end{align}

\subsection{Decomposable witnesses with prescribed zeros}
\label{subsec:decompwitnessesprs0}

One use of decomposable witnesses is that they provide examples of
witnesses with prescribed zeros.  We will describe now how this works,
and we use the same notation as in the previous subsection.

Let $\Omega=\rho+\sigma^P$ be a decomposable witness, as before.  The
necessary and sufficient conditions for $\Omega$ to have $Z$ as a
predefined set of zeros is that $\rho=UU^{\dagger}$ and
$\sigma=VV^{\dagger}$ with $X,Y\in H$, $U=PX$, $V=\widetilde{P}Y$.  We
choose $\rho$ and $\sigma$ to have the maximal ranks
\begin{equation}
\mathrm{rank}\,\rho=d_1\,,\qquad
\mathrm{rank}\,\sigma=d_2\,,
\end{equation}
where $d_1\geq N-k$ and $d_2\geq N-k$ are the dimensions of
$Z^{\perp}$ and $\widetilde{Z}^{\perp}$ respectively.  In the generic
case, when $k<N$ there will be no linear dependencies between the
zeros, or between their partial conjugates, so that we will have
$d_1=d_2=N-k$.

The set of unnormalized decomposable witnesses of this form has
dimension
\begin{align}
\nonumber
d_D & = \mathrm{rank}\,\mathbf{P}'+\mathrm{rank}\,\widetilde{\mathbf{P}}'+\mathrm{rank}\,\mathbf{O}\\
\nonumber
& = \mathrm{rank}\,\mathbf{P}+\mathrm{rank}\,\widetilde{\mathbf{P}}-\mathrm{rank}\,\mathbf{O}\\ 
& = d_1^{\;2}+d_2^{\;2}-\mathrm{rank}\,\mathbf{O}\,.
\end{align}
The corresponding set of normalized witnesses is a convex subset
$\mathcal{F}_D\subset\mathcal{F}_{\Omega}$ of dimension $d_D-1$, consisting of witnesses
of the form $\Omega+t\Gamma$ with
\begin{equation}
\Gamma=A+B^P\!,\quad
\mathbf{P}A=A\,,\quad
\widetilde{\mathbf{P}}(B^P)=B^P\!,\quad
\textrm{Tr}\,A=\textrm{Tr}\,B=0\,.
\end{equation}
These conditions ensure that $\rho+tA$ and $\sigma+tB$ will remain
positive for small enough positive or negative finite values of the
real parameter $t$.

The witness $\Omega$ is an interior point of a unique face of $\mathcal{S}_1^{\circ}$,
which we call $\mathcal{F}_{\Omega}$, as defined in
equation~(\ref{eq:faceOmegadef}).  The dimension of this face is
$d_{\Omega}-1$ when we define
\begin{equation}
d_{\Omega}=\dim\mathrm{Ker}\,\mathbf{U}_{\Omega}\;.
\end{equation}
For a quadratic witness $\Omega$ a lower limit is
\begin{equation}
\label{eq:dOmegalowerbound}
d_{\Omega}\geq N^2-kM_2
\qquad\mathrm{with}\qquad
M_2=2(N_a+N_b)-3\,.
\end{equation}
This inequality will usually be an equality, especially when $k$ is
small.

Since $\mathcal{F}_D\subset\mathcal{F}_{\Omega}$ we must have $d_D\leq d_{\Omega}$,
implying a lower bound for the overlap $\mathbf{O}$,
\begin{equation}
\mathrm{rank}\,\mathbf{O}\geq
d_1^{\;2}+d_2^{\;2}-d_{\Omega}\geq
2(N-k)^2-d_{\Omega}\,.
\end{equation}

In particular, when the inequality in
equation~(\ref{eq:dOmegalowerbound}) is an equality we have the
nontrivial lower bound
\begin{equation}
\mathrm{rank}\,\mathbf{O}\geq
N^2-k(4N-2k-M_2)\,.
\end{equation}
The generic case when $k$ is small is that $d_D=d_{\Omega}$, which
means that $\mathrm{rank}\,\mathbf{O}$ has the minimal value allowed by this
inequality.  See Table~\ref{tbl:decomposable} where we have tabulated
some values for $d_D$ and $d_{\Omega}$ found numerically.

When we prescribe $k$ zeros of the decomposable witness $\Omega$,
there is a possibility that the actual number of zeros $k'\!$, is
larger than $k$.  A simple counting exercise indicates how many zeros
to expect.  A product vector $\phi\otimes\chi$ is a zero if and only if
$\rho(\phi\otimes\chi)=0$ and $\sigma(\phi\otimes\chi^*)=0$.  These are
$2(N-k)$ complex equations for $\phi\otimes\chi$.  There are
$N_a+N_b-2$ complex degrees of freedom in a product vector, up to
normalization of each factor.  The critical value $k=k_c$ is the
number of zeros for which the number of equations equals the number of
variables, \textit{i.e.},
\begin{equation}
k_c=N+1-\frac{N_a+N_b}{2}\;.
\end{equation}
We expect to find
\begin{align}
\nonumber
k'& =  k\phantom{\infty}  \;\; \text{if\ \ }k<k_c\,,\\
\nonumber
k'&\geq k\phantom{\infty} \;\; \text{if\ \ }k=k_c\,,\\
k'& = \infty\phantom{k}   \;\; \text{if\ \ }k>k_c\,.
\end{align}

Whether $\Omega$ will be quadratic or quartic can be estimated as
follows.  The Hessian $G_{\Omega}$ at a zero of $\Omega$ is a real
square matrix of dimension $2(N_a+N_b-2)$.  If $\mathrm{rank}\,\rho=d_1$ and
$\mathrm{rank}\,\sigma=d_2$, then we may write $\Omega=\rho+\sigma^P$ as a
convex combination of $d_1+d_2$ or fewer extremal decomposable
witnesses.  Each extremal decomposable witness contributes at most two
nonzero eigenvalues to $G_{\Omega}$.  Thus, if we define $d_H$ as the
minimal dimension of the kernel of $G_{\Omega}$ at any zero, we have a
lower bound
\begin{equation}
\label{eq:dHlowerbound}
d_H\geq 2(N_a+N_b-d_1-d_2-2)\,.
\end{equation}

\begin{table}
\centering
{\begin{tabular}{cccccc} 
\hline
\noalign{\smallskip}
$N_a\times N_b$ &  $k$  & $k'$  &  $d_D$  &  $d_{\Omega}$ & $d_H$  \\
\noalign{\smallskip}
\hline
\noalign{\smallskip}
     $2\times 4$ &   1  &    1   &     55      &  55         & 0   \\
     $k_c=6$     &   2  &    2   &     46      &  46         & 0    \\
                 &   3  &    3   &     37      &  37         & 0    \\
                 &   4  &    4   &     28      &  28         & 0    \\
                 &   5  &    5   &     18      &  19         & 2    \\
                 &   6  &    8   &\phantom{1}8 &  10         & 4    \\ 
                 &   7  &$\infty$&\phantom{1}2 &\phantom{1}2 & 6    \\ 
\noalign{\smallskip}
\hline
\noalign{\smallskip}
     $3\times 3$ &   1 &  1     &     72      &  72          & 0    \\
     $k_c=7$     &   2 &  2     &     63      &  63          & 0     \\
                 &   3 &  3     &     54      &  54          & 0    \\
                 &   4 &  4     &     44      &  45          & 0    \\
                 &   5 &  5     &     32      &  36          & 0    \\
                 &   6 &  6     &     18      &  27          & 2    \\
                 &   7 & 10,14  &\phantom{1}8 &  18          & 4    \\
                 &   8 &$\infty$&\phantom{1}2 &\phantom{1}9  & 6    \\
\noalign{\smallskip}
\hline
\noalign{\smallskip}
     $3\times 4$ &  1 & 1   &    133      & 133           & 0    \\
     $k_c=9.5$   &  2 & 2   &    122      & 122           & 0    \\
                 &  3 & 3   &    111      & 111           & 0    \\
                 &  4 & 4   &    100      & 100           & 0    \\
                 &  5 & 5   &\phantom{1}88 &\phantom{1}89 & 0    \\
                 &  6 & 6   &\phantom{1}72 &\phantom{1}78 & 0    \\
                 &  7 & 7   &\phantom{1}50 &\phantom{1}67 & 2    \\
                 &  8 & 8   &\phantom{1}32 &\phantom{1}56 & 4    \\
                 &  9 & 9   &\phantom{1}18 &\phantom{1}45 & 6    \\
       & 10\phantom{1}&$\infty$&\phantom{11}8 &\phantom{1}34 & 8    \\
       & 11\phantom{1}&$\infty$&\phantom{11}2 &\phantom{1}23 & 10\phantom{8}\\
       & 12\phantom{1}& $-$    &\phantom{11}0 &\phantom{1}12 & $-$    \\
\noalign{\smallskip}
\hline
\noalign{\smallskip}
\end{tabular}}
  \caption{Numbers related to decomposable witnesses, as explained in the text.
    $k$ is the prescribed number of zeros and $k'$ is the actual number,
    which may be larger.
    $k_c$ is the critical number of prescribed zeros for which we expect
    a finite number of zeros in total. If $k<k_c$ we expect only the 
    prescribed zeros.  If $k=k_c$ we expect some finite number of zeros.
    If $k>k_c$ we expect a continuum of zeros.}
\label{tbl:decomposable}
\end{table}

In Table~\ref{tbl:decomposable} we list some numbers related to
decomposable witnesses, for different dimensions $N_a\times N_b$.
The numbers presented were obtained numerically as follows.  We choose
$k$ random product vectors.  When $k<N$ there will exist decomposable
witnesses with these as zeros, and $d_D$ is the computed dimension of
the set of such witnesses.  $k'$ is the actual number of zeros they
have.  The listed value of $d_{\Omega}$, the dimension of the kernel
of $\mathbf{U}_\Omega$ when $\Omega$ is a quadratic witness with these zeros,
is the lower bound given in equation \eqref{eq:dOmegalowerbound}.  The
listed value of the dimension $d_H$ of the kernel of the Hessian is
the lower bound from equation \eqref{eq:dHlowerbound}.

We see from the table that with only a few zeros, the set of
decomposable witnesses in the face has the same dimension as the face
itself.  But $\mathcal{F}_{\Omega}$ can not consist entirely of decomposable
witnesses, because in that case our numerical searches for extremal
witnesses, where we search for witnesses with increasing numbers of
zeros, would only produce decomposable witnesses, exactly the opposite
of what actually happens.  Hence we conclude that $\mathcal{F}_D$ is a closed
subset of $\mathcal{F}_{\Omega}$, of the same dimension as $\mathcal{F}_{\Omega}$ but
strictly smaller than $\mathcal{F}_{\Omega}$.

On a face defined by a higher number of zeros the set of decomposable
witnesses has a lower dimension.  In some cases the decomposable
witnesses will lie on the boundary of the face in question. This will
either be because they have more than $k$ zeros, or because some of
the $k$ zeros are quartic.  In other cases the decomposable witnesses
make up a low dimensional part of the interior of the quadratic face.
This is the case \textit{e.g} for faces $\mathcal{F}_4$ in $3\times 3$ systems, where
the set of decomposable witnesses is 44 dimensional, the face is 45
dimensional, and the decomposable witnesses have four quadratic zeros and
hence are situated in the interior of the face.

%%%%%%%%%%%%%%%%%%%%%%%%%%%%%%%%%%%%%%%%%%%%%%%%%%%%%%%%%%%%%%%%%%%%%%%%%%%%%%%%%%%%%%%%%%%%%%%%%%%%%%%%%%%%
%%%%%%%%%%%%%%%%%%%%%%%%%%%%%%%%%%%%%%%%%%%%%%%%%%%%%%%%%%%%%%%%%%%%%%%%%%%%%%%%%%%%%%%%%%%%%%%%%%%%%%%%%%%%

\pagebreak

\section{Examples of extremal witnesses}
\label{sec:knownexamples}

In this section we apply Theorem~\ref{thm:extremecond} and
Algorithm~1 to study examples of extremal
witnesses.  We study first examples known from the literature, in
particular the Choi--Lam witness~\cite{Choi1975,ChoiLam1977} and the
Robertson witness~\cite{Robertson1985}.  In
Sections~\ref{sec:newexamples}, \ref{sec:quadratic},
and~\ref{sec:quartic} we construct numerical examples of generic
witnesses, study these and report some observations.  Finally, in
Section~\ref{sec:nongeneric} we present an example of a nongeneric
extremal witness, numerically constructed, having more than the
minimum number of zeros.

%\pagebreak

\subsection{The Choi--Lam witness}

Define, like in~\cite{Ha2012a},

\settowidth{\mycolwd}{$-\mathrm{e}^{-\mathrm{i}\theta}$}
\begin{equation}
\label{eq:choiwitgen}
\Omega_K(a,b,c;\theta)=\left(
\begin{array}{*{9}{@{}C{\mycolwd}@{}}}
    a & . & . & . & -\mathrm{e}^{\mathrm{i}\theta} & . & . & . & -\mathrm{e}^{-\mathrm{i}\theta} \\
    . & c & . & . & . & . & . & . & . \\
    . & . & b & . & . & . & . & . & . \\
    . & . & . & b & . & . & . & . & . \\
   -\mathrm{e}^{-\mathrm{i}\theta} & . & . & . & a & . & . & . & -\mathrm{e}^{\mathrm{i}\theta} \\
    . & . & . & . & . & c & . & . & . \\
    . & . & . & . & . & . & c & . & . \\
    . & . & . & . & . & . & . & b & . \\
   -\mathrm{e}^{\mathrm{i}\theta} & . & . & . & -\mathrm{e}^{-\mathrm{i}\theta} & . & . & . & a   
\end{array}
\right)
\end{equation}
We write dots instead of zeros in the matrix to make it more readable.
The special case $\Omega_C=\Omega_K(1,0,1;0)$ is the Choi--Lam witness,
one of the first examples of a nondecomposable
witness~\cite{Choi1975,ChoiLam1977}.  The set of zeros of $\Omega_C$
consists of three isolated quartic zeros
\begin{equation}
\label{eq:choiisolzeros}
e_{13}=e_1\otimes e_3\,,\quad
e_{21}=e_2\otimes e_1\,,\quad
e_{32}=e_3\otimes e_2\,,
\end{equation}
where $e_1,e_2,e_3$ are the natural basis vectors in $\mathbb{C}^3\!$,
and a continuum of zeros $\phi\otimes\chi$ where
$\alpha,\beta\in\mathbb{R}$ and
\begin{equation}
\label{eq:choicontzeros}
\phi=e_1+\mathrm{e}^{i\alpha}\,e_2+\mathrm{e}^{i\beta}\,e_3\,,
\qquad
\chi=\phi^*\!.
\end{equation}
The product vectors defined in equation~\eqref{eq:choicontzeros} span
a seven dimensional subspace consisting of all vectors $\psi\in\mathbb{C}^9$
with components $\psi_1=\psi_5=\psi_9$.  The three product vectors
defined in equation~\eqref{eq:choiisolzeros} have
$\psi_1=\psi_5=\psi_9=0$ and lie in the same subspace.

The Hessian has a doubly degenerate kernel at each of these zeros.
Hence a single zero of $\Omega=\Omega_C$ contributes $29$ equations
in $\mathbf{U}_{\Omega}$: nine from $\mathbf{T}_0$ and $\mathbf{T}_1$, $2\cdot8$ from
$\mathbf{T}_2$, and four from $\mathbf{T}_3$.  However, by a singular value
decomposition of 

\begin{equation}
\mathbf{T}=\mathbf{T}_0\oplus\mathbf{T}_1\oplus\mathbf{T}_2\oplus\mathbf{T}_3
\end{equation}
at one zero
at a time we find numerically that the number of independent equations
is $24$ for each of the isolated zeros and $28$ for any randomly
chosen nonisolated zero. We again see that redundant equations appear
when the kernel of the Hessian is more than one dimensional.  We also
observe that the redundancies depend on the nature of the zero.

Choosing increasingly many nonisolated zeros, only 67 linearly
independent equations are obtained, out of the $80$ needed for proving
extremality.  These 67 equations are obtained with the quadratic and
quartic constraints from three zeros, or with only the quadratic
constraints from nine zeros.  With all three isolated zeros and any
single zero from the continuum, $\textrm{Ker}\,\mathbf{U}_{\Omega}$ is one dimensional
and uniquely defines the Choi--Lam witness.  This verifies numerically
that it is extremal.  We need the quartic constraints from the three
isolated zeros in order to prove extremality, because the quadratic
constraints from all the zeros provide only 76 independent equations,
leaving a four dimensional face of witnesses having the same set of
zeros as the Choi--Lam witness $\Omega_C$.  This proves that
$\Omega_C$ is not exposed, but is an extremal point of a four
dimensional exposed face of $\mathcal{S}_1^{\circ}$.

Equation~(\ref{eq:choiwitgen}) with $\theta=0$ and the restrictions
$0\leq a\leq 1$, $a+b+c=2$, $bc=(1-a)^2\!$, defines more generally a one
parameter family of extremal witnesses considered by Ha and
Kye"\cite{Choetal1992,Ha2011}.  They prove that $\Omega_K(a,b,c;0)$ is
both extremal and exposed for $0<a<1$.  The original Choi--Lam witness
$\Omega_C$ is extremal but not exposed, it is the limiting case
$a=c=1$, $b=0$.  We will return in Section~\ref{sec:separable} to a
more detailed discussion of the facial structure of the set of
separable states and the set of witnesses.

We have verified by our numerical methods, for several values of $a$
with $0<a<1$, that $\Omega_K(a,b,c;0)$ is indeed extremal.  As explained
in~\cite{Ha2011} there are four classes of zeros.  The zeros in one of
these classes have Hessians with two dimensional kernels, while the
Hessians of the zeros in the other three classes have one dimensional
kernels.  It turns out that a set of four zeros, one from each class,
uniquely defines $\Omega_K(a,b,c;0)$ as the only solution to the
constraints imposed by the zeros when utilizing both quadratic and
quartic constraints.  This shows numerically that the witness is
extremal.

The more general case with $\theta\neq 0$ has been treated as an
example of optimal, and in fact extremal, witnesses with structural
physical approximations that are entangled PPT states~\cite{Ha2012}.
We will return to these concepts in Section~\ref{sec:spa}.

\subsection{The Robertson witness}

Another example we have studied is the extremal positive map in
dimension $4\times 4$ introduced by Robertson~\cite{Robertson1985}
\begin{equation}
\label{eq:robmap}
X\rightarrow
\begin{pmatrix}
    X_{33}+X_{44} & \phantom{+} & 0 & \phantom{+}  & X_{13}+X_{42} & \phantom{+} & X_{14}-X_{32}\\
    0 & \phantom{+} & X_{33}+X_{44} & \phantom{+} & X_{23}-X_{41} & \phantom{+} & X_{24}+X_{31}\\
    X_{31}+X_{24} & \phantom{+} & X_{32}-X_{14} & \phantom{+} & X_{11}+X_{22} & \phantom{+} & 0\\
    X_{41}-X_{23} & \phantom{+} & X_{42}+X_{13} & \phantom{+} & 0 & \phantom{+} & X_{11}+X_{22}
\end{pmatrix}.
\end{equation}
By equation~(\ref{eq:jamiso01}), it corresponds to the witness
\settowidth{\mycolwd}{$-1$}
\begin{equation}
B=\left(
\begin{array}{*{16}{@{}C{\mycolwd}@{}}}
   . & . & . & . & . & . & . & . & . & . & . & . & . & . & . & .\\
   . & . & . & . & . & . & . & . & . & . & . & 1 & . & . &-1 & .\\
   . & . & 1 & . & . & . & . & . & 1 & . & . & . & . & . & . & .\\
   . & . & . & 1 & . & . & . & . & . & . & . & . & 1 & . & . & .\\
   . & . & . & . & . & . & . & . & . & . & . &-1 & . & . & 1 & .\\
   . & . & . & . & . & . & . & . & . & . & . & . & . & . & . & .\\
   . & . & . & . & . & . & 1 & . & . & 1 & . & . & . & . & . & .\\
   . & . & . & . & . & . & . & 1 & . & . & . & . & . & 1 & . & .\\
   . & . & 1 & . & . & . & . & . & 1 & . & . & . & . & . & . & .\\
   . & . & . & . & . & . & 1 & . & . & 1 & . & . & . & . & . & .\\
   . & . & . & . & . & . & . & . & . & . & . & . & . & . & . & .\\
   . & 1 & . & . &-1 & . & . & . & . & . & . & . & . & . & . & .\\
   . & . & . & 1 & . & . & . & . & . & . & . & . & 1 & . & . & .\\
   . & . & . & . & . & . & . & 1 & . & . & . & . & . & 1 & . & .\\
   . &-1 & . & . & 1 & . & . & . & . & . & . & . & . & . & . & .\\
   . & . & . & . & . & . & . & . & . & . & . & . & . & . & . & .
\end{array}
\right)
\end{equation}
\begin{align}
\label{eq:robbiqform}
\nonumber
f_R(\phi,\chi) & = (|\phi_1|^2+|\phi_2|^2)(|\chi_3|^2+|\chi_4|^2)+(|\phi_3|^2+|\phi_4|^2)(|\chi_1|^2+|\chi_2|^2)\\
%\nonumber
& + 2\,\mathrm{Re}\,[(\phi_1^{\ast}\phi_3+\phi_4^{\ast}\phi_2)(\chi_3^{\ast}\chi_1+\chi_2^{\ast}\chi_4)+(\phi_2^{\ast}\phi_3-\phi_4^{\ast}\phi_1)(\chi_3^{\ast}\chi_2-\chi_1^{\ast}\chi_4)]\;.
\end{align}
Every product vector $\phi\otimes\chi$ with
$\phi_1=\phi_2=\chi_1=\chi_2=0$ or $\phi_3=\phi_4=\chi_3=\chi_4=0$ is
a zero.  More generally, $\phi\otimes\chi$ is a zero if
\begin{equation}
\label{eq:robzeroeq1}
(|\phi_1|^2+|\phi_2|^2)(|\chi_3|^2+|\chi_4|^2)
=(|\phi_3|^2+|\phi_4|^2)(|\chi_1|^2+|\chi_2|^2)
\end{equation} 
and
\begin{eqnarray}
\label{eq:robzeroeq2}
\phi_1^{\ast}\phi_3+\phi_4^{\ast}\phi_2
& = &
-\chi_1^{\ast}\chi_3-\chi_4^{\ast}\chi_2\;,
\nonumber
\\
\phi_2^{\ast}\phi_3-\phi_4^{\ast}\phi_1
& = &
-\chi_2^{\ast}\chi_3+\chi_4^{\ast}\chi_1\;.
\end{eqnarray} 
For any given $\phi$ we obtain a continuum of zeros in the following
way.  Define
\begin{equation}
a=\phi_1^{\ast}\phi_3+\phi_4^{\ast}\phi_2\;,
\qquad
b=\phi_2^{\ast}\phi_3-\phi_4^{\ast}\phi_1\;.
\end{equation} 
Then choose $\chi_1,\chi_2$ at random and define
\begin{eqnarray}
\chi_3=\frac{-a\chi_1-b\chi_2}{|\chi_1|^2+|\chi_2|^2}\;,
\qquad
\chi_4=\frac{ b^{\ast}\chi_1-a^{\ast}\chi_2}{|\chi_1|^2+|\chi_2|^2}\;.
\end{eqnarray} 
This solves equation~(\ref{eq:robzeroeq2}).  In order to solve also
equation~(\ref{eq:robzeroeq1}), rescale $\chi_i\rightarrow c\chi_i$ for
$i=1,2$ and $\chi_i\rightarrow \chi_i/c$ for $i=3,4$ with a suitably chosen
constant $c>0$.

We have verified numerically that $\Omega_R$ is extremal.  Four zeros
determine the witness uniquely through the quadratic and quartic
constraints, \textit{e.g.} the zeros
\begin{equation}
\label{eq:robzeros}
e_{ij}=e_i\otimes e_j,\,\,\, ij=11,12,33,34\;.
\end{equation} 
It appears that all zeros have Hessians with eight dimensional
kernels. The witness has a continuum of zeros, and 20 randomly chosen
of these turn out to also uniquely determine the witness through only
quadratic constraints.  The fact that the quadratic constraints are
sufficient to determine the witness uniquely proves that it is
exposed.  See Section~\ref{sec:separable} for a discussion of exposed
faces.

As a related example we have looked at a witness $\Omega_Z$ belonging
to a new class of witnesses in dimension $N\times(2K)$ introduced by
Zwolak and Chru{\'s}ci{\'n}ski~\cite{Chru3,Zwolak2012,Chru2}.  We take
$N=2$, $K=1$, and $|z_{12}|=1$ in their notation.  Structurally
similar to $\Omega_R$, $\Omega_Z$ also has the zeros defined in
equation~(\ref{eq:robzeros}), and again the quadratic and quartic
constraints from these four zeros determine $\Omega_Z$ uniquely.  This
verifies numerically that $\Omega_Z$ is extremal, and exemplifies the
fact that the same set of zeros with different Hessian zeros can
uniquely specify two different witnesses.  This witness, like
$\Omega_R$, has a continuous set of zeros and can be determined
uniquely through only the quadratic constraints from a randomly chosen
set of 20 of these zeros.  Hence it is exposed.

\subsection{Numerical examples, generic and nongeneric}
\label{sec:newexamples}

We have successfully implemented Algorithm~1 and
used it to locate numerical examples of extremal witnesses in $2\times
4$, $3\times 3$ and $3\times 4$ dimensions.  In the process of
searching for an extremal witness we produce witnesses situated on a
hierarchy of successively lower dimensional faces of $\mathcal{S}_1^{\circ}$.

We define a class of extremal witnesses to be generic if such
witnesses can be found with nonzero probability by means of
Algorithm~1 when the search direction $\Gamma$ is
chosen at random in every iteration.  An overwhelming majority of the
extremal witnesses found in numerical random searches are quadratic,
but a small number of quartic witnesses are also found.  There may be
numerical problems in locating a zero which is quartic or close to
quartic.  In such cases our implementation of the algorithm stops
prematurely.  We discuss the quartic witnesses in
Section~\ref{sec:quartic}.

Line 5 of Algorithm~1 regards locating the
boundary of a face.  \ref{sec:findfacebnd} describes how this
can be formulated as a problem of locating a simple root of a special
function, and also mentions other possible approaches.

\subsubsection{Quadratic extremal witnesses}
\label{sec:quadratic}

We make the following comments concerning the quadratic extremal
witnesses found.

\begin{itemize}

\item[--] We have experienced premature stops due to numerical
  problems with zeros that are close to quartic. See further comments
  in Section~\ref{sec:quartic}.

\item[--] When no existing zero becomes close to quartic, a single new
  quadratic zero appears in every iteration of
  Algorithm~1 when the boundary of the current
  face is reached.  For a small number of zeros, there is no
  redundancy between constraints from the existing zeros and
  constraints from the new zero.  A redundancy appears typically with
  the seventh zero in $2\times 4$, with the ninth zero in $3\times 3$,
  and never in $3\times 4$.  This gives a hierarchy of faces of $\mathcal{S}_1^{\circ}$
  of dimension $N^2-1-kM_2$, where $k$ is the number of zeros of a
  witness in the interior of the face, and $M_2=2(N_a+N_b)-3$.

\item[--] The extremal witnesses have the expected number of zeros as
  listed in Table~\ref{tbl:constraints}.  The zeros and the partially
  conjugated zeros span $\mathcal{H}$, as required by
  Theorem~\ref{thm:zerosspan2}.

\item[--] A quadratic extremal witness has at least one negative
  eigenvalue, and the same is true for its partial transpose.  We do
  not find extremal witnesses with more than three negative eigenvalues in
  $3\times 3$ or more than four negative eigenvalues in $3\times 4$.  See
  Table~\ref{tbl:negeigs} and Figure~\ref{fig:negeigs}.

\item[--] Every witness and its partial transpose have full rank.

\end{itemize}

\begin{figure}
\centering
\includegraphics[width=0.8\textwidth]{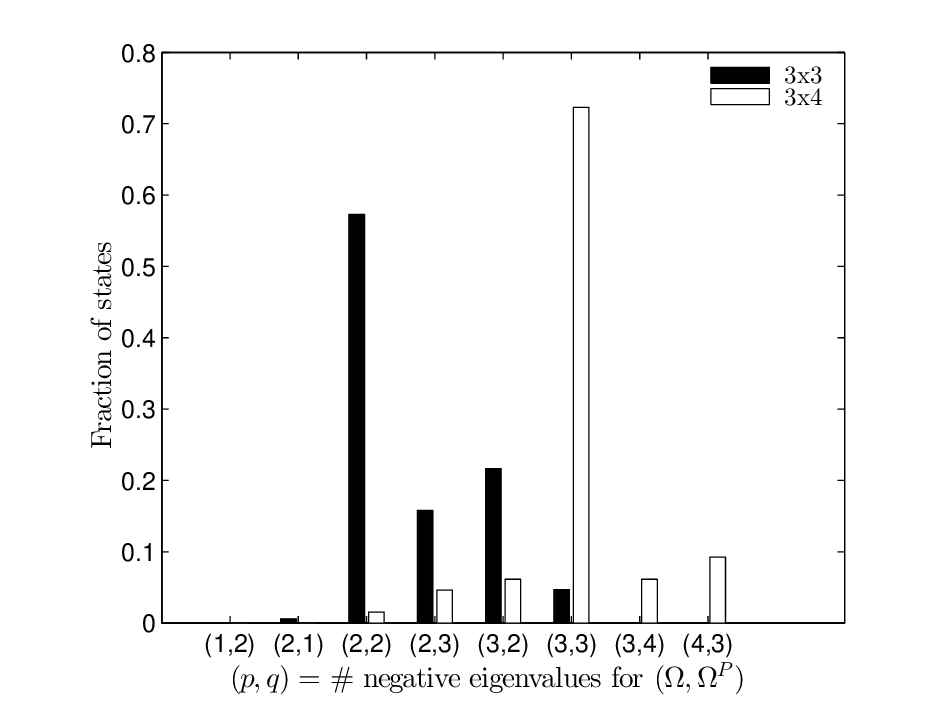}
%\centerline{\psfig{file=negeigs04.eps,width=0.8\textwidth}} 
\caption{Classification of generic quadratic extremal witnesses
found numerically by Algorithm~1 in
dimensions $3\times 3$ and $3\times 4$.  A witness $\Omega$ of
type $(p,q)$ has $p$ negative eigenvalues, and its partial
transpose $\Omega^P$ has $q$ negative eigenvalues.}
\label{fig:negeigs}	
\end{figure}

\begin{table}
\centering
{\begin{tabular}{@{}ccc@{}} 
\hline
    $(p,q)$   & $3\times 3$  & $3\times 4$\\%[0.1cm] 
\hline\\[-0.4cm]
    $(1,2)/(2,1)$       & 0/1 & 0/0 \\	
    $(2,2)$             & 98   &   1  \\
    $(2,3)/(3,2)$       & 27/37   &   3/4 \\
    $(3,3)$             & 8    &   47  \\
    $(3,4)/(4,3)$       & 0/0    &   4/6 \\ 
\hline
\end{tabular}}
  \caption{Classification of generic quadratic extremal witnesses,
    obtained by Algorithm~1,
    in dimensions $3\times 3$ and $3\times 4$.
    A witness and its partial transpose has respectively
    $p$ and $q$ negative eigenvalues.
    The table lists the number of witnesses found.
  }
\label{tbl:negeigs}
\end{table}

\subsubsection{Quartic extremal witnesses}
\label{sec:quartic}

In each iteration of Algorithm~1 the critical
parameter value $t_c$ is reached when either a new quadratic zero
appears, or a new zero eigenvalue of the Hessian matrix appears at one
of the existing zeros.  If the second alternative occurs at least once
during a search, then the extremal witness found will be quartic,
otherwise it will be quadratic.  Our experience is that a random
search most often produces a quadratic extremal witness and only
rarely a quartic extremal witness.

In order to make this observation more quantitative we have made
random searches for quartic witnesses in dimension $3\times 3$ as
follows.  We take 58 hierarchies of faces of quadratic witnesses
generated by Algorithm~1, and generate 100 random
perturbations $\Gamma$ away from the quadratic witness found on each
face. Since there are eight faces in each of the 58  hierarchies, this is a
total of 46400 tests.  For each test we compute $t_c'$ as the smallest
$t$ resulting in a zero eigenvalue of the Hessian at a zero.  $t_c'$
is thus an upper bound for $t_c$.  At $t=t_c'$ we test whether
$\Omega+t\Gamma$ is still a witness, in which case it has got a
quartic zero, or if it is no longer a witness, which will mean that
there has appeared a new quadratic zero for some $t<t_c'$.  Three runs
failed, probably due to some bug in the algorithm.  With a tolerance
of $\pm 10^{-14}$ on function values we found that only 91 out of
46397 successful tests resulted in quartic witnesses.  Thus, in this
quantitative test the probability for finding a quartic witness was
$0.2\,\%$.  We conclude that the fraction of quartic extremal
witnesses among the generic extremal witnesses is small but nonzero.

We constructed an explicit example of a quartic extremal witness in
$3\times 3$ with eight zeros in the following way.  Starting at a
quadratic witness on a face $\mathcal{F}_4$ generated by
Algorithm~1 we found a quartic witness on the
boundary of $\mathcal{F}_4$
by choosing one quadratic zero to become 
quartic and perturbing the quadratic witness accordingly (this may 
not succeed, in which case we choose a different zero).
Continuing Algorithm~1 from
this quartic witness resulted in an extremal witness with eight zeros, of
which one is a quartic zero with a one dimensional kernel of the
Hessian.  Running Algorithm~1 succeeded since we
knew a priori which zero was quartic, so problems with the numerical
precision could be overcome.  We observe that as $M_2=9$ constraints
from each of the six first quadratic zeros are added to $\mathbf{U}_{01}$, and
further as $M_2+M_4(1)=18$ constraints from the quartic zero are added
to form $\mathbf{U}_\Omega$, all of these are linearly independent. Hence
$\mathbf{U}_\Omega$ has rank 72 as expected, and a single new quadratic
zero is located in the final step, adding eight new linearly independent
constraints so as to make $\dim\textrm{Ker}\,\mathbf{U}_\Omega=1$.

\subsubsection{Nongeneric quadratic extremal witnesses}
\label{sec:nongeneric}

In order to demonstrate that a quadratic extremal witness may have
more than the minimum number of zeros, we have constructed quadratic
extremal witnesses in $3\times 3$ dimensions with $10$ rather than the
expected number of nine zeros.

One method for finding such witnesses is illustrated in
Figure~\ref{fig:dshaped3x3}, where the kink in the curve to the right
represents a witness with 10 zeros.  All witnesses on the face have
eight zeros in common, whereas the curved lines on both sides of the
kink consist of quadratic extremal witnesses with nine zeros.  We will
return to this example in the next section.

An entirely different method is described here.  The basic idea is to
use Theorem~\ref{thm:zerosspan2} and try to construct nine product
vectors $\psi_i=\phi_i\otimes\chi_i$ that are linearly dependent, and for
good measure such that also the partially conjugated product vectors
$\widetilde{\psi}_i=\phi_i\otimes\chi_i^*$ are linearly dependent.  Then a
quadratic witness with these nine zeros, and no more, can not be
extremal, but it might lead to a quadratic extremal witness with 10
zeros.

A first attempt is to choose the nine product vectors directly, by
minimizing the sum of the smallest singular value of the $9\times 9$
matrix
$\psi={[\psi_1,\psi_2,\ldots,\psi_9]}$ and the smallest singular value of 
the corresponding $9\times 9$ matrix $\widetilde{\psi}$.
Since singular values are nonnegative by definition, minimization
will make both these singular values zero.  The minimization problem
is solved \textit{e.g.} by the Nelder--Mead algorithm, or by a random search.

Let $Z$ be a set of nine product vectors generated in this way,
then any witness $\Omega$ with these as zeros has to satisfy the
constraints $\mathbf{U}_{01}\Omega=0$.  The generic result we find is that
$\dim\textrm{Ker}\,\mathbf{U}_{01}=2$.  Hence $\Omega$ has to be a decomposable witness,
\begin{equation}
\Omega=p\,\eta\eta^{\dagger}
+(1-p)\,(\widetilde{\eta}\widetilde{\eta}^{\dagger})^P\!,
\end{equation}
with $\eta,\widetilde{\eta}\in\mathcal{H}$,
$\eta^{\dagger}\psi_i=\widetilde{\eta}^{\dagger}\widetilde{\psi}_i=0$
for $i=1,\ldots,9$, and $0\leq p\leq 1$.  This decomposable $\Omega$
is not what we are looking for, in fact it has a continuum of quartic
zeros.  The extremal decomposable witnesses $\eta\eta^{\dagger}$ and
$(\widetilde{\eta}\widetilde{\eta}^{\dagger})^P$ each contribute two
nonzero eigenvalues to the $8\times 8$ Hessian matrix $G_{\Omega}$ at
every zero, hence $\dim\textrm{Ker}\,G_{\Omega}=4$ at every zero, unless
$p=0$ or $p=1$ in which case $\dim\textrm{Ker}\,G_{\Omega}=6$.

A second attempt is to choose eight product vectors $\psi_i$ that are
linearly dependent and have linearly dependent partial conjugates
$\widetilde{\psi}_i$.  This almost works, but not quite.  We may
construct a decomposable witness $\Omega$ with these zeros from two
pure states orthogonal to all $\psi_i$ and two partially transposed
pure states orthogonal to all $\widetilde{\psi}_i$.  As a convex
combination of four extremal decomposable witnesses, each contributing
two nonzero eigenvalues to each $8\times 8$ Hessian matrix, $\Omega$
will be quadratic.  $\Omega$ is determined by eight real parameters,
including a normalization constant, since $\Omega=\rho+\sigma^P$
where $\rho$ and $\sigma$ are positive matrices of rank two, with
$\rho\psi_i=\sigma\widetilde{\psi}_i=0$.  The maximum number of
independent constraints from eight quadratic zeros is 72, giving
$\dim\textrm{Ker}\,\mathbf{U}_{01}=9$, one more than the dimension of the set of
decomposable witnesses having these zeros.  However, we find that
$\Omega$ generically has more than the eight prescribed zeros, increasing
the number of independent constraints to 73.  Hence, again the
constraints leave only the decomposable witnesses, none of which are
quadratic and extremal.

The third and successful attempt is to choose seven product vectors with
similar linear dependencies.  Denote by $\mathbf{U}_{01}^{(7)}$ the
corresponding linear system of constraints.  The kernel of
$\mathbf{U}_{01}^{(7)}$ is found, in two different cases, to have dimension
21 or 22.  A decomposable witness with the given product vectors as
zeros has the form $\Omega_7=\rho+\sigma^P$ where $\rho$ and $\sigma$
are positive matrices of rank three.  Hence, the set of such
decomposable witnesses is $18$ dimensional, so that there are three or four
dimensions in $\textrm{Ker}\,\mathbf{U}_{01}^{(7)}$ orthogonal to the face of
decomposable witnesses.  Defining $\Gamma_7$ to lie in these three or four
dimensions one can walk towards the boundary of the face
$\mathcal{F}_7=(\textrm{Ker}\,\mathbf{U}_{01}^{(7)})\cap\mathcal{S}_1^{\circ}$, finding
$\Omega_8=\Omega_7+t_c\Gamma_7$ with eight zeros.  $\Omega_8$ is now
guaranteed to be nondecomposable.  Let $\mathbf{U}_{01}^{(8)}$ be the system
defined by these eight zeros, defining the face
$\mathcal{F}_8=(\textrm{Ker}\,\mathbf{U}_{01}^{(8)})\cap\mathcal{S}_1^{\circ}$.  We find that
$\textrm{Ker}\,\mathbf{U}_{01}^{(8)}$ has dimension nine less than $\textrm{Ker}\,\mathbf{U}_{01}^{(7)}$.
Defining $\Gamma_8\in\textrm{Ker}\,\mathbf{U}_{01}^{(8)}$, we locate an
$\Omega_9\in\mathcal{F}_9$, on the boundary of $\mathcal{F}_8$, with nine zeros.  The
kernel of $\mathbf{U}_{01}^{(9)}$ has dimension nine less than
$\textrm{Ker}\,\mathbf{U}_{01}^{(8)}$, \textit{i.e.} three or four, hence there is still freedom to move
along $\mathcal{F}_9$.  Doing so produces a quadratic extremal witness in
$3\times 3$ with $10$ zeros rather than nine.

%%%%%%%%%%%%%%%%%%%%%%%%%%%%%%%%%%%%%%%%%%%%%%%%%%%%%%%%%%%%%%%%%%%%%%%%%%%%%%%%%%%%%%%%%%%%%%%%%%%%%%%%%%%%%%%%%%%
%%%%%%%%%%%%%%%%%%%%%%%%%%%%%%%%%%%%%%%%%%%%%%%%%%%%%%%%%%%%%%%%%%%%%%%%%%%%%%%%%%%%%%%%%%%%%%%%%%%%%%%%%%%%%%%%%%%

\section{D-shaped faces of the set of witnesses in low dimensions}
\label{sec:dshapedI}

In this section we reveal a special geometry of next-to-extremal faces
of $\mathcal{S}_1^{\circ}$ in $2\times 4$ and $3\times 3$ systems, related to the
presence of decomposable witnesses.

Let $\mathcal{F}_k$ denote a face of $\mathcal{S}_1^{\circ}$ with interior points that are
quadratic witnesses with $k$ zeros.  This is typically what we find in
the $k$-th iteration of Algorithm~1.  A face
$\mathcal{F}_7$ in dimension $2\times 4$, or $\mathcal{F}_8$ in dimension $3\times 3$,
is the last face found before an extremal quadratic witness is
reached.  These particular faces have a special geometry, because the
number of zeros is one less than the dimension of the Hilbert space,
and as a result part of the boundary is a line segment of decomposable
witnesses.

A decomposable witness on such a face has the form
\begin{equation}
\Omega=(1-p)\,\psi\psi^\dagger+p\,(\eta\eta^\dagger)^P\!,
\quad 0\leq p\leq 1\,,
\end{equation}
where $\psi$ is orthogonal to the $N-1$ product vectors
$\phi_i\otimes\chi_i$ that are the zeros of all the witnesses in the
interior of the face, and $\eta$ is orthogonal to the partially
conjugated product vectors $\phi_i\otimes\chi_i^*$.  We take $\psi$ and
$\eta$ to be normalized vectors,
$\psi^\dagger\psi=\eta^\dagger\eta=1$.

\begin{figure}
\centering
\includegraphics[width=0.99\textwidth]{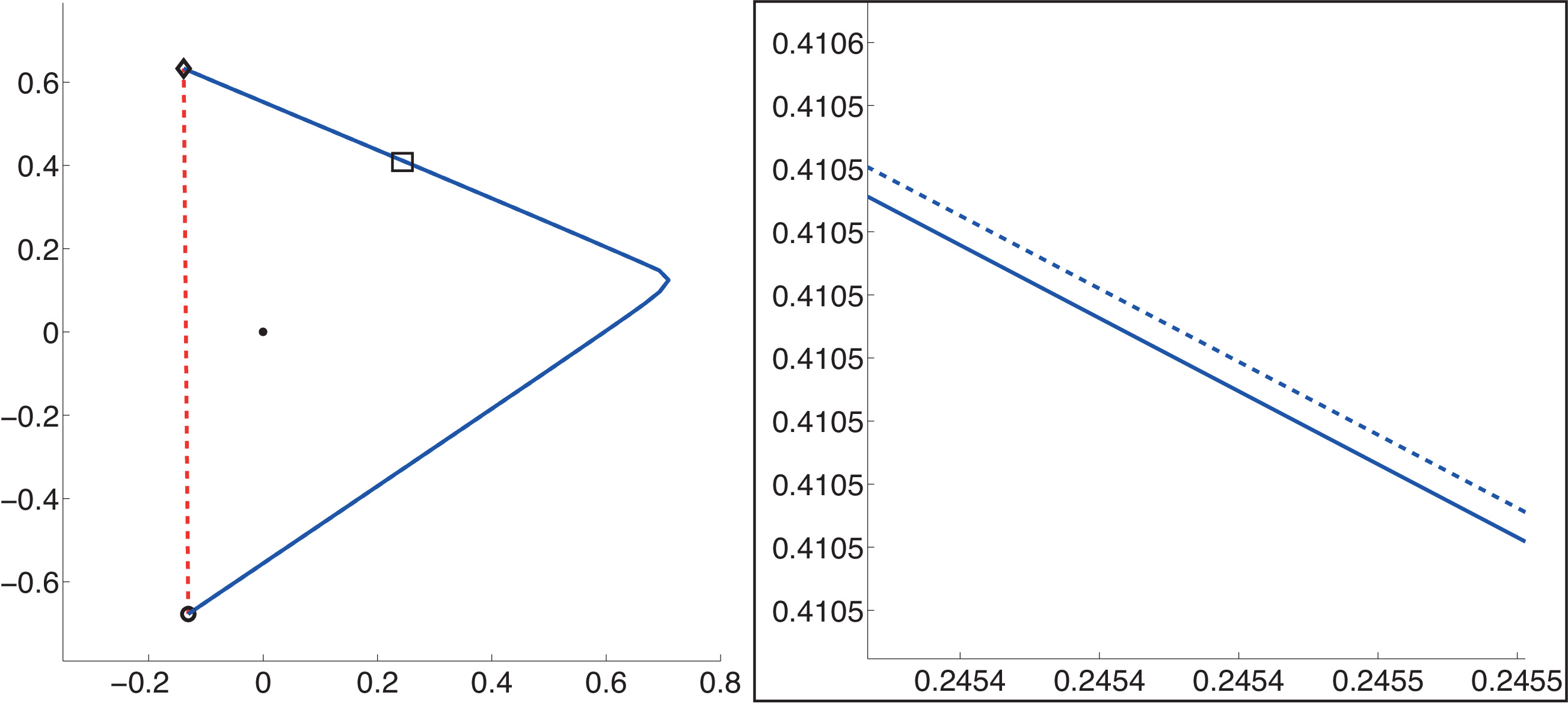}
%\centerline{\psfig{file=dface2x4.eps,width=0.99\textwidth}} 
\caption{A two dimensional face $\mathcal{F}_7\subset\mathcal{S}_1^{\circ}$ obtained by
applying Algorithm~1 in dimension
$2\times 4$.  Distances are defined by the Hilbert--Schmidt
metric.  The straight line segment (dashed, in red) consists of
decomposable witnesses, its upper end point is a pure state,
marked by a diamond, and its lower end point is the partial
transpose of a pure state, marked by a circle.  Starting from a
quadratic witness with seven zeros, marked by the black dot, the
curved boundary of the face (in blue) was located by perturbing in
all directions.  Every point on this curve represents an extremal
witness with eight quadratic zeros, except at the kink in the
middle of the curved boundary where the witness has nine quadratic
zeros.  These extremal witnesses are very nearly quartic.  We show
this by drawing another curved line (dashed, in blue) where the
first Hessian zero appears, that is, where one zero becomes
quartic.  This dashed curve is only visible in the enlarged part
of the figure.}
\label{fig:dshaped2x4}
\end{figure}

When we apply Algorithm~1 in $2\times 4$
dimensions and find a quadratic extremal witness, the generic case is
that the face $\mathcal{F}_7$ is two dimensional.  An example is shown in
Figure~\ref{fig:dshaped2x4}.  The line segment of decomposable
witnesses must be part of the boundary of the face $\mathcal{F}_7$, because the
interior of the face consists of quadratic witnesses with a fixed set
of seven quadratic zeros, whereas the decomposable witnesses have
additional zeros and Hessian zeros, in fact infinitely many quartic
zeros.  The rest of the boundary of the face is curved, and consists
of quadratic extremal witnesses with eight zeros.  There is one
exception, however,  seen in the figure as a kink
in the curved part of the boundary, and  this is a quadratic extremal
witness with nine zeros.

The interesting explanation is as follows.  As we go along the curve
consisting of quadratic witnesses with eight zeros, seven zeros are fixed,
they are the zeros defining the face.  But the eighth zero has to change
along the curve, because any two points on this part of the boundary
can be joined by a line segment passing through the interior of the
face, hence these two boundary points can have only the seven zeros in
common.  Starting from the two extremal decomposable witnesses
$\psi\psi^\dagger$ and $(\eta\eta^\dagger)^P$ we get two curved sections
of the boundary where the eighth zero changes continuously.  These two
sections meet in one point which is then a witness with two quadratic
zeros in addition to the seven zeros defining the face.

Accordingly, this next-to-extremal face $\mathcal{F}_7$ is two dimensional and
has the shape of a ``D'', where the straight edge is the line segment
of decomposable witnesses and the round part consists of quadratic
extremal witnesses.  Figure~\ref{fig:dshaped2x4} is a numerically
produced example of such a face.  Along similar lines we expect that
next-to-extremal faces in any $2\times N_b$ systems will be either
D-shaped, or line segments if by accident we hit the straight edge of
the D.

In the case of $\mathcal{F}_8$ in $3\times 3$ we can also construct the line
segment of decomposable witnesses on the boundary of $\mathcal{F}_8$.  The
remaining boundary will again consist of quadratic extremal witnesses,
this time with nine or ten zeros.  This part of the boundary is a curved seven
dimensional surface, since the face itself is eight dimensional.  Any two
dimensional section of $\mathcal{F}_8$ passing through the line segment of
decomposable witnesses is shaped as a D.  See
Figure~\ref{fig:dshaped3x3} for a numerically computed example.

\begin{figure}
\centering
\includegraphics[width=0.8\textwidth]{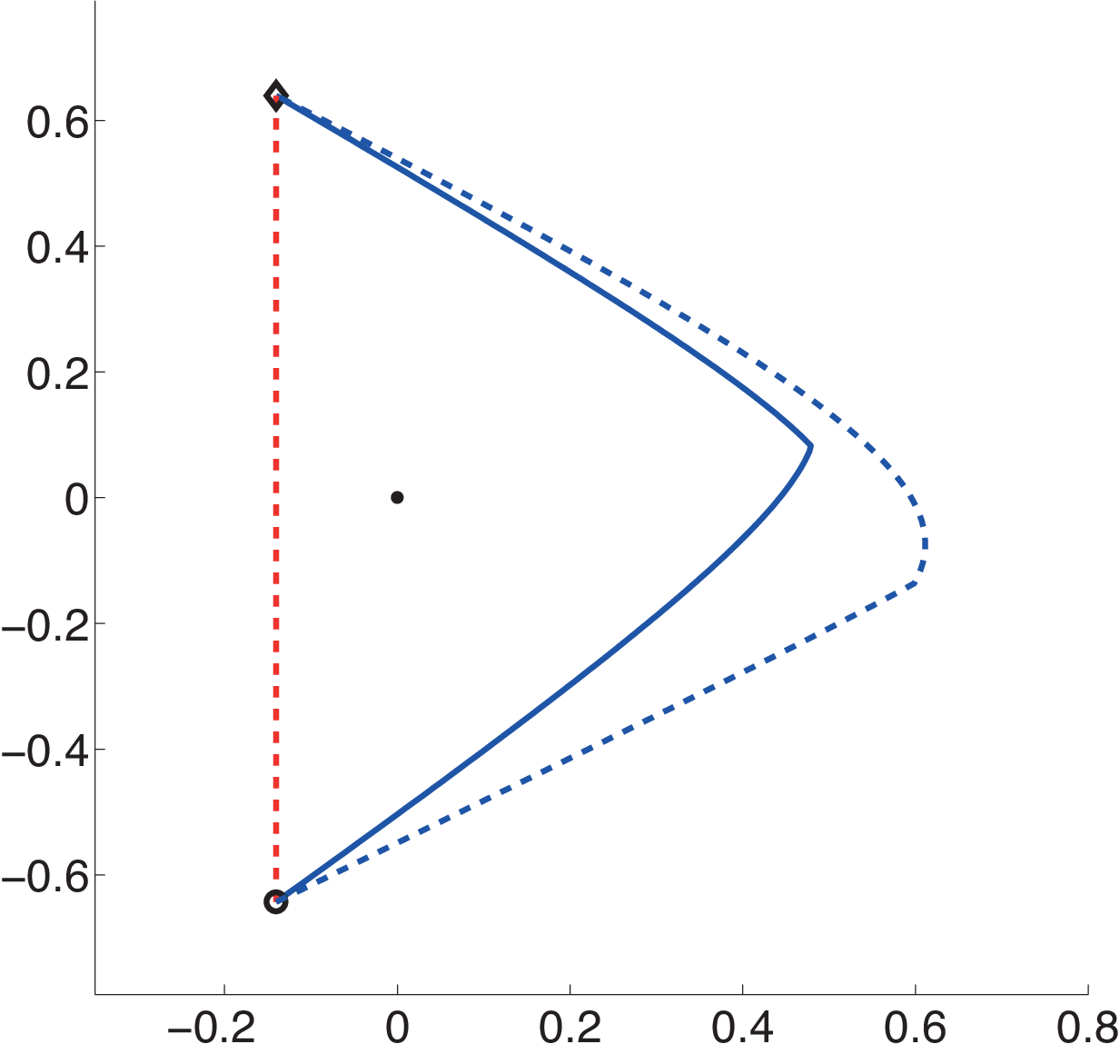}
%\centerline{\psfig{file=Dface3x3new.eps,width=0.8\textwidth}} 
\caption{A special two dimensional section of an eight dimensional
face $\mathcal{F}_8\subset\mathcal{S}_1^{\circ}$ obtained from output of
Algorithm~1 in dimension $3\times 3$.  This
section passes through a straight line segment of decomposable
witnesses (dashed, in red) which is part of the boundary of the
face.  The upper end point of the line segment is a pure state,
marked by a diamond, and the lower end point is the partial
transpose of a pure state, marked by a circle.  Starting from a
quadratic witness with eight zeros, marked by the black dot, the
curved boundary of the face (in blue) was located by perturbing in
all directions in the plane.  Every point on this curve represents
an extremal witness with nine quadratic zeros, except at the kink
where the witness has ten zeros.  The curved dashed line (in blue)
is where one zero becomes quartic.}
\label{fig:dshaped3x3}
\end{figure}

Note that we can not guarantee that any choice of eight random product
vectors gives rise to such a D.  Many choices may give rise to only
the line segment of decomposable witnesses, since there are inequality
constraints that are not automatically satisfied, even if we are able
to satisfy the equality constraints that we have discussed here.

In other dimensions line segments of decomposable witnesses
constructed from $N-1$ zeros also define D-shaped faces, but the
round part of the ``D'' will in those cases consist not of extremal
witnesses but of lower dimensional faces.

%%%%%%%%%%%%%%%%%%%%%%%%%%%%%%%%%%%%%%%%%%%%%%%%%%%%%%%%%%%%%%%%%%%%%%%%%%%%%%%%%%%%%%%%%%%%%%%%%%%%%%%%%%%%%%%%%
%%%%%%%%%%%%%%%%%%%%%%%%%%%%%%%%%%%%%%%%%%%%%%%%%%%%%%%%%%%%%%%%%%%%%%%%%%%%%%%%%%%%%%%%%%%%%%%%%%%%%%%%%%%%%%%%%

\section{Faces of the set of separable states}
\label{sec:separable}

Our understanding of witnesses as exposed in
Section~\ref{sec:extremal} translates into an understanding of faces
of $\mathcal{S}_1$, the set of separable states.  A face of a compact convex set
is defined by the extremal points it contains, and in the case of
$\mathcal{S}_1$ the extremal points are the pure product states.  Alfsen and
Shultz~\cite{Alfsen2010} describe faces of $\mathcal{S}_1$ in two categories,
either simplexes defined by at most $\max(N_a,N_b)$ pure product
states, or direct convex sums of faces isomorphic to matrix algebras.
According to our understanding these two categories correspond to
quadratic and certain quartic witnesses respectively.  It has been
known for some time that the set of entanglement witnesses has
unexposed faces.  These questions have been studied by
Chru{\'s}ci{\'n}ski and collaborators~\cite{Chru3,Chru4}.  In our
understanding the unexposed faces of $\mathcal{S}_1^{\circ}$ are those containing
quartic witnesses as interior points.  The facial structure of various
sets related to quantum entanglement has been studied especially by
Kye and collaborators~\cite{Kye2013}.

In this section we describe faces of $\mathcal{S}_1$ defined by different types
of witnesses.  We state some basic results which may be well known and
which are actually quite generally valid for any pair of dual convex
cones.  The distinction between exposed and unexposed faces is
central, and it would be interesting to know whether all the faces of
the set of separable states are exposed.  We believe that this is
true, although we have no proof.  Some support for our conjecture may
be drawn from the known facts that both the set of density matrices
and the set of PPT states have only exposed faces.

To finish this section we point out how the facial structure of the
set of separable states is related to a question which is of practical
importance when we want to test whether a given state is separable.
The question is how many pure product states we need, if the state is
separable, in order to write it as a convex combination of pure
product states.

\subsection{Duality of faces}

Given any subset $\mathcal{X}\subset\mathcal{S}_1$ we define its dual in
$\mathcal{S}_1^{\circ}$ as
\begin{equation}
\mathcal{X}^{\circ}=\{\,\Omega\in\mathcal{S}_1^{\circ} \mid \textrm{Tr}\,\Omega\rho=0
\;\;\forall\rho\in\mathcal{X}\,\}\;.
\end{equation}
Similarly, given any subset $\mathcal{Y}\subset\mathcal{S}_1^{\circ}$ we define its
dual in $\mathcal{S}_1$ as
\begin{equation}
\mathcal{Y}^{\circ}=\{\,\rho\in\mathcal{S}_1 \mid \textrm{Tr}\,\Omega\rho=0
\;\;\forall\Omega\in\mathcal{Y}\,\}\;.
\end{equation}
We will assume here that $\mathcal{X}^{\circ}$ and
$\mathcal{Y}^{\circ}$ are nonempty, a minimum requirement is that
$\mathcal{X}\subset\partial\mathcal{S}_1$ and $\mathcal{Y}\subset\partial\mathcal{S}_1^{\circ}$.  Then also
$\mathcal{X}^{\circ}\subset\partial\mathcal{S}_1^{\circ}$ and
$\mathcal{Y}^{\circ}\subset\partial\mathcal{S}_1$.

There always exists one single $\rho_0\in\partial\mathcal{S}_1$ such that
\begin{equation}
\mathcal{X}^{\circ}
=\{\rho_0\}^{\circ}
=\{\,\Omega\in\mathcal{S}_1^{\circ} \mid \textrm{Tr}\,\Omega\rho_0=0\,\}\;.
\end{equation}
In fact, every $\rho\in\mathcal{X}$ gives one linear constraint
$\textrm{Tr}\,\Omega\rho=0$ as part of the definition of $\mathcal{X}^{\circ}\!$.
In finite dimension at most a finite number of linear constraints can
be independent, hence
\begin{equation}
\mathcal{X}^{\circ}=\{\,\Omega\in\mathcal{S}_1^{\circ} \mid \textrm{Tr}\,\Omega\rho_i=0
\;\;\mathrm{for}\;\;i=1,2,\ldots,k\,\}
\end{equation}
for some states $\rho_1,\rho_2,\ldots,\rho_k\in{\mathcal{X}}$.  Define for
example
\begin{equation}
\rho_0=\frac{1}{k}\sum_{i=1}^k\rho_i\;.
\end{equation}
If ${\mathcal{X}}$ is not
convex, it may happen that $\rho_0\notin{\mathcal{X}}$.  Because $\textrm{Tr}\,\Omega\rho_i\geq 0$ for $i=1,2,\ldots,k$ the
equation $\textrm{Tr}\,\Omega\rho_0=0$ implies that $\textrm{Tr}\,\Omega\rho_i=0$ for
$i=1,2,\ldots,k$ and hence $\Omega\in{\mathcal{X}}^{\circ}\!$.

Taking three times the dual we get that
${\mathcal{X}}^{\circ}=\mathcal{F}^{\circ}$ where
$\mathcal{F}={\mathcal{X}}^{\circ\circ}$ is the double dual of ${\mathcal{X}}$.
Clearly $\mathcal{F}$ contains ${\mathcal{X}}$, and by our next theorem $\mathcal{F}$ is
an exposed face of $\mathcal{S}_1$.

Reasoning in the same way we conclude that it is always possible to
find one single witness $\Omega_0$ such that
\begin{equation}
\mathcal{Y}^{\circ}
=\{\Omega_0\}^{\circ}
=\{\,\rho\in\mathcal{S}_1 \mid \textrm{Tr}\,\Omega_0\rho=0\,\}\;,
\end{equation}
and an exposed face $\mathcal{G}$ of $\mathcal{S}_1^{\circ}$ containing ${\mathcal{Y}}$, in fact
$\mathcal{G}=\mathcal{Y}^{\circ\circ}\!$, such that
$\mathcal{Y}^{\circ}={\mathcal{G}}^{\circ}\!$.
\begin{theorem}
\label{th:exposeddualsI}
\textit{$\mathcal{X}^{\circ}$ is an exposed face of $\mathcal{S}_1^{\circ}$, and
$\mathcal{Y}^{\circ}$ is an exposed face of $\mathcal{S}_1$.}
\end{theorem}
\begin{proof}
Assume that $\Omega\in\mathcal{X}^{\circ}$ is a proper convex
combination of $\Omega_1,\Omega_2\in\mathcal{S}_1^{\circ}$,
\begin{equation}
\Omega=(1-p)\,\Omega_1+p\,\Omega_2
\qquad\mathrm{with}\qquad
0<p<1\;.
\end{equation}
We have to prove that $\Omega_1,\Omega_2\in\mathcal{X}^{\circ}\!$.  The
assumption that $\Omega_1,\Omega_2\in\mathcal{S}_1^{\circ}$ means that
$\textrm{Tr}\,\Omega_1\rho\geq 0$ and $\textrm{Tr}\,\Omega_2\rho\geq 0$ for every
$\rho\in\mathcal{S}_1$.  For every $\rho\in\mathcal{X}$ we have in addition that
\begin{equation}
0=\textrm{Tr}\,\Omega\rho=(1-p)\,\textrm{Tr}\,\Omega_1\rho+p\,\textrm{Tr}\,\Omega_2\rho\;,
\end{equation}
and from this we conclude that $\textrm{Tr}\,\Omega_1\rho=\textrm{Tr}\,\Omega_2\rho=0$.
This proves that $\Omega_1,\Omega_2\in{\mathcal{X}}^{\circ}\!$, so that
${\mathcal{X}}^{\circ}$ is a face.

It is exposed because it is dual to one single $\rho_0\in\mathcal{S}_1$.  In
fact, the equation $\textrm{Tr}\,\Lambda\rho_0=0$ for $\Lambda$ defines a
hyperplane of dimension $N^2-2$ in the $N^2-1$ dimensional affine
space of Hermitian $N\times N$ matrices of unit trace.

The proof that $\mathcal{Y}^{\circ}$ is an exposed face is entirely
similar.
\end{proof}
This theorem has the following converse.
\begin{theorem}
\label{th:exposeddualsII}
\textit{An exposed face of $\mathcal{S}_1^{\circ}$ is the dual of a separable state, and an
exposed face of $\mathcal{S}_1$ is the dual of a witness.}
\end{theorem}
\begin{proof}
We prove the second half of the theorem, the first half is proved in
a similar way.

An exposed face $\mathcal{F}$ of $\mathcal{S}_1$ is the intersection of $\mathcal{S}_1$ with a
hyperplane given by an equation for $\rho\in H$ of the form
$\textrm{Tr}\,\Lambda\rho=0$ where $\Lambda\in H$ is fixed.  The maximally
mixed state $\rho_0=I/N$ is an interior point of $\mathcal{S}_1$, hence
$\textrm{Tr}\,\Lambda=N\textrm{Tr}\,\Lambda\rho_0\neq 0$ and we may impose the
normalization condition $\textrm{Tr}\,\Lambda=1$.  We have to prove that
$\Lambda\in\mathcal{S}_1^{\circ}$, which means that $\textrm{Tr}\,\Lambda\rho\geq 0$ for all
$\rho\in\mathcal{S}_1$.

Assume that there exists some $\rho_1\in\mathcal{S}_1$ with
$\textrm{Tr}\,\Lambda\rho_1<0$.  Choose any $\rho_2\in\mathcal{S}_1$ with
$\textrm{Tr}\,\Lambda\rho_2>0$, for example $\rho_2=\rho_0$.  Then
$\rho_1,\rho_2\not\in\mathcal{F}$, since $\mathcal{F}$ is defined as the set of all
$\rho\in\mathcal{S}_1$ having $\textrm{Tr}\,\Lambda\rho=0$.  But there exists a proper
convex combination $\rho=(1-p)\rho_1+p\rho_2$ with $0<p<1$ such that
$\textrm{Tr}\,\Lambda\rho=0$, and hence $\rho\in\mathcal{F}$, contradicting the
assumption that $\mathcal{F}$ is a face of $\mathcal{S}_1$.
\end{proof}

We summarize Theorems~\ref{th:exposeddualsI}
and~\ref{th:exposeddualsII} as follows.
\begin{theorem}
\label{th:exposeddualsIII}
\textit{There is a one to one correspondence between exposed faces of $\mathcal{S}_1$
and exposed faces of $\mathcal{S}_1^{\circ}$.  The faces in each pair are dual
(orthogonal) to each other.}
\end{theorem}

Since the extremal points of $\mathcal{S}_1$ are the pure product states, by
Theorem~\ref{thm:extrpface} the extremal points of a face $\mathcal{F}\in\mathcal{S}_1$
are the pure product states contained in $\mathcal{F}$.  It follows that the
extremal points of the face $\mathcal{Y}^{\circ}$ are the common zeros
of all the witnesses in $\mathcal{Y}$.  We have actually proved the
following result.
\begin{theorem}
\label{th:exposeddualsIV}
\textit{A set of product vectors in $\mathcal{H}$ is the complete set of zeros of
some witness if and only if they are the extremal points of an
exposed face of $\mathcal{S}_1$ when regarded as states in $\mathcal{S}_1$.}
\end{theorem}

The remaining question, to which we do not know the answer, is whether
$\mathcal{S}_1$ has unexposed faces.  We state the following theorem, which
actually holds not only for $\mathcal{S}_1$ but for any compact convex set.
\begin{theorem}
\label{th:exposedS1}
\textit{Every proper face of $\mathcal{S}_1$ is  contained in an exposed face of $\mathcal{S}_1$.}
\end{theorem}
\begin{proof}
Given a face $\mathcal{F}$ of $\mathcal{S}_1$, we have to prove that there exists a
witness $\Omega$ such that $\mathcal{F}$ is contained in the dual face
$\Omega^{\circ}$.

Choose $\rho\in\mathcal{F}$, in the interior of $\mathcal{F}$ if $\mathcal{F}$ contains more
than one point.  Choose also a separable state $\sigma\notin\mathcal{F}$,
and define $\tau=(1+t)\rho-t\sigma$.  Since $\mathcal{F}$ is a face, and
$\tau\in\mathcal{S}_1$ for $-1\leq t\leq 0$, we know that $\tau\not\in\mathcal{S}_1$ for
every $t>0$, and the set
\begin{equation}
\mathcal{Y}(t)=\{\,\Lambda\in\mathcal{S}_1^{\circ}g\mid\textrm{Tr}\,\Lambda\tau\leq 0\,\}
\end{equation}
is nonempty for every $t>0$.  Every $\mathcal{Y}(t)$ is a compact set,
and $\mathcal{Y}(t_1)\subset\mathcal{Y}(t_2)$ for $0<t_1<t_2$.  Hence,
the intersection of all sets $\mathcal{Y}(t)$ for $t>0$ is nonempty
and contains at least one witness $\Omega$ such that
$\textrm{Tr}\,\Omega\tau\leq 0$ for every $t>0$.  Clearly we must then have
$\textrm{Tr}\,\Omega\rho=0$.  Since $\textrm{Tr}\,\Omega\rho=0$ for one point $\rho$ in
the interior of the face $\mathcal{F}$, it follows that $\textrm{Tr}\,\Omega\rho=0$ for
every $\rho\in\mathcal{F}$.
\end{proof}

Unfortunately, this does not amount to a proof that $\mathcal{F}$ is exposed,
because it might happen that $\textrm{Tr}\,\Omega\sigma=0$ even though
$\sigma\notin\mathcal{F}$.

\subsection{Faces of $\mathcal{D}_1$ and $\mathcal{P}_1$}

When discussing faces of $\mathcal{S}_1$, the set of separable states, it may be
illuminating to consider the simpler examples of faces of $\mathcal{D}_1$, the
set of density matrices, and faces of $\mathcal{P}_1$, the set of PPT states.
Recall that $\mathcal{D}_1$ is selfdual, $\mathcal{D}_1^{\circ}=\mathcal{D}_1$.  The fact that all faces of
$\mathcal{D}_1$ and $\mathcal{P}_1$ are exposed may indicate that the same is true for all
faces of $\mathcal{S}_1$.

A face $\mathcal{F}$ of $\mathcal{D}_1$ is a complete set of density matrices on a
subspace $\mathcal{U}\subset\mathcal{H}$, thus there is a one to one correspondence
between faces of $\mathcal{D}_1$ and subspaces of $\mathcal{H}$.  A density matrix $\rho$
belongs to the face $\mathcal{F}$ when $\textrm{Img}\,\rho\subset\mathcal{U}$, and it is an
interior point of $\mathcal{F}$ when $\textrm{Img}\,\rho=\mathcal{U}$.

It is straightforward to show that when $\rho,\sigma\in\mathcal{D}_1$ we have
$\textrm{Tr}\,\rho\sigma=0$ if and only if $\textrm{Img}\,\rho\perp\textrm{Img}\,\sigma$.  Hence,
the dual, or opposite, face $\mathcal{F}^{\circ}$ is the set of density
matrices on the subspace $\mathcal{U}^{\perp}\!$, the orthogonal complement of
$\mathcal{U}$.  The double dual of $\mathcal{F}$ is $\mathcal{F}$ itself, $\mathcal{F}^{\circ\circ}=\mathcal{F}$,
since $(\mathcal{U}^{\perp})^{\perp}=\mathcal{U}$.

Every proper face $\mathcal{F}$ of $\mathcal{D}_1$ is exposed, since it is the dual of an
arbitrarily chosen interior point $\sigma\in\mathcal{F}^{\circ}$.

The definition $\mathcal{P}_1=\mathcal{D}_1\cap\mathcal{D}_1^P$ implies, by
Theorem~\ref{thm:convxints}, that every face $\mathcal{G}$ of $\mathcal{P}_1$ is an
intersection $\mathcal{G}=\mathcal{E}\cap\mathcal{F}^P\!$, where $\mathcal{E}$ and $\mathcal{F}$ are faces of $\mathcal{D}_1$.
This is the geometrical meaning of the procedure for finding extremal
PPT states introduced in~\cite{pptII}.

It follows that every face of $\mathcal{P}_1$ is exposed.  In fact, the face
$\mathcal{G}$ is dual to a decomposable witness $\rho+\sigma^P$ where
$\rho,\sigma\in\mathcal{D}$, such that $\mathcal{E}$ is dual to $\rho$ and $\mathcal{F}$ is dual
to $\sigma$.

\subsection{Unexposed faces of $\mathcal{S}_1^{\circ}$}

The Choi--Lam witness $\Omega_C$, as given in
equation~\eqref{eq:choiwitgen} with $\theta=0$, $a=c=1$, $b=0$, is an
example of an extremal witness, a zero dimensional face of $\mathcal{S}_1^{\circ}$,
which is not exposed.  The three isolated product vectors defined in
equation~\eqref{eq:choiisolzeros} and the continuum of product vectors
defined in equation~\eqref{eq:choicontzeros}, interpreted as states in
$\mathcal{S}_1$, are the extremal points of the dual face
$\mathcal{F}_C=\{\Omega_C\}^{\circ}\subset\mathcal{S}_1$.  One may check both numerically
and analytically that the face $\mathcal{F}_C$ has dimension 21.  The separable
states in the interior of $\mathcal{F}_C$ have rank seven, since they are
constructed from product vectors in a seven dimensional subspace.

We find numerically that the constraints $\mathbf{U}_0$ and $\mathbf{U}_1$
associated with the zeros of $\Omega_C$ define a four dimensional face
of $\mathcal{S}_1^{\circ}$.  This is then the dual face $\mathcal{F}_C^{\circ}$, the double dual
of $\Omega_C$.  The last four constraints needed to prove that
$\Omega_C$ is extremal come from the constraints $\mathbf{U}_2$ and $\mathbf{U}_3$
expressing the quartic nature of the zeros.

This is one example showing the mechanism for how faces of $\mathcal{S}_1^{\circ}$ may
avoid being exposed.  In general, a witness $\Omega$ having one or
more isolated quartic zeros will be an interior point of an unexposed
face.  This unexposed face is then a face of a larger exposed face
consisting of witnesses having the same zeros as $\Omega$, but such
that all the isolated zeros are quadratic.

\subsection{Simplex faces of $\mathcal{S}_1$}

Theorem~\ref{th:exposeddualsIV} expresses the relation between zeros
of entanglement witnesses and exposed faces of $\mathcal{S}_1$, the set of
separable states.  We do not know whether $\mathcal{S}_1$ has unexposed faces.
We consider first faces that are simplexes, having only a finite
number of extremal points.

\subsubsection{Faces of $\mathcal{S}_1$ dual to quadratic extremal witnesses}

The exposed faces of $\mathcal{S}_1$ defined by extremal witnesses with only
quadratic zeros are simplexes, and as such are particularly simple to
study.  Given the product vectors $\psi_i=\phi_i\otimes\chi_i$ for
$i=1,2,\ldots,k$, with $k=n+1$, as the zeros of a quadratic extremal
witness $\Omega$.  The corresponding product states
$\rho_i=\psi_i\psi_i^{\dagger}$ are the vertices of an $n$-simplex, which
is the exposed face $\Omega^{\circ}$ dual to $\Omega$.

Let $\rho$ be an interior point of this face,
\begin{equation}
\rho=\sum_{i=1}^kp_i\,\rho_i\,,
\qquad \sum_{i=0}^kp_i=1\,,
\qquad p_i>0\,.
\end{equation}
According to Theorem~\ref{thm:zerosspan2}, both $\rho$, constructed
from the zeros $\phi_i\otimes\chi_i$, and its partial transpose $\rho^P\!$,
constructed in the same way from the partially conjugated zeros
$\phi_i\otimes\chi_i^{\ast}$, have full rank $N=N_{a}N_{b}$.  Thus,
$\rho$ lies not only in the interior of $\mathcal{D}_1$, the set of density
matrices, but also in the interior of $\mathcal{P}_1$, the set of PPT states.

This geometric fact has the following interesting consequence.  Let
$\lambda$ denote the smallest one among the eigenvalues of $\rho$ and
$\rho^P\!$, we know that $0<\lambda<1/N$.  Then we may construct
entangled PPT states from the separable state $\rho$ and the maximally
mixed state $\rho_0$ as
\begin{equation}
\sigma=p\rho+(1-p)\rho_0\,,
\qquad 1<p\leq \frac{1}{1-N\lambda}\,.
\end{equation}

The number of zeros of generic quadratic extremal witnesses presented
in Table~\ref{tbl:constraints} has consequences also for the geometry
of $\mathcal{S}_1$.  In dimensions $2\times 4$ and $3\times 3$ such a witness
will have $N$ zeros and define a face of $\mathcal{S}_1$ in the interior of
$\mathcal{P}_1\subset\mathcal{D}_1$.  The boundary of this face will consist of other
faces of $\mathcal{S}_1$ defined by less than $N$ zeros, accordingly these
faces are located inside faces of $\mathcal{D}_1$ on the boundary of $\mathcal{D}_1$.  In
higher dimensions such a witness has more than $N$ zeros, and defines
a face of $\mathcal{S}_1$ also in the interior of $\mathcal{D}_1$, but such that its
boundary contains faces still in the interior of $\mathcal{D}_1$.  In this way
the geometry of $\mathcal{S}_1$ in relation to $\mathcal{D}_1$ becomes more and more
complicated as the dimensions increase.

\subsubsection{Some numerical results}

In an attempt to learn more about the geometry of simplex faces we
have studied numerically in dimension $3\times 3$ the faces of $\mathcal{S}_1$
dual to generic quadratic extremal witnesses.  They are 8-simplexes,
each defined by nine linearly independent pure product states that are
the zeros of the witness.  Our sample consisted of 171 extremal
witnesses found in random numerical searches.

\subsubsection*{Volumes}

Define the edge length factors $f_{ij}=d_{ij}^2=\|\rho_i-\rho_j\|^2$
in the Hilbert--Schmidt norm, and the symmetric $(n+2)\times(n+2)$
Cayley--Menger matrix
\begin{equation}
D=
\begin{pmatrix}
	0 & 1 & 1 & \cdots & 1 \\ 
	1 & 0 & f_{12} & \cdots & f_{1k} \\
	1 & f_{21} & 0 & \cdots & f_{2k} \\
	\vdots & \vdots & \vdots & \ddots & \cdots \\
	1 & f_{k1} & f_{k2} & \cdots & 0 \\
\end{pmatrix}\,.
\end{equation}
The volume of this $n$-simplex is then given by
\begin{equation}
V=\frac{{\sqrt{|\det(D)|}}}{2^n\,n!}\,.
\end{equation}
The volume of the regular $n$-simplex with edge length $s$ is
\begin{equation}
V_{\mathrm{reg}}=\frac{s^n}{n!}\sqrt{\frac{n+1}{2^n}}\,.
\end{equation}
The maximal distance between two pure states is $\sqrt{2}$, when the
states are orthogonal.  For $n=8$ and $s=\sqrt{2}$ we have 
$V_{\mathrm{reg}}=3/8!\approx 7.440\times 10^{-5}$.

Due to the highly irregular shapes we find volumes of the faces
varying over five orders of magnitude.  Since we regard density
matrices related by $\textrm{SL}\otimes\textrm{SL}$ transformations as
equivalent, a natural question is how regular these simplexes can be
made by such transformations, in other words, what is the maximum
volume $V^*$ we may obtain.  Figure~\ref{fig:volum} shows our data
for the ratio $V^*/V_{\rm reg}$.  There is some variation left in
this ratio, indicating that the 8-simplexes are genuinely irregular.

\begin{figure}
\centering
\includegraphics[width=0.8\textwidth]{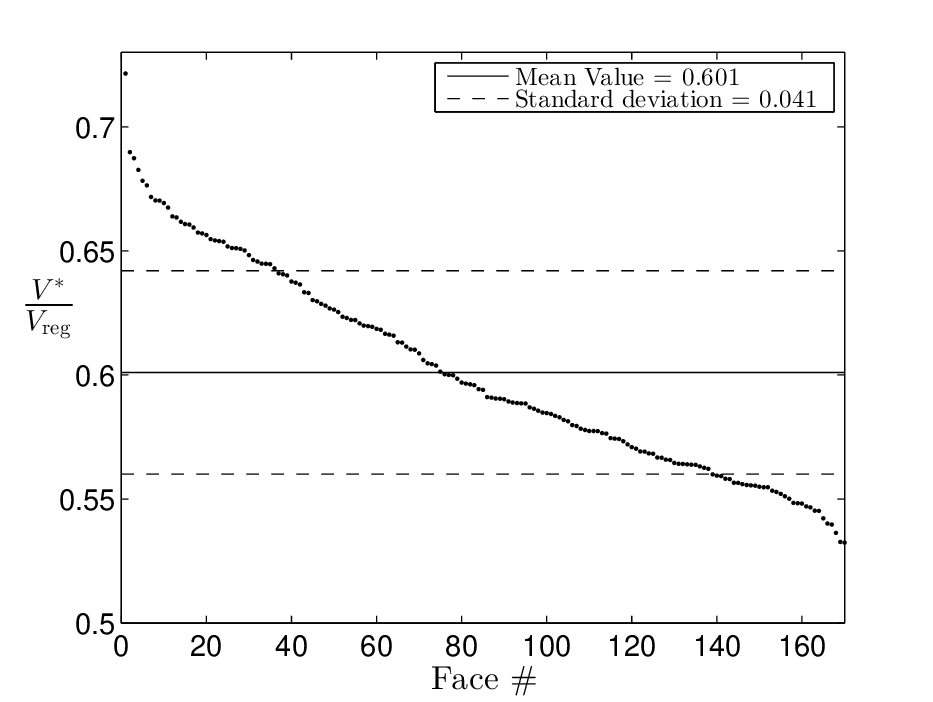}
%\centerline{\psfig{file=nyvolumplott01.eps,width=0.8\textwidth}} 
\caption{Ratio of volumes $V^*$ and $V_{\mathrm{reg}}$ of 8-simplex
faces of $\mathcal{S}_1$ in $3\times 3$. $V^*$ is the simplex volume
maximized under $\textrm{SL}\otimes\textrm{SL}$ transformations, and
$V_{\mathrm{reg}}$ is the volume of the regular simplex of orthogonal
pure product states.}
\label{fig:volum}
\end{figure}

\subsubsection*{Positions relative to the maximally mixed state}

The second geometric property we have studied is the distance from the
maximally mixed state to the center of each simplex face, defined as
the average of the vertices,
\begin{equation}
\rho_c=\frac{1}{k}\sum_{i=1}^k\rho_i\,.
\end{equation}
We compare this distance to the radius of the maximal ball of
separable states centered around the maximally mixed
state~\cite{Gur2002}
\begin{equation}
R_m=\frac{1}{\sqrt{N(N-1)}}\,.
\end{equation}
For $N=3\times3$ we have $R_m=1/\sqrt{72}\approx 0.1179$.  The
distance $d_c=\|\rho_c-\rho_0\|$ can be minimized by
$\textrm{SL}\otimes\textrm{SL}$ transformations on $\rho_c$, resulting in a
unique minimal distance $d_c^*$ for each equivalence class of faces.
Figure~\ref{fig:avstand} shows our data for the ratio $d_c^*/R_m$.
We see that $d_c^*$ does not saturate the lower bound $R_m$, and
this is another indication of the intrinsic irregularity of the
8-simplexes.

\begin{figure}
\centering
\includegraphics[width=0.8\textwidth]{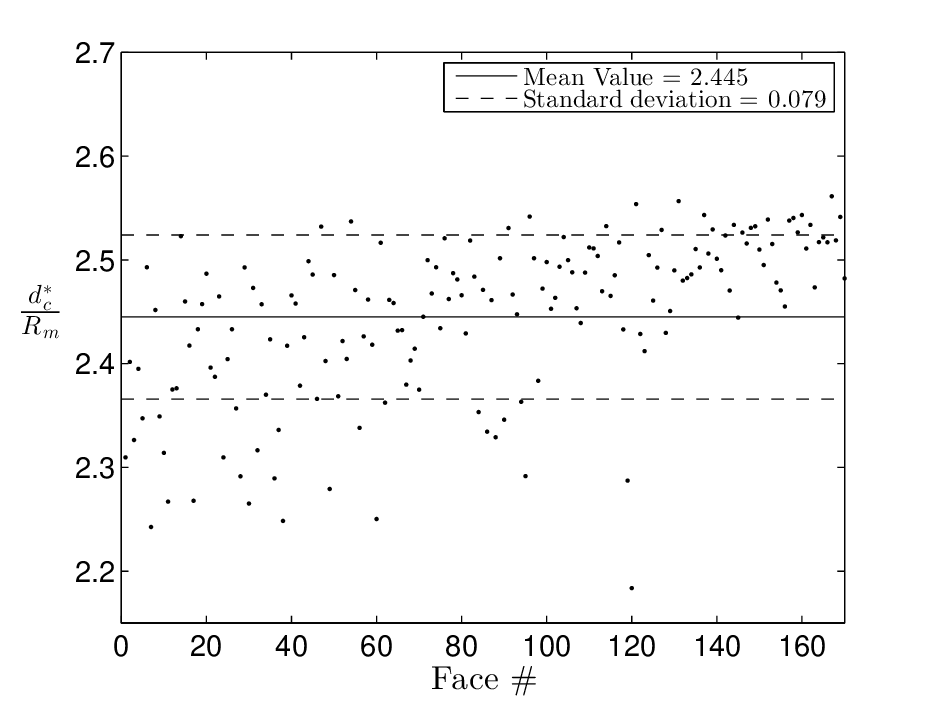}
%\centerline{\psfig{file=nyavstandplott01.eps,width=0.8\textwidth}} 
\caption{Ratio of distances $d_c^*$ and $R_m$ for 8-simplex faces
of $\mathcal{S}_1$ in dimension $3\times 3$. $d_c^*$ is the distance from
the center of the face minimized under $\textrm{SL}\otimes\textrm{SL}$
transformations, and $R_m$ is the radius of the maximal ball of
separable states.  The points are plotted in the same order as in
Figure~\ref{fig:volum}.}
\label{fig:avstand}
\end{figure}

The third property studied is the orientation of each simplex relative
to the maximally mixed state $\rho_0$.  On a given face there exists a
unique state $\rho_{\mathrm{min}}$ closest to $\rho_0$, the minimum
distance is $d_{\mathrm{min}}=\|\rho_{\mathrm{min}}-\rho_0\|$.  In dimension
$N=N_{a}N_{b}$ the distance from any pure state to the maximally
mixed state $\rho_0$ is $\sqrt{(N-1)/N}$, and in our case this gives
an upper limit
\begin{equation}
d_{\mathrm{min}}<\sqrt{\frac{8}{9}}\approx 0.9428\,.
\end{equation}

With our sample of 171 faces it happens in only four cases, or 2\,\%,
that $\rho_{\mathrm{min}}$ lies in the interior of the face and has the
full rank nine.  In the remaining 98\,\% of the cases, $\rho_{\mathrm{min}}$
lies on the boundary of the face and has lower rank.

Figure~\ref{fig:minavstandrank} shows the rank and distance for the
state $\rho_{\rm min}$ in each face in our sample.  We see a tendency
that the minimum distance $d_{\rm min}$ is smaller when the rank of
$\rho_{\rm min}$ is higher.  This confirms our expectation that the
most regular simplex faces are positioned most symmetrically relative
to the maximally mixed state, and also come closest to this state.

We find numerically that if we generate an 8-simplex by generating a
random set of nine pure product states, then $\rho_{\rm min}$ will lie in the
interior of the simplex and have full rank in about $90\,\%$ of the
cases.  Also we observe no ranks smaller than six.  Thus,
Figure~\ref{fig:minavstandrank} proves that a simplex face generated
by a random search for an extremal witness looks very different from a
randomly generated simplex.  A random set of product vectors that
define an 8-simplex will in general not be the zeros of an
entanglement witness, and the simplex will not be a face of $\mathcal{S}_1$.

\begin{figure}
\centering
\includegraphics[width=0.8\textwidth]{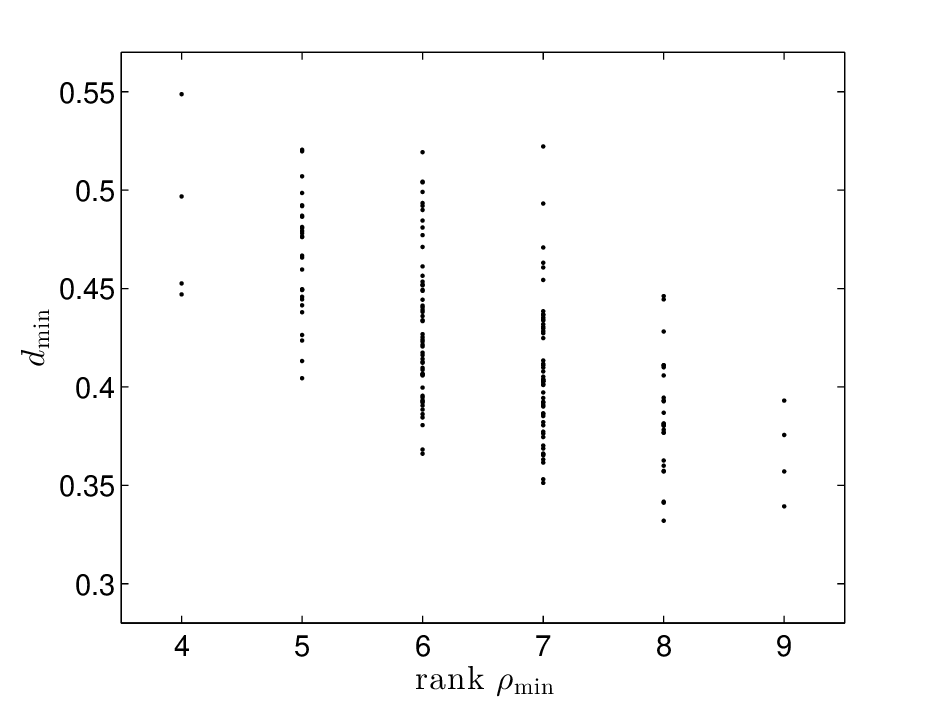}
%\centerline{\psfig{file=minavstandrank02.eps,width=0.8\textwidth}} 
\caption{The rank of $\rho_{\rm min}$ and its distance $d_{\rm min}$
from the maximally mixed state for 8-simplex faces in dimension
$3\times 3$.  The faces are derived from our set of quadratic
extremal entanglement witnesses found in random searches using
Algorithm~1.}
\label{fig:minavstandrank}
\end{figure}

\subsection{Other types of faces of $\mathcal{S}_1$}

To every entanglement witness there corresponds an exposed face of
$\mathcal{S}_1$, and the great variety of entanglement witnesses implies a
similar variety of faces of $\mathcal{S}_1$.  The class of simplex faces, having
only a finite number of extremal points, is rather special, other
faces may have only continuous sets of extremal points, or both
discrete and continuous subsets of extremal points.  Some special
examples in dimension $3\times 3$ may serve as illustrations.

The Choi--Lam witness is a good example.  It has three discrete zeros
and one continuous set of zeros, and these zeros in their role as pure
product states are the extremal points of a face of $\mathcal{S}_1$.

An extremal decomposable witness, \textit{i.e.} a pure state or the partial
transpose of a pure state, defines a face of $\mathcal{S}_1$ with one single
continuous set of extremal points.  To be more specific, consider a
pure state $\Omega=\psi\psi^{\dagger}\!$, then the zeros of $\Omega$ are
the product vectors orthogonal to $\psi$.  We consider two typical
cases.

One example is the product vector $\psi=e_1\otimes e_1$.  Then the
zeros are the product vectors
\begin{equation}
\label{eq:prodst3x3}
\phi\otimes\chi=
(a_1e_1+a_2e_2+a_3e_3)\otimes
(b_1e_1+b_2e_2+b_3e_3)
\end{equation}
with either $a_1=0$ or $b_1=0$.  These product vectors produce pure
product states belonging to a full matrix algebra in dimension
$2\times 3$ if $a_1=0$, and in dimension $3\times 2$ if $b_1=0$.  They
are the extremal points of a face of the type called
in~\cite{Alfsen2010} a convex combination of matrix algebras.

Another example is the Bell type entangled pure state
\begin{equation}
\psi=\frac{1}{\sqrt{3}}\,
(e_1\otimes e_1+e_2\otimes e_2+e_3\otimes e_3)\,.
\end{equation}
The zeros in this case are the product vectors as in
equation~(\ref{eq:prodst3x3}) with
\begin{equation}
a_1b_1+a_2b_2+a_3b_3=0\,.
\end{equation}
For every vector $\phi=a_1e_1+a_2e_2+a_3e_3$ there is a two
dimensional subspace of vectors $\chi=b_1e_1+b_2e_2+b_3e_3$ satisfying
this orthogonality condition.  And vice versa, for every vector $\chi$
there is a two dimensional subspace of vectors $\phi$.  These pure
product states are the extremal points of a face which is neither a
simplex nor a convex combination of matrix algebras.

\subsection{Minimal convex decomposition}

We know that any point in a compact convex set may be written as a
convex combination of extremal points, but if we want to do the
decomposition in practice an important question is how many extremal
points we need.  We would like to know a minimal number of extremal
points that is sufficient for all points.  The facial structure of the
set is then relevant.

The theorem of Carath{\'e}odory gives a sufficient number of extremal
points as $n+1$ where $n$ is the dimension of the set~\cite{Car}.
However, the set $\mathcal{D}_1$ of normalized density matrices is an example
where one can do much better.  The dimension of the set is $n=N^2-1$,
where $N$ is the dimension of the Hilbert space, but the spectral
representation of a density matrix in terms of its eigenvalues and
eigenvectors is a decomposition using only $N$ extremal points.  In
order to see how this number $N$ is related to the facial structure of
$\mathcal{D}_1$, consider the following proof of Carath{\'e}odory's theorem.

Any point $x$ in a compact convex set ${\mathcal{C}}$ is either extremal or
an interior point of a unique face $\mathcal{F}_1$, which might be the whole
set ${\cal C}$.  If $x$ is not extremal it may be written as a convex
combination
\begin{equation}
x=(1-p_1)x_1+p_1y_1\;,
\end{equation}
where $x_1$ is an arbitrary extremal point of $\mathcal{F}_1$ and $y_1$ is
another boundary point of $\mathcal{F}_1$.  If $y_1$ is extremal we define
$x_2=y_1$, otherwise $y_1$ is an interior point of a proper face
$\mathcal{F}_2$ of $\mathcal{F}_1$, and we write
\begin{equation}
y_1=(1-p_2)x_2+p_2y_2\;,
\end{equation}
where $x_2$ is an extremal point of $\mathcal{F}_2$ and $y_2$ is a boundary
point of $\mathcal{F}_2$.  Continuing this process we obtain a decomposition of
$x$ as a convex combination of extremal points
$x_1,x_2,\ldots,x_k\in{\cal C}$, and a sequence
$\mathcal{F}_1\supset\mathcal{F}_2\supset\ldots\supset\mathcal{F}_k$ of faces of decreasing
dimensions
\begin{equation}
n\geq n_1>n_2>\ldots>n_k=0\;.
\end{equation}
The length of the sequence is $k$, and the obvious inequality
$k\leq n+1$ is Carath{\'e}odory's theorem.

For the set $\mathcal{D}_1$ of normalized density matrices the longest possible
sequence of face dimensions has length $N$, it is
\begin{equation}
N^2-1>(N-1)^2-1>\ldots>(N-j)^2-1>\ldots>8>3>0\;.
\end{equation}
In general, we decompose an arbitrary density matrix $\rho$ of rank
$r$ as a convex combination of $r$ pure states that may be
eigenvectors of $\rho$, but need not be.

We may apply this procedure to the decomposition of an arbitrary
separable state as a convex combination of pure product states.  As we
see, the number of pure product states we have to use depends very
much on the facial structure of the set of separable states.  It might
happen, for example, that the face $\mathcal{F}_2$ is the dual of a quadratic
witness, and then it contains only a finite number of pure product
states, much smaller than the dimension $N^2-1$ of the set $\mathcal{S}_1$ of
normalized separable states.  Recall that the pure product states in
such a face are the zeros of the witness, and the number of zeros is
largest when the witness is extremal.  In the generic case the number
of zeros of a quadratic extremal witness is $n_c$, as given by
equation~(\ref{eq:nzquadrextrw}) and tabulated in
Table~\ref{tbl:constraints}.

It seems likely therefore that also in the decomposition of separable
states one could do much better than the $N^2$ pure product states
guaranteed by Carath{\'e}odory's theorem.  We consider this an
interesting open problem for further study.

%%%%%%%%%%%%%%%%%%%%%%%%%%%%%%%%%%%%%%%%%%%%%%%%%%%%%%%%%%%%%%%%%%%%%%%%%%%%%%%%%%%%%%%%%%%%%%%%%%%%%%%%%%%%%%%
%%%%%%%%%%%%%%%%%%%%%%%%%%%%%%%%%%%%%%%%%%%%%%%%%%%%%%%%%%%%%%%%%%%%%%%%%%%%%%%%%%%%%%%%%%%%%%%%%%%%%%%%%%%%%%%

\section{Optimal and extremal witnesses}
\label{sec:optimal}

The notion of optimal entanglement witnesses was developed by
Lewenstein \textit{et al.}~\cite{Lew,Lew2001}.  In this section we comment
briefly on the relation between optimal and extremal witnesses.  In
systems of dimension $2\times N_b$ and $3\times 3$ the two concepts
are, at least generically, closely related.  This has significant
geometric consequences.

Recall the definition that a witness $\Omega$ is extremal if and only
if there does not exist two other witnesses $\Omega_0$ and $\Omega_1$
such that
\begin{equation}
\label{eq:prprcnvcmb}
\Omega=(1-p)\Omega_0+p \Omega_1\,, \qquad 0<p<1\,.
\end{equation}

Lewenstein \textit{et al.} introduce two different concepts of optimality and
give two necessary and sufficient conditions for optimality, similar
to the condition for extremality~\cite{Lew}.  By their Theorem~1, a
witness $\Omega$ is optimal if and only if it is not a convex
combination of this form where $\Omega_1$ is a positive semidefinite
matrix.  By their Theorem~1(b), $\Omega$ is nondecomposable optimal if
and only if it is not a convex combination of this form where
$\Omega_1$ is a decomposable witness.  Obviously, extremal witnesses
are optimal by both optimality criteria.

Ha and Kye give examples to show that an optimal witness which is
nondecomposable need not be nondecomposable optimal according to the
definition of Lewenstein {\em et al.}, hence they propose to replace
the ambiguous term ``nondecomposable optimal witness'' by ``optimal
PPTES witness'', where PPTES refers to entangled states with positive
partial transpose~\cite{Ha2012b}.  Here we will use the original
terminology.

Note that the criterion for optimality is not invariant under partial
transposition, hence a witness $\Omega$ may be optimal while its
partial transpose $\Omega^P$ is not optimal, and vice versa.  The
criterion for nondecomposable optimality, on the other hand, is
invariant under partial transposition, so that $\Omega$ is
nondecomposable optimal if and only if $\Omega^P$ is nondecomposable
optimal.

Apparently the condition of being extremal is stricter than those of
being optimal or nondecomposable optimal.  Accordingly there should,
in general, exist plenty of nonextremal witnesses that are either
optimal or nondecomposable optimal.  We now investigate this question
further in the spirit of Theorem~\ref{thm:zerosspan2}.  The following
sufficient optimality conditions were proved by Lewenstein \textit{et al.}~\cite{Lew}
\begin{theorem}
\label{thm:optimalif}
\textit{A witness $\Omega$ is optimal if its zeros span the Hilbert space.}

\textit{$\Omega$ is nondecomposable optimal if its zeros span the Hilbert
space, and at the same time its partially conjugated zeros span the
Hilbert space.}
\end{theorem}
\begin{proof}
Equation~(\ref{eq:prprcnvcmb}) implies that the zeros and Hessian
zeros of $\Omega$ are also zeros and Hessian zeros of $\Omega_0$ and
$\Omega_1$.  But if $\Omega_1$ is positive semidefinite and its
zeros span the Hilbert space, then $\Omega_1=0$.  Similarly, if
$\Omega_1$ is decomposable and its zeros and partially conjugated
zeros both span the Hilbert space, then $\Omega_1=0$.
\end{proof}

For a quadratic witness these conditions are not only sufficient but
also necessary.

\begin{theorem}
\label{thm:quadraticoptimalII}
\textit{Let $\Omega$ be a quadratic witness.}

\textit{Then $\Omega$ is optimal if and only if its zeros span the Hilbert
space.}

\textit{$\Omega$ is nondecomposable optimal if and only if its zeros span
the Hilbert space and its partially conjugated zeros also span the
Hilbert space.}

\textit{Hence, by Theorem~\ref{thm:zerosnospan2}, if $\Omega$ is
nondecomposable optimal it is nondecomposable.}
\end{theorem}
\begin{proof}
We have to prove the ``only if'' part, and we reason like in the
proofs of Theorem~\ref{thm:zerosspan2} and
Theorem~\ref{thm:extremecondI}.  If the zeros do not span the
Hilbert space, there exists a vector $\psi$ orthogonal to all zeros,
and we define $\Omega_1=\psi\psi^{\dagger}$.  Then the zeros of
$\Omega_1$ include the zeros of $\Omega$, and we know that
$\Omega_1\neq\Omega$ because $\Omega$ is quadratic and $\Omega_1$ is
quartic.  The line segment from $\Omega_1$ to $\Omega$ consists of
witnesses having exactly the same zeros and Hessian zeros as
$\Omega$ ($\Omega$ has no Hessian zeros since it is quadratic).  By
Theorem~\ref{thm:witzerosHessianII} this line segment can be
prolonged within $\mathcal{S}_1^{\circ}$ so that it gets $\Omega$ as an interior
point.  Hence $\Omega$ is neither optimal nor nondecomposable
optimal.

If the partially conjugated zeros do not span the Hilbert space,
then there exists a vector $\eta$ orthogonal to all the partially
conjugated zeros.  Now we define $\Omega_1=(\eta\eta^{\dagger})^P\!$,
and use it in the same way as a counterexample to prove that
$\Omega$ is not nondecomposable optimal.
\end{proof}

It is known that the equivalence between optimality and the spanning
property holds also in the case of decomposable witnesses in
$2\times N_b$ dimensions.  Augusiak, Tura and Lewenstein proved the
following result~\cite{Aug2011A}

\begin{theorem}
\label{thm:optimaldecomp2xN}
\textit{Let $\Omega$ be a decomposable witness in $2\times N_b$ dimensions,
then the following conditions are equivalent.}
\begin{enumerate}[i)]
{\setlength\itemsep{0.0\baselineskip}
{\setlength\itemindent{25pt}
\item \textit{$\Omega$ is optimal;}}
{\setlength\itemindent{25pt}
\item \textit{The zeros of $\Omega$ span the Hilbert space;}}
{\setlength\itemindent{25pt}
\item \textit{$\Omega=\sigma^P\!$, where $\sigma>0$ and $\mathrm{Img}\,\sigma$
contains no product vectors.}}
}
\end{enumerate}
\end{theorem}
A subspace containing no product vectors is often called a completely
entangled subspace.
Obviously,
 $\Omega$ is not extremal
unless $\sigma$ is
a pure state.  It is also known
that for dimensions $3\times 3$ and higher, there exist optimal
decomposable witnesses without the spanning property~\cite{Aug2011B}.

Lewenstein \textit{et al.} also give an optimality test for quartic witnesses
without the spanning property.  Since we do not have a complete
understanding of constraints from quartic zeros, we leave the relation
between nondecomposable optimality and extremality of quartic
witnesses as an open problem.

In dimensions $3\times 3$ and $2\times N_b$, a witness with doubly
spanning zeros will generically be extremal, by
Theorem~\ref{thm:2xn3x3}, and hence both optimal and nondecomposable
optimal.  From Table~\ref{tbl:constraints} we see that in higher
dimensions the number of zeros of a quadratic extremal witness must be
larger than the Hilbert space dimension $N=N_aN_b$.  Hence
nonextremal witnesses with doubly spanning zeros must be very common,
so that nondecomposable optimality is a significantly weaker property
than extremality.  In particular, when we search for generic quadratic
extremal witnesses and reach the stage where the witnesses have $N$ or
more zeros, then every face of $\mathcal{S}_1^{\circ}$ we encounter will consist
entirely of nondecomposable optimal witnesses, most of which are not
extremal.

The observations made above give a picture of increasingly complicated
geometry as dimensions increase.  In dimensions $2\times 2$ and
$2\times 3$ every witness is decomposable~\cite{HorMPR96,Sto1963}.  In
$2\times N_b$ and $3\times 3$ a generic quadratic witness is either
nondecomposable optimal and extremal, or it is neither.  In higher
dimensions there exist an abundance of nondecomposable optimal
nonextremal witnesses.

%%%%%%%%%%%%%%%%%%%%%%%%%%%%%%%%%%%%%%%%%%%%%%%%%%%%%%%%%%%%%%%%%%%%%%%%%%%%%%%%%%%%%%%%%%%%%%%%%%%%%%%%%%%
%%%%%%%%%%%%%%%%%%%%%%%%%%%%%%%%%%%%%%%%%%%%%%%%%%%%%%%%%%%%%%%%%%%%%%%%%%%%%%%%%%%%%%%%%%%%%%%%%%%%%%%%%%%

\section{The SPA separability conjecture}
\label{sec:spa}

Horodecki and Ekert introduced the concept of a ``structural physical
approximation'' (SPA) of a non-physical map that can be used as a
mathematical operation for detecting
entanglement~\cite{PHor2003,HorEk}.  The SPA is a physical map and may
be implemented in an experimental setup.  An example is partial
transposition, where the transposition map is applied to one
subsystem.  The transposition map is unphysical, since it is not
completely positive, hence the need for an SPA.

The Jamio{\l}kowski isomorphism relates positive maps and entanglement
witnesses, and the SPA of a witness $\Omega$ is usually defined as the
nearest positive matrix that is a convex combination of $\Omega$ and
the maximally mixed state $I/N$.  The witness itself corresponds to
a map which is positive but not completely positive, whereas its SPA
corresponds to a completely positive map.  Define
\begin{equation}
\Sigma(p)=(1-p)\Omega+\frac{p}{N}\,I\,,
\end{equation}
and let

\begin{itemize}
{\setlength\itemsep{0.0\baselineskip}
  \item[--] $p_0$ be the smallest value of $p$ such that $\Sigma(p)$ is
  a separable density matrix;
  \item[--] $p_1$ be the smallest value of $p$ such that $\Sigma(p)$ is
  a density matrix;
  \item[--] $p_2$ be the smallest value of $p$ such that
  $(\Sigma(p))^P$ is a density matrix.
}
\end{itemize}
If $\lambda_1\leq 0$ and $\lambda_2\leq 0$ are the lowest eigenvalues
of $\Omega$ and $\Omega^P$ respectively, then
\begin{equation}
p_1=-\frac{N\lambda_1}{1-N\lambda_1}\,,\qquad
p_2=-\frac{N\lambda_2}{1-N\lambda_2}\,.
\end{equation}
The SPA of $\Omega$ is defined in~\cite{Kor} as
$\widetilde{\Omega}=\Sigma(p_1)$.

Clearly $p_0\geq p_1$ and $p_0\geq p_2$, but we may have either
$p_1<p_2$, $p_1=p_2$, or $p_1>p_2$.  The SPA of $\Omega$ is a PPT
state if and only if $p_1\geq p_2$.  It follows directly from the
definitions that the SPA of the witness $\Omega^P$ is
$(\Sigma(p_2))^P\!$, and this is a PPT state if and only if
$p_2\geq p_1$.

The so called SPA separability conjecture was put forward by Korbicz
{\em et al.}, and states that the SPA of an optimal entanglement
witness is always a separable density matrix~\cite{Kor}.  Thus for
optimal witnesses we should have $p_0=p_1\geq p_2$.  One reason for
the interest in this conjecture is that separability simplifies the
physical implementation of the corresponding map.

The conjecture was supported by several explicit
constructions~\cite{Aug,Chru2,Chru,Chru6}.  In particular, it is true
for extremal decomposable witnesses, that is, for witnesses of the
form $\Omega=(\psi\psi^{\dagger})^P$ where $\psi$ is an entangled
vector~\cite{Kor,Aug}.  Note that in this case
$\Omega^P=\psi\psi^{\dagger}$ is not an entanglement witness, according
to the standard definition that a witness must have at least one
negative eigenvalue.  According to our definition a witness is not
required to have negative eigenvalues, thus we define
$\psi\psi^{\dagger}$ to be an extremal decomposable witness.

Counterexamples to the SPA conjecture are known.  It was shown
recently by Chru{\'s}ci{\'n}ski and Sarbicki that the SPA of an
optimal but non-extremal decomposable witness, of the form
$\Omega=\sigma^P$ with $\sigma\geq 0$ of rank three, may be
entangled, although it is automatically a PPT state~\cite{Chru2014}.
This counterexample leaves open the possibility that the conjecture
may hold for all such witnesses in dimensions $2\times N_b$ with
$N_b\geq 4$.

Counterexamples of the form of the generalized Choi witness,
Equation~(\ref{eq:choiwitgen}), have been given by Ha and Kye and by
St{\o}rmer~\cite{Ha2012,Stormer2013}.  With the parameter values
considered by St{\o}rmer we find that the witnesses are not extremal,
and their SPAs are entangled because they are not PPT states.  More
interesting are the parameter values considered by Ha and Kye, because
they give extremal witnesses with SPAs that are entangled PPT states.
We describe these witnesses in some more detail below.

The SPA conjecture applies directly to extremal entanglement
witnesses, since they are optimal.  If $\Omega$ is an extremal
witness, then so is the partial transpose $\Omega^P\!$.  However, if
the SPA of $\Omega$ is a PPT state, then the SPA of $\Omega^P$ is
usually not, and vice versa.  This means that the original SPA
conjecture can not be true both for $\Omega$ and $\Omega^P\!$, simply
because separable states are PPT states.  Ha and Kye gave essentially
this argument in~\cite{Ha2012}.  What saves the SPA conjecture in the
case of an extremal decomposable witness
$\Omega=(\psi\psi^{\dagger})^P$ and its partial transpose
$\Omega^P=\psi\psi^{\dagger}$ is that the latter is excluded by the
standard definition.

An obvious modification of the conjecture would be to redefine the SPA
as the nearest PPT matrix instead of the nearest positive matrix.  In
the case of a pair of extremal witnesses $\Omega$ and $\Omega^P$ the
modified conjecture holds either for both or for none.  Since the
original SPA conjecture holds for extremal decomposable witnesses
$\Omega=(\psi\psi^{\dagger})^P\!$, the modified SPA conjecture holds for
$\Omega^P=\psi\psi^{\dagger}\!$, which is an extremal decomposable witness
according to our definition.

The counterexamples cited above disprove also this modified
separability conjecture in its full generality, including some
extremal nondecomposable witnesses.  However, we want to point out
that the known counterexamples are rather nongeneric.  Therefore we
have made a small numerical investigation of how the modified
conjecture works for generic quadratic extremal witnesses.  We are
unable to draw a definitive conclusion because of the limited
numerical precision of our separability test.

\subsection*{The counterexample of Ha and Kye}

The generalized Choi witness, Equation~(\ref{eq:choiwitgen}), contains
four parameters $a,b,c,\theta$.  It is required here that
$2/3\leq a<1$.  Each value of $a$ gives twelve different witnesses,
depending on two arbitrary signs and one arbitrary angle
$\theta_1=0,\pm 2\pi/3$.  We define
\begin{align}
\nonumber
a_1 = \sqrt{\frac{1-(1-a)^2}{2}}\;,\quad & \nonumber \quad a_2=\sqrt{\frac{1-9(1-a)^2}{2}}\;,\\
\nonumber
b = \frac{a_1\pm a_2}{2}\;,\quad & \nonumber\quad c=a_1-b\;,\\
\theta_0 = \arccos\!\left(\frac{a+a_1}{2}\right)\;,\quad & \quad\theta=\pm\theta_0+\theta_1\;.
\end{align}
The unnormalized SPA of this witness $\Omega$ is $\rho=\Omega+a_1I$,
of rank eight, and the SPA of $\Omega^P$ is $\rho^P\!$, of rank six.
Thus $\rho$ is a PPT state.  The kernel of $\rho$ is spanned by the
vector $\psi^{(0)}\!$, and the kernel of $\rho^P$ is spanned by
$\psi^{(1)}\!,\psi^{(2)}\!,\psi^{(3)}$ where
\begin{equation}
(\psi^{(0)},\psi^{(1)},\psi^{(2)},\psi^{(3)})
=\begin{pmatrix}
  1 & 0 & 0 & 0\\
  0 &{\rm e}^{{\rm i}\theta} & 0 & 0\\
  0 & 0 &c+a_1& 0\\
  0 &c+a_1& 0 & 0\\
  {\rm e}^{-{\rm i}\theta_1}& 0 & 0 & 0\\
  0 & 0 & 0 &{\rm e}^{{\rm i}\theta}\\
  0 & 0 &{\rm e}^{{\rm i}\theta}& 0\\
  0 & 0 & 0 &c+a_1\\
  {\rm e}^{{\rm i}\theta_1}& 0 & 0 & 0
\end{pmatrix}.
\end{equation}

A separable state $\sigma$ satisfies the range criterion that
$\textrm{Img}\,\sigma$ is spanned by product vectors such that the partially
conjugated product vectors span $\textrm{Img}\,\sigma^P\!$.  The present PPT
state $\rho$ is entangled, in fact it is a so called edge state,
violating the range criterion maximally.  It is straightforward to
verify that there exists no product vector $\varphi\otimes\chi$
orthogonal to $\psi^{(0)}$ such that $\varphi\otimes\chi^{\ast}$ is
orthogonal to all of $\psi^{(1)}\!,\psi^{(2)}\!,\psi^{(3)}$.

\subsection*{Generic quadratic extremal witnesses}

When we find an extremal witness $\Omega$ in a random numerical search
by Algorithm~1, presumably the probabilities of
finding $\Omega$ and $\Omega^P$ are equal.  Hence we expect to have
about fifty per cent probability of finding an extremal witness with
an SPA which is a PPT state.  This is also what we see in practice.

For each extremal witness $\Omega$ found in dimensions $3\times 3$ and
$3\times 4$ we have computed the SPA of $\Omega$ and of $\Omega^P\!$.
As expected, in each case one SPA is a PPT state and the other one is
not.  The number of PPT states in our numerical searches originating from 
the quadratic witnesses $\Omega$ or
$\Omega^P\!$, is listed in Table~\ref{tbl:spapptnpt}.

\begin{table}
\centering
{\begin{tabular}{ccc}
\hline
     & $3\times 3$ & $3\times 4$\\%[0.1cm]
\hline\\[-0.4cm]
     $\Omega$       & 81 & 41 \\	
     $\Omega^P$    & 90 & 24 \\
\hline
\end{tabular}}
  \caption{
    Number of cases where the SPA of $\Omega$ or $\Omega^P$ is a
    PPT state, when $\Omega$ is a quadratic
    extremal witness found numerically in dimension
    $3\times 3$ or $3\times 4$.}
\label{tbl:spapptnpt}
\end{table}

Since the original separability conjecture of Korbicz \textit{et al.} fails both
in the generic case and by non-generic counterexamples, a natural
question is whether the SPA is always separable when it is a PPT
state.  As mentioned there are counterexamples
in~\cite{Ha2012,Stormer2013}.  On the other hand, these
counterexamples are rather special, and there remains a possibility
that the SPA of a generic extremal witness with only quadratic zeros
could be separable if it is a PPT state.

For this reason, in our numerical examples we have tried to check the
separability of the SPA whenever it is a PPT state.  We find that it
is always very close to being separable, in no case are we able to
conclude that it is certainly not separable.  Unfortunately, our
numerical separability test is not sufficiently precise that we may
state with any conviction the opposite conclusion that the state is
separable.

In view of the known examples and counterexamples the status of the
modified SPA separability conjecture, which defines the SPA as a PPT
state, is not completely clear.  It would be interesting to know some
simple criterion for exactly when it is true.  This may also be
regarded as one aspect of the more general question of which border
states of the set of separable states are also border states of the
set of PPT states.

On the other hand, if we regard separability as an essential property
of the SPA, it seems that the natural solution would be to define the
SPA by the parameter value $p=p_0$, where $\Sigma(p)$ first becomes
separable, rather than by the value $p=p_1$, where $\Sigma(p)$ first
becomes positive, or by $p=\max(p_1,p_2)$, where $\Sigma(p)$ first
becomes a PPT state.

%%%%%%%%%%%%%%%%%%%%%%%%%%%%%%%%%%%%%%%%%%%%%%%%%%%%%%%%%%%%%%%%%%%%%%%%%%%%%%%%%%%%%%%%%%%%%%%%%%%%%%%%%%%
%%%%%%%%%%%%%%%%%%%%%%%%%%%%%%%%%%%%%%%%%%%%%%%%%%%%%%%%%%%%%%%%%%%%%%%%%%%%%%%%%%%%%%%%%%%%%%%%%%%%%%%%%%%

\section{A real problem in a complex cloak}
\label{sec:real}

We want to outline here a different approach to the problem of
describing the convex cone $\mathcal{S}^{\circ}$ of entanglement witnesses, where we
regard it as a real rather than a complex problem.  This means that we
treat $\mathcal{S}_1^{\circ}$ as a subset of a higher dimensional compact convex set.
We are motivated by the observation that the problem is intrinsically
real, since $H$ is a real vector space, $f_\Omega$ is a real valued
function, and some of the constraints discussed in
Section~\ref{sec:constraints} are explicitly real.  In particular, the
second derivative matrix, the Hessian, defining constraints at a
quartic zero of a witness is a real and not a complex matrix.

Another motivation is the possibility of expressing the extremal
points of $\mathcal{S}_1^{\circ}$ as convex combinations of extremal points of the
larger set.  As an example we show how a pure state as a witness is a
convex combination of real witnesses that are not of the complex form.

In order to arrive at an explicitly real formulation of the problem we
proceed as follows.  Define $J:\mathbb{R}^{2N_a}\rightarrow\mathbb{C}^{N_a}$ and
$K:\mathbb{R}^{2N_b}\rightarrow\mathbb{C}^{N_b}$ as
\begin{equation}
\label{eq:realtransform}
J=(I_a,{\rm i} I_a)\,,\qquad
K=(I_b,{\rm i} I_b)\,,
\end{equation}
and define
\begin{equation}
Z=(J\otimes K)^\dagger\Omega(J\otimes K)\,.
\end{equation}
Then if $\phi=Jx$ and $\chi=Ky$ we have that
\begin{equation}
g_Z(x,y)=(x\otimes y)^T Z(x\otimes y)
=(\phi\otimes\chi)^\dagger\Omega(\phi\otimes\chi)=f_\Omega(\phi,\chi)\,. 
\end{equation}
Since $Z$ is a Hermitian matrix, its real part is symmetric and its
imaginary part antisymmetric.  The expectation value $g_Z(x,y)$ is
real because only the symmetric part of $Z$ contributes.  Furthermore,
only the part of $Z$ that is symmetric under partial transposition
contributes.  Hence $g_Z(x,y)=g_W(x,y)$ when we define
\begin{equation}
\label{eq:realwitness}
W=\frac{1}{2}\,\mathrm{Re}\,\,(Z+Z^P)\,.
\end{equation}
Since $J\otimes K=(I_a\otimes K,{\rm i} I_a\otimes K)$ is an $N\times(4N)$ matrix,
$Z$ is a $(4N)\times(4N)$ matrix,
\begin{equation}
Z=\begin{pmatrix} X & {\rm i} X \\ -{\rm i} X & X\end{pmatrix},
\end{equation}
where $X$ is a $(2N)\times(2N)$ matrix,
\begin{equation}
X=(I_a\otimes K)^{\dagger}\Omega(I_a\otimes K)\,.
\end{equation}
The partial transposition of $Z$ transposes submatrices of size
$(2N_b)\times(2N_b)$.

We may now replace the complex biquadratic form $f_\Omega$ by the real
form $g_W$, and repeat the analysis in Section~\ref{sec:constraints}
almost unchanged.  We therefore see that complex witnesses in
dimension $N=N_a\times N_b$ correspond naturally to real witnesses
in $(2N_a)\times(2N_b)=4N$.

Write the matrix elements of $\Omega$ as
\begin{equation}
(\Omega_{ik})_{jl}=\Omega_{ij;kl}\,.
\end{equation}
In this notation, $\Omega_{ik}$ is an $N_b\times N_b$ submatrix of
the $N\times N$ matrix $\Omega$.  Then we have that
\begin{equation}
X_{ik}=K^{\dagger}\Omega_{ik}K
=\begin{pmatrix}
\Omega_{ik} & {\rm i}\Omega_{ik}\\
-{\rm i}\Omega_{ik} & \Omega_{ik}
\end{pmatrix}.
\end{equation}
The complex matrix $\Omega_{ik}$ may always be decomposed as
\begin{equation}
\Omega_{ik}=A_{ik}+B_{ik}+{\rm i}\,(C_{ik}+D_{ik})\,,
\end{equation}
where $A_{ik},B_{ik},C_{ik},D_{ik}$ are real matrices, and
$A_{ik},C_{ik}$ are symmetric, $B_{ik},D_{ik}$ are antisymmetric.  The
Hermiticity of $\Omega$ means that
\begin{equation}
\Omega_{ik}=(\Omega_{ki})^{\dagger}
= A_{ki}-B_{ki}-{\rm i}\,(C_{ki}-D_{ki})\,,
\end{equation}
or equivalently,
\begin{equation}
\label{eq:ABCDikki}
A_{ki}= A_{ik}\,,\quad
B_{ki}=-B_{ik}\,,\quad
C_{ki}=-C_{ik}\,,\quad
D_{ki}= D_{ik}\,.
\end{equation}
In particular,
\begin{equation}
B_{ii}=C_{ii}=0\,.
\end{equation}

We finally arrive at the following explicit  relation between
the complex witness $\Omega$ and the real witness $W$,
\begin{equation}
\label{eq:WUV}
W=\begin{pmatrix} U & -V \\ V & U\end{pmatrix},
\end{equation}
with
\begin{eqnarray}
\label{eq:UikVik}
U_{ik} & = &
\frac{1}{2}\,\mathrm{Re}\,\,(X_{ik}+X_{ik}^T)
=\begin{pmatrix}
A_{ik} & -D_{ik}\\
D_{ik} & A_{ik}
\end{pmatrix},
\nonumber
\\
V_{ik} & = &
\frac{1}{2}\,\mathrm{Im}\,\,(X_{ik}+ X_{ik}^T)
=\begin{pmatrix}
C_{ik} & B_{ik}\\
-B_{ik} & C_{ik}
\end{pmatrix}.
\end{eqnarray}
Note that the $(2N_b)\times(2N_b)$ matrices $U_{ik}$ and $V_{ik}$ are
symmetric by construction, $U_{ik}=(U_{ik})^T$ and
$V_{ik}=(V_{ik})^T\!$.  This means that $W$ is symmetric under partial
transposition, $W^P=W$.

It also follows from the equations~(\ref{eq:UikVik})
and~(\ref{eq:ABCDikki}) that $U_{ki}=U_{ik}$ and $V_{ki}=-V_{ik}$.
This means that the $(2N)\times(2N)$ matrix $U$ is symmetric, $U^T=U$
because
\begin{equation}
(U^T)_{ik}=(U_{ki})^T=U_{ki}=U_{ik}\,,
\end{equation}
whereas $V$ is antisymmetric, $V^T=-V$ because
\begin{equation}
(V^T)_{ik}=(V_{ki})^T=V_{ki}=-V_{ik}\,.
\end{equation}
This means that $W$ is symmetric, $W^T=W$.

\subsection{A pure state as a witness}

If $\Omega$ is a pure state, $\Omega=\psi\psi^{\dagger}$ with
$\psi^{\dagger}\psi=1$, then
\begin{equation}
f_\Omega(\phi,\chi)
=(\phi\otimes\chi)^\dagger\Omega(\phi\otimes\chi)
=|z|^2=(\mathrm{Re}\,z)^2+(\mathrm{Im}\,z)^2\,, 
\end{equation}
with
\begin{equation}
z=(\phi\otimes\chi)^{\dagger}\psi
=(x\otimes y)^T(J\otimes K)^{\dagger}\psi\,.
\end{equation}
We have seen that $\Omega$ is extremal in $\mathcal{S}_1^{\circ}$, but we see now that
it is not extremal in the larger set of real witnesses that are not of
the complex form given in equation~(\ref{eq:WUV}).

Write $\psi=a+{\rm i} b$ with $a,b$ real, and
\begin{equation}
p=(J\otimes K)^{\dagger}\psi
=\begin{pmatrix}
    q\\ -{\rm i} q
\end{pmatrix},
\end{equation}
with
\begin{equation}
q=(I_a\otimes K)^{\dagger}\psi\,.
\end{equation}
Here $q$ is a $(2N)\times 1$ matrix, and if we write $q_i$ with
$i=1,2,\ldots,N_a$ for the submatrices of size $(2N_b)\times 1$, we
have that
\begin{equation}
q_i=K^{\dagger}\psi_i=
\begin{pmatrix}
\psi_i\\ -{\rm i}\psi_i
\end{pmatrix}=
\begin{pmatrix}
a_i+{\rm i} b_i\\ b_i-{\rm i} a_i
\end{pmatrix}
=r_i+{\rm i} s_i\,,
\end{equation}
with
\begin{equation}
r_i=
\begin{pmatrix}
a_i\\ b_i
\end{pmatrix},\qquad
s_i=
\begin{pmatrix}
b_i\\ -a_i
\end{pmatrix}.
\end{equation}
We now have that
\begin{equation}
(\mathrm{Re}\,z)^2=(x\otimes y)^T Z_r(x\otimes y)\,,\qquad
(\mathrm{Im}\,z)^2=(x\otimes y)^T Z_i(x\otimes y)\,,
\end{equation}
with
\begin{eqnarray}
Z_r & = &
(\mathrm{Re}\,p)(\mathrm{Re}\,p)^T=
\begin{pmatrix}
rr^T & rs^T\\
sr^T & ss^T
\end{pmatrix},
\nonumber
\\
Z_i & = &
(\mathrm{Im}\,p)(\mathrm{Im}\,p)^T=
\begin{pmatrix}
ss^T & -sr^T\\
-rs^T & rr^T
\end{pmatrix}.
\end{eqnarray}
These two matrices are symmetric, $Z_r^T=Z_r$ and $Z_i^T=Z_i$, but we
should also make them symmetric under partial transposition.
Therefore we replace them by the matrices
\begin{eqnarray}
\nonumber
W_r & = &
\frac{1}{2}\,(Z_r+Z_r^P)=
\begin{pmatrix}
    A & B\\
    C & D
\end{pmatrix},
\\
W_i & = &
\frac{1}{2}\,(Z_i+Z_i^P)=
\begin{pmatrix}
    D & -C\\
    -B & A
\end{pmatrix},
\end{eqnarray}
where
\begin{eqnarray}
\nonumber
A_{ik} & = &
\frac{1}{2}\,(r_ir_k^T+r_kr_i^T)=
\frac{1}{2}
\begin{pmatrix}
  a_ia_k^T+a_ka_i^T & a_ib_k^T+a_kb_i^T\\
  b_ia_k^T+b_ka_i^T & b_ib_k^T+b_kb_i^T
\end{pmatrix},
\\
\nonumber
B_{ik} & = &
\frac{1}{2}\,(r_is_k^T+s_kr_i^T)=
\frac{1}{2}
\begin{pmatrix}
    a_ib_k^T+b_ka_i^T & -a_ia_k^T+b_kb_i^T\\
    b_ib_k^T-a_ka_i^T & -b_ia_k^T-a_kb_i^T
\end{pmatrix},
\\
\nonumber
C_{ik} & = &
\frac{1}{2}\,(s_ir_k^T+r_ks_i^T)=
\frac{1}{2}
\begin{pmatrix}
    b_ia_k^T+a_kb_i^T & b_ib_k^T-a_ka_i^T\\
    -a_ia_k^T+b_kb_i^T & -a_ib_k^T-b_ka_i^T
\end{pmatrix},
\\
D_{ik} & = &
\frac{1}{2}\,(s_is_k^T+s_ks_i^T)=
\frac{1}{2}
\begin{pmatrix}
    b_ib_k^T+b_kb_i^T & -b_ia_k^T-b_ka_i^T\\
    -a_ib_k^T-a_kb_i^T & a_ia_k^T+a_ka_i^T
\end{pmatrix}.
\end{eqnarray}

%%%%%%%%%%%%%%%%%%%%%%%%%%%%%%%%%%%%%%%%%%%%%%%%%%%%%%%%%%%%%%%%%%%%%%%%%%%%%%%%%%%%%%%%%%%%%%%%%%%%%%%%%%%%%
%%%%%%%%%%%%%%%%%%%%%%%%%%%%%%%%%%%%%%%%%%%%%%%%%%%%%%%%%%%%%%%%%%%%%%%%%%%%%%%%%%%%%%%%%%%%%%%%%%%%%%%%%%%%%

\section{Summary and outlook}
\label{sec:conclusion}

In this article we have approached the problem of distinguishing
entangled quantum states from separable states through the related
problem of classifying and understanding entanglement witnesses.  We
have studied in particular the extremal entanglement witnesses from
which all other witnesses may be constructed.  We give necessary and
sufficient conditions for a witness to be extremal, in terms of the
zeros and Hessian zeros of the witness, and in terms of linear
constraints on the witness imposed by the existence of zeros.  The
extremality conditions lead to systematic methods for constructing
numerical examples of extremal witnesses.

We find that the distinction between quadratic and quartic zeros is
all important when we classify extremal witnesses.  Nearly all
extremal witnesses found in random searches are quadratic, having only
quadratic zeros.  There is a minimum number of zeros a quadratic
witness must have in order to be extremal, and most of the quadratic
extremal witnesses have this minimum number of zeros.  To our
knowledge, extremal witnesses of this type have never been observed
before, even though they are by far the most common.

The number of zeros of a quadratic extremal witness increases faster
than the Hilbert space dimension.  Since a witness is optimal if its
zeros span the Hilbert space, this implies that in all but the lowest
dimensions witnesses may be optimal and yet be very far from extremal.

The facial structures of the convex set of separable states and the
convex set of witnesses are closely related.  The zeros of a witness,
whether quadratic or quartic, always define an exposed face of the set
of separable states.  The other way around, an exposed face of the set
of separable states is defined by a set of pure product states that
are the zeros of some witness, in fact, they are the common zeros of
all the witnesses in a face of the set of witnesses.  The existence of
witnesses with quartic zeros is the root of the existence of
nonexposed faces of the set of witnesses.  It is unknown to us whether
the set of separable states has unexposed faces.

One possible way to test whether or not a given state is separable is
to try to decompose it as a convex combination of pure product states.
In such a test it is of great practical importance to know the number
of pure product states that is needed in the worst case.  The general
theorem of Carath{\'e}odory provides a limit of $n+1$ where $n$ is the
dimension of the set, but this is a rather large number, and an
interesting unsolved problem is to improve this limit, if possible.
We point out that this is a question closely related to the facial
structure of the set of separable states, which in turn is related to
the set of entanglement witnesses.

Another unsolved problem we may mention here again is the following.
We have found a procedure for constructing an extremal witness from
its zeros and Hessian zeros, but an arbitrarily chosen set of zeros
and Hessian zeros does not in general define an extremal entanglement
witness.  How do we construct an extremal witness by choosing a set of
zeros and Hessian zeros we want it to have?  Clearly this is a much
more complicated problem than the trivial problem of constructing a
polynomial in one variable from its zeros.

In conclusion, by studying the extremal entanglement witnesses we have
made some small progress towards understanding the convex sets of
witnesses and separable mixed states.  But this has also made the
complexity of the problem even more clear than it was before.  The
main complication is the fact that zeros of extremal witnesses may be
quartic, since this opens up an almost unlimited range of variability
among the quartic witnesses.  A full understanding of this variability
seems a very distant goal.  Nevertheless we believe that it is useful
to pursue the study of extremal witnesses in order to learn more about
the geometry of the set of separable states, and it is rather clear
that a combination of analytical and numerical methods will be needed
also in future work.

%%%%%%%%%%%%%%%%%%%%%%%%%%%%%%%%%%%%%%%%%%%%%%%%%%%%%%%%%%%%%%%%%%%%%%%%%%%%%%%%%%%%%%%%%%%%%%%%%%%%%%%%%%%%
%%%%%%%%%%%%%%%%%%%%%%%%%%%%%%%%%%%%%%%%%%%%%%%%%%%%%%%%%%%%%%%%%%%%%%%%%%%%%%%%%%%%%%%%%%%%%%%%%%%%%%%%%%%%

\section*{Acknowledgments}

We acknowledge gratefully research grants from The Norwegian University of
Science and Technology (Leif Ove Hansen) and from The Norwegian Research
Council (Per \O yvind Sollid). 
  
%%%%%%%%%%%%%%%%%%%%%%%%%%%%%%%%%%%%%%%%%%%%%%%%%%%%%%%%%%%%%%%%%%%%%%%%%%%%%%%%%%%%%%%%%%%%%%%%%%%%%%%%%%%%
%%%%%%%%%%%%%%%%%%%%%%%%%%%%%%%%%%%%%%%%%%%%%%%%%%%%%%%%%%%%%%%%%%%%%%%%%%%%%%%%%%%%%%%%%%%%%%%%%%%%%%%%%%%%

\appendix

%%%%%%%%%%%%%%%%%%%%%%%%%%%%%%%%%%%%%%%%%%%%%%%%%%%%%%%%%%%%%%%%%%%%%%%%%%%%%%%%%%%%%%%%%%%%%%%%%%%%%%%%%%%%
%%%%%%%%%%%%%%%%%%%%%%%%%%%%%%%%%%%%%%%%%%%%%%%%%%%%%%%%%%%%%%%%%%%%%%%%%%%%%%%%%%%%%%%%%%%%%%%%%%%%%%%%%%%%

\section{Explicit expressions for positivity constraints}
\label{sec:explicit}

We show here explicit expressions for the constraints developed in
Section~\ref{sec:polynomeq}.

Given a zero $(\phi_0,\chi_0)$, let $j_m$ for $m=1,\ldots,2(N_a-1)$
be the column vectors of the matrix $J_0$.  We construct them for
example from an orthonormal basis
$\phi_0,\phi_1,\ldots,\phi_{N_a-1}$ of $\mathcal{H}_a$, taking
$j_{2l-1}=\phi_l$, $j_{2l}={\rm i}\phi_l$ for $l=1,\ldots,N_a-1$.
Similarly, let $k_n$ for $n=1,\ldots,2(N_b-1)$ be the column vectors
of $K_0$.

Then the equations $\mathbf{T}_0\Omega=0$, $\mathbf{T}_1\Omega=0$ may be written as

\begin{equation}
\label{eq:bndconds1explicit}
 \textrm{Tr}\,(E^0\Omega)=0\;,\quad\textrm{Tr}\,(E^1_m\Omega)=0\;,\quad\textrm{Tr}\,(E^2_n\Omega)=0\;,
\end{equation}
with
\begin{align}
\nonumber
E^0 & = \phi_0\phi_0^\dagger\otimes\chi_0\chi_0^\dagger\,,\\
\nonumber
E^1_m & = (j_m\phi_0^\dagger+\phi_0 j_m^\dagger)\otimes\chi_0\chi_0^\dagger\,,\\
E^2_n & = \phi_0\phi_0^\dagger\otimes(k_n\chi_0^\dagger+\chi_0 k_n^\dagger)\,.
\end{align}

With $z_i\in\textrm{Ker}\,G_{\Omega}$, $z_i^T=(x_i^T,y_i^T)$ and
$\xi_i=J_0x_i$, $\zeta_i=K_0y_i$, for $i=1,\ldots,K$, the constraints
$(\mathbf{T}_2\Omega)_i=0$ given in equation~\eqref{eq:T2icnstr} may be
written explicitly as follows,

\begin{equation}
\label{eq:bndconds2explicit}
\textrm{Tr}\,(F^1_{im}\Omega)=0\;,\quad\textrm{Tr}\,(F^2_{in}\Omega)=0\;,
\end{equation}
with
\begin{align}
\nonumber
F^1_{im} & =(\xi_i j_m^\dagger+j_m\xi_i^\dagger)\otimes\chi_0\chi_0^\dagger
+(\phi_0 j_m^\dagger+j_m\phi_0^\dagger)\otimes(\zeta_i\chi_0^\dagger
+\chi_0\zeta_i^\dagger)\,,\\
F^2_{in} & =\phi_0\phi_0^\dagger\otimes(\zeta_i k_n^\dagger+k_n\zeta_i^\dagger)
+(\xi_i\phi_0^\dagger+\phi_0\xi_i^\dagger)\otimes(\chi_0 k_n^\dagger
+k_n\chi_0^\dagger)\,.
\end{align}

System $\mathbf{T}_3$ can be written as follows. For any linear combinations
\begin{equation}
\xi=\sum_{i=1}^Ka_i\xi_i\,,\qquad
\zeta=\sum_{i=1}^Ka_i\zeta_i\,,
\end{equation}
with real coefficients $a_i$ and $\xi_i,\zeta_i$ as above, we have
\begin{equation}
\label{eq:bndconds3explicit}
\textrm{Tr}\,G\Omega=0,\quad
G=(\xi\phi_0^\dagger+\phi_0\xi^\dagger)\otimes\zeta\zeta^\dagger
+\xi\xi^\dagger\otimes(\zeta\chi_0^\dagger+\chi_0\zeta^\dagger)\,.
\end{equation}
This is one single constraint for each set of coefficients $a_i$.  The
minimum number of different linear combinations that we have to use is
given by the binomial coefficient
\begin{equation}
\binom{K+2}{3}=\frac{K(K+1)(K+2)}{6}\,.
\end{equation}

%%%%%%%%%%%%%%%%%%%%%%%%%%%%%%%%%%%%%%%%%%%%%%%%%%%%%%%%%%%%%%%%%%%%%%%%%%%%%%%%%%%%%%%%%%%%%%%%%%%%%%%%
%%%%%%%%%%%%%%%%%%%%%%%%%%%%%%%%%%%%%%%%%%%%%%%%%%%%%%%%%%%%%%%%%%%%%%%%%%%%%%%%%%%%%%%%%%%%%%%%%%%%%%%%

\section{Numerical implementation of Algorithm~1}
\label{sec:numimp}

In this appendix we first describe possibilities for solving the optimization
problem \eqref{eq:kernelproblem} and then, assuming that this problem
can be solved, we describe how to localize the boundary of a face of
quadratic witnesses in $\mathcal{S}_1^{\circ}$.

\subsection*{Solving the optimization problem}

We first describe some possibilities for solving the
optimization problem \eqref{eq:kernelproblem}.  This question has
received little attention in the physics community, and in the
optimization literature focus shifted towards a rigorous rather than
towards an ``applied'' approach.  Two comments are in place.
In~\cite{pptIII,Dahl2007} simple algorithms were presented
which in our experience work fine in most situations.  However, when
applied to Algorithm~1 we repeatedly observe poor
convergence, presumably due to degeneracies of eigenvalues.  In~\cite{Eis2004} it is
stated that a straight forward parameterization of
problem \eqref{eq:kernelproblem} can not be solved in an efficient
manner.  Instead they present a formalism which again, in our view, is
unnecessarily complicated for low dimensions.  The conceptually
simplest approach is to perform repeated local minimization from
random starting points, thus heuristically (though sufficiently safe)
obtaining all global optimal points. Below we describe different
possible approaches for local minimization, either applying available
optimization algorithms or implementing simple ones ourselves.

In the case that only the optimal value $p^*$ and not the optimal
solution $(\phi^*,\chi^*)$ is sought for, a positive maps approach
is fruitful.  Given an operator $A\in H$ one can define the linear map
$\mathbf{L}_A$ such that
\begin{equation}
(\phi\otimes\chi)^\dagger A (\phi\otimes\chi)=\chi^\dagger \mathbf{L}_A(\phi\phi^\dagger)\chi\,.
\end{equation}
Problem \eqref{eq:kernelproblem} is then equivalent to minimizing the
smallest eigenvalue of $\mathbf{L}_A(\phi\phi^\dagger)$ as a function of
normalized $\phi$, and hence to the unconstrained minimization of
\begin{equation}
\widetilde{f}_\Omega(\phi)=\mathrm{min}{\left[\mathrm{spectrum}\,\,
\frac{\mathbf{L}_A(\phi\phi^\dagger)}{\phi^\dagger\phi}\right]}.
\end{equation}
With a real parametrization $\phi=Jx$ as in equation
\eqref{eq:realtransform} this function is readily minimized by the
Nelder--Mead downhill simplex algorithm, listed \textit{e.g.} in~\cite{NR} and
implemented \textit{e.g.} in MATLAB's ``fminsearch'' function~\cite{matlab} or
Mathematica's ``NMinimize''~\cite{mathematica}.

In the case that also the optimal point $(\phi^*,\chi^*)$ is
wanted several possibilities exist.  An approximate solution is
easily found by feeding the real problem given by
\eqref{eq:realwitness} to a SQP-algorithm~\cite{NW,Garberg2011}, \textit{e.g.}
the one provided in MATLAB's ``fmincon'' function~\cite{matlab}.
Another possibility is to apply the infinity norm in
\eqref{eq:kernelproblem} and perform box-constrained minimization in
each face of the $\infty$-ball.  Since the Hessian of the function is
readily available trust-region algorithms are suitable (also provided
in MATLAB's ``fmincon'').  If $A$ is known to be a witness, even the
box-constraints on each face of the $\infty$-ball can be removed,
since in that case the objective function is bounded below.  The
problem has then been reduced to unconstrained minimization in a
series of affine spaces, one for each face of the $\infty$-ball.  With
this latter formulation, we have had success with a quasi-Newton
approach~\cite{NW} performing exact line search, see~\cite{Hauge2011}
for details. Further, we have had promising results with a complex
conjugate gradient algorithm also performing exact line search~\cite{Garberg2011},
though further development is necessary.

\subsection*{Finding the boundary of a face}
\label{sec:findfacebnd}

Here we describe how to localize the boundary of a face of quadratic
witnesses in $\mathcal{S}_1^{\circ}$.  We assume we have some routine
for solving problem \eqref{eq:kernelproblem}.

Given a quadratic witness $\Omega$ situated in the interior of a face
defined by $k$ product vectors, and a perturbation $\Gamma$, we define
a family of functions as follows.  Consider problem
\eqref{eq:kernelproblem} with $A(t)=\Omega+t\Gamma$.  Given all local
minima, sort these in ascending order according to function value.
Then the function $\kappa_0(t)$ is the value of local minimum number
$k+1$ if this exists and 0 otherwise.  Function $\kappa_i(t)$,
$i=1,\ldots,k$ is the smallest positive eigenvalue of
$\mathrm{D}^2(f_\Omega+t f_\Gamma)$ evaluated at zero number $i$.  Now each
$\kappa_i(t)$, $i=0,\ldots,k$ is a function with a single simple root
in $t\in[0,\infty)$ which can be easily located using any standard
rootfinding technique.

Another possible approach to finding the boundary of the face is as
follows.  Rather than locating all local minima of problem
\eqref{eq:kernelproblem}, aim only at locating the global minimum
value.  This global minimum value as a function of $\Omega$ will be
called the face function.  The face function is zero on a face,
positive in the interior of $\mathcal{S}^{\circ}$ and negative outside of $\mathcal{S}^{\circ}$.
Finding some initial $\theta>\theta_c$ then allows for one-sided
extrapolation towards $\theta_c$.

If we employ the 2-norm for the constraints in
\eqref{eq:kernelproblem} the face function is concave as a function of
$\theta>\theta_0$.  To see this, note that the optimal value of
problem \eqref{eq:kernelproblem} is equal to the optimal value of
\begin{equation}
\widetilde{f}(\phi,\chi)=\frac{(\phi\otimes\chi)^\dagger\Omega
(\phi\otimes\chi)}{(\phi^\dagger\phi)(\chi^\dagger\chi)}\,.
\end{equation}
Let $p^*(\Omega)$ denote the optimal value.  Since $p^*(\Omega)$
is defined as a point-wise minimum $p^*(\Omega)$ is a concave
function of $\Omega$~\cite{BV}.  Projection of $\Omega$ on some face is
an affine operation, an operation preserving concavity.  Accordingly
the face function is a concave function once the new local minimum
exists. This fact makes either approach for locating the boundary of a
face simpler, since the qualitative form of the function whose root to
find is known.

%%%%%%%%%%%%%%%%%%%%%%%%%%%%%%%%%%%%%%%%%%%%%%%%%%%%%%%%%%%%%%%%%%%%%%%%%%%%%%%%%%%%%%%%%%%%%%%%%%%%%%%
%%%%%%%%%%%%%%%%%%%%%%%%%%%%%%%%%%%%%%%%%%%%%%%%%%%%%%%%%%%%%%%%%%%%%%%%%%%%%%%%%%%%%%%%%%%%%%%%%%%%%%%


\begin{thebibliography}{10}

\bibitem{hhhh}
R.~Horodecki, P.~Horodecki, M.~Horodecki, and K.~Horodecki, Quantum
  entanglement, {\em Rev. Mod. Phys.} \textbf{81} (2009) 865.

\bibitem{NC}
M.~A. Nielsen and I.~L. Chuang, {\em Quantum Computation and Quantum
  Information} (Cambridge University Press, 2000).

\bibitem{epr1935}
A.~Einstein, B.~Podolsky, and N.~Rosen, Can quantum-mechanical description of
  physical reality be considered complete? {\em Phys. Rev.} \textbf{47}
  (1935) 777.

\bibitem{CHSH1969}
J.~F. Clauser, M.~A. Horne, A.~Shimony, and R.~A. Holt, Proposed experiment
  to test local hidden-variable theories, {\em Phys. Rev. Lett.} \textbf{23} (1969) 880.

\bibitem{Wer1989}
R.~F. Werner, Quantum states with {E}instein--{P}odolsky--{R}osen
  correlations admitting a hidden-variable model, {\em Phys. Rev. A}
  \textbf{40} (1989) 4277.

\bibitem{Mas2008}
L.~Masanes, Y.-C. Liang, and A.~C. Doherty, All bipartite entangled states
  display some hidden nonlocality, {\em Phys. Rev. Lett.} \textbf{100} (2008) 090403.

\bibitem{schro1935}
E.~Schr\"{o}dinger, Die gegenw\"{a}rtige situation in der quantenmechanik,
  {\em Naturwissenschaften}, \textbf{23} (1935) 807,823,844.

\bibitem{HorPR1994}
R.~Horodecki and P.~Horodecki, Quantum redundancies and local realism, {\em
  Phys. Lett. A} \textbf{194} (1994) 147.

\bibitem{Peres1996}
A.~Peres, Separability criterion for density matrices, {\em Phys. Rev.
  Lett.} \textbf{77} (1996) 1413.

\bibitem{BZ}
I.~Bengtsson and K.~\.{Z}yczkowski, {\em Geometry of Quantum States} (Cambridge University Press, 2006).

\bibitem{HorMPR96}
M.~Horodecki, P.~Horodecki, and R.~Horodecki, Separability of mixed states:
  necessary and sufficient conditions, {\em Phys. Lett. A} \textbf{223} (1996) 1.

\bibitem{Zyc}
K.~\ifmmode~\dot{Z}\else \.{Z}\fi{}yczkowski, P.~Horodecki, A.~Sanpera, and
  M.~Lewenstein, Volume of the set of separable states, {\em Phys. Rev. A} \textbf{58} (1998) 883.

\bibitem{Gur2002}
L.~Gurvits and H.~Barnum, Largest separable balls around the maximally mixed
  bipartite quantum state, {\em Phys. Rev. A} \textbf{66} (2002) 062311.

\bibitem{Sto1963}
E.~St{\o}rmer, Positive linear maps of operator algebras, {\em Acta
  Mathematica} \textbf{66}{110} (1963) 233.

\bibitem{poslinmaps2013}
E.~St{\o}rmer, {\em Positive {L}inear {M}aps of {O}perator {A}lgebras}
(Springer, 2013).

\bibitem{Alfsen2010}
E.~Alfsen and F.~Shultz, Unique decompositions, faces, and automorphisms of
  separable states, {\em J. Math. Phys.} \textbf{66}{51} (2010) 052201.

\bibitem{pptI}
J.~M. Leinaas, J.~Myrheim, and E.~Ovrum, Geometrical aspects of
  entanglement, {\em Phys. Rev. A} \textbf{66}{74} (2006) 012313.

\bibitem{ppt55}
L.~O. Hansen, A.~Hauge, J.~Myrheim, and P.~{\O}. Sollid, Low rank positive
  partial transpose states and their relation to product vectors, {\em Phys.
  Rev. A} \textbf{66}{85} (2012) 022309.

\bibitem{HorMPR98}
M.~Horodecki, P.~Horodecki, and R.~Horodecki, Mixed-{S}tate {E}ntanglement
  and {D}istillation: {I}s there a ''{B}ound'' {E}ntanglement in {N}ature?
  {\em Phys. Rev. Lett.} \textbf{66}{80} (1998) 5239.

\bibitem{Terhal2000}
B.~M. Terhal, Bell inequalities and the separability criterion, {\em Phys.
  Lett. A} \textbf{66}{271} (2000) 319.

\bibitem{Choi1975}
M.-D. Choi, Positive {S}emidefinite {B}iquadratic {F}orms, {\em Linear Alg.
  Appl.} \textbf{66}{12} (1975) 95.

\bibitem{ChoiLam1977}
M.-D. Choi and T.-Y. Lam, Extremal {P}ositive {S}emidefinite {F}orms, {\em
  Math. Ann.} \textbf{66}{231} (1977) 1.

\bibitem{Robertson1985}
A.~{Guyan Robertson}, Positive {P}rojections on ${C}^*$-algebras and an
  {E}xtremal {P}ositive {M}ap, {\em Journal of the London Mathematical
  Society} \textbf{2} (1985) 133.

\bibitem{Grab}
J.~Grabowski, M.~Ku\'{s}, and G.~Marmo, Wigner's theorem and the geometry of
  extreme positive maps, {\em Journal of Physics A: Mathematical and
  Theoretical} \textbf{42} (2009) 345301.

\bibitem{Mar2008}
M.~Marciniak, On extremal positive maps acting between type 1 factors,
{\em Banach Center Publ.} \textbf{89} (2010) 201.

\bibitem{Lew}
M.~Lewenstein, B.~Kraus, J.~I. Cirac, and P.~Horodecki, Optimization of
  entanglement witnesses, {\em Phys. Rev. A} \textbf{62} (2000) 052310.

\bibitem{Ha2012a}
K.-C. Ha and S.-H. Kye, Entanglement witnesses arising from {C}hoi type
  positive linear maps, {\em J. Phys. A: Math. Theor.} \textbf{45} (2012) 415305.

\bibitem{Chru3}
G.~Sarbicki and D.~Chru{\'s}ci{\'n}ski, A class of exposed indecomposable
  positive maps, {\em J. Phys. A: Math. Theor.} \textbf{46} (2013) 015306.

\bibitem{Lew2001}
M.~Lewenstein, B.~Kraus, P.~Horodecki, and J.~I. Cirac, Characterization of
  separable states and entanglement witnesses, {\em Phys. Rev. A} \textbf{63} (2001)
044304.

\bibitem{PHor2003}
P.~Horodecki, From limits of quantum operations to multicopy entanglement
  witnesses and state-spectrum estimation, {\em Phys. Rev. A} \textbf{68} (2003)
052101.

\bibitem{HorEk}
P.~Horodecki and A.~Ekert, Method for direct detection of quantum
  entanglement, {\em Phys. Rev. Lett.} \textbf{89} (2002) 127902.

\bibitem{Kor}
J.~K. Korbicz, M.~L. Almeida, J.~Bae, M.~Lewenstein, and A.~Ac\'{i}n,
  Structural approximations to positive maps and entanglement-breaking
  channels, {\em Phys. Rev. A} \textbf{78} (2008) 062105.

\bibitem{Ha2012}
K.-C. Ha and S.-H. Kye, The structural physical approximations and optimal
  entanglement witnesses, {\em J. Math. Phys.} \textbf{53} (2012) 102204.

\bibitem{pptII}
J.~M. Leinaas, J.~Myrheim, and E.~Ovrum, Extreme points of the set of density
  matrices with positive partial transpose, {\em Phys. Rev. A} \textbf{76} (2007) 034304.

\bibitem{pptIII}
J.~M. Leinaas, J.~Myrheim, and P.~{\O}. Sollid, Numerical studies of
  entangled positive-partial-transpose states in composite quantum systems,
  {\em Phys. Rev. A} \textbf{81} (2010) 062329.

\bibitem{pptIV}
J.~M. Leinaas, J.~Myrheim, and P.~{\O}. Sollid, Low rank extremal
  positive-partial-transpose states and unextendible product bases, {\em
  Phys. Rev. A} \textbf{81} (2010) 062330.

\bibitem{pptV}
P.~{\O}. Sollid, J.~M. Leinaas, and J.~Myrheim, Unextendible product bases
  and extremal density matrices with positive partial transpose, {\em Phys.
  Rev. A}, vol.~{84}, p.~042325, 2011.

\bibitem{Hauge2011}
A.~Hauge, A geometrical and computational study of entanglement witnesses,
  Master's thesis, Norwegian University of Science and Technology, Dept.\ of
  Physics, March 2011.

\bibitem{HorP1997}
P.~Horodecki, Separability criterion and inseparable mixed states with
  positive partial transposition, {\em Phys. Lett. A} \textbf{232} (1997) 333.

\bibitem{Car}
C.~Carath\'{e}odory, \"{U}ber den {V}ariabilit\"{a}tsbereich der
  {F}ourierschen {K}onstanten von positiven harmonischen {F}unktionen, {\em
  Rendiconti del Circolo Matematico di Palermo} \textbf{32} (1911) 193.

\bibitem{Stormer2013}
E.~St{\o}rmer, Separable states and the structural physical approximation of
  a positive map, {\em J. Funct. Anal.} \textbf{264} (2013) 2197.

\bibitem{BV}
S.~Boyd and L.~Vandenberghe, {\em Convex Optimization}, (Cambridge University Press, 2004).

\bibitem{NW}
J.~Nocedal and S.~J. Wright, {\em Numerical Optimization}, (Springer, 2006).

\bibitem{Brandao2005}
F.~G. S.~L. Brand\~ao, Quantifying entanglement with witness operators,
  {\em Phys. Rev. A} \textbf{72} (2005) 022310.

\bibitem{Bre2006}
H.-P. Breuer, Optimal entanglement criterion for mixed quantum states, {\em
  Phys. Rev. Lett.} \textbf{97} (2006) 080501.

\bibitem{Doh2004}
A.~C. Doherty, P.~A. Parrilo, and F.~M. Spedalieri, Complete family of
  separability criteria, {\em Phys. Rev. A} \textbf{69} (2004) 022308.

\bibitem{Ioa2006}
L.~M. Ioannou and B.~C. Travaglione, Quantum separability and entanglement
  detection via entanglement-witness search and global optimization, {\em
  Phys. Rev. A} \textbf{73} (2006) 052314.

\bibitem{San2001}
A.~Sanpera, D.~Bru\ss{}, and M.~Lewenstein, Schmidt-number witnesses and
  bound entanglement, {\em Phys. Rev. A} \textbf{63} (2001) 050301.

\bibitem{Toth2005}
G.~T\'oth and O.~G\"uhne, Detecting genuine multipartite entanglement with
  two local measurements, {\em Phys. Rev. Lett.} \textbf{94} (2005) 060501.

\bibitem{Alfsen2012}
E.~Alfsen and F.~Shultz, Finding decompositions of a class of separable
  states, {\em Lin. Alg. Appl.} \textbf{437} (2012) 2613.

\bibitem{Jam1972}
A.~Jamio{\l}kowski, Linear transformations which preserve trace and positive
  semidefiniteness of operators, {\em Reports on mathematical physics},
  \textbf{3} (1972) 275.

\bibitem{Aug}
R.~Augusiak, J.~Bae, L.~Czekaj, and M.~Lewenstein, On structural physical
  approximations and entanglement breaking maps, {\em J. Phys. A: Math.
  Theor.} \textbf{44} (2011) 185308.

\bibitem{Choetal1992}
S.-J. Cho, S.-H. Kye, and S.~G. Lee, Generalized {C}hoi maps in
  three-dimensional matrix algebra, {\em Lin. Alg. Appl.} \textbf{171} (1992) 213.

\bibitem{Ha2011}
K.-C. Ha and S.-H. Kye, Entanglement {W}itnesses {A}rising from {E}xposed
  {P}ositive {L}inear {M}aps, {\em Open Systems \& Information Dynamics}, 
\textbf{18} (2011) 323.

\bibitem{Zwolak2012}
J.~P. Zwolak and D.~Chru{\'s}ci{\'n}ski, New tools for investigating positive
  maps in matrix algebras, {\em Rep. Math. Phys} \textbf{71} (2013) 163.

\bibitem{Chru2}
D.~Chru{\'s}ci{\'n}ski and J.~Pytel, Constructing optimal entanglement witnesses.
{II}. {W}itnessing entanglement in
$4{N}\otimes4{N}$ systems, {\em Phys. Rev. A} \textbf{82} (2010) 052310.

\bibitem{Chru4}
D.~Chru{\'s}ci{\'n}ski and G.~Sarbicki, Exposed positive maps: a sufficient
  condition, {\em J. Phys. A: Math. Theor.} \textbf{45} (2012) 115304.

\bibitem{Kye2013}
S.-H. Kye, Facial structures for various notions of positivity and
  applications to the theory of entanglement, {\em Rev. Math. Phys.} \textbf{25} (2013) 1330002.

\bibitem{Ha2012b}
K.-C. Ha and S.-H. Kye, Optimality for indecomposable entanglement
  witnesses, {\em Phys. Rev. A}, \textbf{86} (2012) 034301.

\bibitem{Aug2011A}
R.~Augusiak, J.~Tura, and M.~Lewenstein, A note on the optimality of
  decomposable entanglement witnesses and completely entangled subspaces,
  {\em J. Phys. A: Math. Theor.} \textbf{44} (2011) 212001.

\bibitem{Aug2011B}
R.~Augusiak, G.~Sarbicki, and M.~Lewenstein, Optimal decomposable witnesses
  without the spanning property, {\em Phys. Rev. A} \textbf{84} (2011) 052323.

\bibitem{Chru}
D.~Chru\ifmmode \acute{s}\else \'{s}\fi{}ci\ifmmode~\acute{n}\else
  \'{n}\fi{}ski, J.~Pytel, and G.~Sarbicki, Constructing optimal entanglement
  witnesses, {\em Phys. Rev. A} \textbf{80} (2009) 062314.

\bibitem{Chru6}
D.~Chru{\'s}ci{\'n}ski and J.~Pytel, Optimal entanglement witnesses from
  generalized reduction and {R}obertson maps, {\em J. Phys. A: Math. Theor.},
  \textbf{44} (2011) 165304.

\bibitem{Chru2014}
D.~Chru{\'s}ci{\'n}ski and G.~Sarbicki, Disproving the conjecture on the
  structural physical approximation to optimal decomposable entanglement
  witnesses, {\em J. Phys. A: Math. Theor.} \textbf{47} (2014) 195301.

\bibitem{Dahl2007}
G.~Dahl, J.~M. Leinaas, J.~Myrheim, and E.~Ovrum, A tensor product matrix
  approximation problem in quantum physics, {\em Lin. Alg. Appl.},
  \textbf{420} (2007) 711.

\bibitem{Eis2004}
J.~Eisert, P.~Hyllus, O.~G\"uhne, and M.~Curty, Complete hierarchies of
  efficient approximations to problems in entanglement theory, {\em Phys.
  Rev. A} \textbf{70} (2004) 062316.

\bibitem{NR}
W.~H. Press, S.~A. Teukolsky, W.~T. Vetterling, and B.~P. Flannery, {\em
  Numerical Recipes} (Cambridge University Press, 2007).

\bibitem{matlab}
{MATLAB} {R}2011b {O}ptimization {T}oolbox {D}ocumentation.\\
{\tt http://www.mathworks.se/help/toolbox/optim/index.html}.

\bibitem{mathematica}
Mathematica 8 {O}ptimization {D}ocumentation.\\
{\tt http://reference.wolfram.com/mathematica/guide/Optimization.html}.

\bibitem{Garberg2011}
{\O}.~S. Garberg, Numerical optimization of expectation values in product
states, project thesis, Norwegian University of Science and Technology,
Dept. of Physics, December 2011.

\end{thebibliography}
\end{document}